%% file: Main.tex
\newif\ifnotes
\newif\ifcr
\notesfalse
\crfalse

\ifnotes
\newcommand{\nir}[1]{$\ll$\textsf{\color{orange} Nir: { #1}}$\gg$}
\newcommand{\omri}[1]{$\ll$\textsf{\color{blue} Omri: { #1}}$\gg$}

\else
\newcommand{\omri}[1]{}
\newcommand{\nir}[1]{}
\fi

\documentclass[a4paper,11pt]{article}
\usepackage{footmisc} 
\usepackage[hypertexnames=false,colorlinks=true,citecolor=Maroon]{hyperref}
\usepackage{newtxtext}
\usepackage{amsthm}
\usepackage{fullpage}
\usepackage{amsmath,amssymb,enumerate,algorithm,algorithmic,latexsym}
\usepackage{tikz}
\usetikzlibrary{fpu}
\usepackage{breakcites}
\usepackage{float}
\newfloat{algorithm}{H}{lop}
\usepackage{color}
\usepackage{colortbl}
\definecolor{Maroon}{cmyk}{0, 0.87, 0.68, 0.32}
\usepackage{url}
\usepackage{graphicx}
\usepackage{xspace}
\usepackage{subfigure}
\usepackage{amsfonts}
\usepackage{caption}
\usepackage{enumitem}
\usepackage{soul}
\usepackage[normalem]{ulem}
\usepackage{url}
\usepackage{graphicx}
\usepackage{bookmark}
\numberwithin{algorithm}{section}

\renewcommand{\paragraph}[1]{\vspace{1.5mm}\noindent \textbf{#1}}

\newcommand{\poly}{\mathrm{poly}}

\newcommand{\fhe}{\mathsf{FHE}}

\newcommand{\fheE}{\fhe.\mathsf{Enc}}
\newcommand{\fheD}{\fhe.\mathsf{Dec}}

\newcommand{\fheK}{\mathsf{sk}}

\newcommand{\fheCT}{\mathsf{ct}}

\newcommand{\evciph}{\hat{\fheCT}}

\newcommand{\qhe}{\mathsf{QHE}}
\newcommand{\qheG}{\qhe.\mathsf{Keygen}}
\newcommand{\qheQE}{\qhe.\mathsf{QEnc}}
\newcommand{\qheE}{\qhe.\mathsf{Enc}}
\newcommand{\qheQD}{\qhe.\mathsf{QDec}}
\newcommand{\qheD}{\qhe.\mathsf{Dec}}
\newcommand{\qheEv}{\qhe.\mathsf{Eval}}
\newcommand{\qhepk}{\mathsf{pk}}
\newcommand{\qhesk}{\mathsf{sk}}

\newcommand{\Obf}{\mathsf{Obf}}

\newcommand{\oCC}{\mathbf{CC}}

\newcommand{\obfC}{\widetilde{\mathbf{CC}}}
\newcommand{\CC}[1]{\mathbf{CC}[#1]}
\newcommand{\ccSim}{\mathsf{Sim}}

\newcommand{\sigmaP}{\Sigma.\mathsf{\zkP}}
\newcommand{\sigmaV}{\Sigma.\mathsf{\zkV}}

\newcommand{\qsigmaP}{\Xi.\mathsf{\zkP}}
\newcommand{\qsigmaV}{\Xi.\mathsf{\zkV}}
\newcommand{\qsigmaS}{\Xi.\mathsf{\zkS}}

\newcommand{\Hyb}{\mathsf{Hyb}}

\newcommand{\sfedk}{\mathsf{dk}}
\newcommand{\sfeGen}{\mathsf{SFE.Gen}}
\newcommand{\sfeEnc}{\mathsf{SFE.Enc}}
\newcommand{\sfeD}{\mathsf{SFE.Dec}}
\newcommand{\sfeEval}{\mathsf{SFE.Eval}}
\newcommand{\ciph}{\mathsf{ct}}

\newcommand{\Cir}{{C}}
\newcommand{\HESim}{\mathsf{Sim}}

\newcommand{\SFEin}{x}
\newcommand{\SFEout}{\hat{\SFEin}}

\newcommand{\cmt}{\mathsf{cmt}}
\newcommand{\Com}{\mathsf{Com}}

\newcommand{\comS}{\mathsf{Sen}}
\newcommand{\commS}{\mathsf{Sen}^*}
\newcommand{\comR}{\mathsf{Rec}}
\newcommand{\commR}{\mathsf{Rec}^*}

\newcommand{\decom}{\mathsf{VDcom}}

\newcommand{\Ext}{\mathsf{Ext}}

\newcommand{\Disting}{\mathsf{D}^*}

\newcommand{\compP}{\mathsf{P}_{\star}}
\newcommand{\compV}{\mathsf{V}_{\star}}

\newcommand{\zkslS}{\zkS_{\mathrm{comb}}}
\newcommand{\zkslSmV}{\zkS_{\mathrm{comb}, x, \zkmV}}

\newcommand{\ket}[1]{|{#1}\rangle}

\newcommand{\TD}{\mathrm{TD}}
\newcommand{\Q}{\mathsf{Q}}
\newcommand{\Wat}{\mathsf{R}}

\newcommand{\A}{\mathsf{A}^*}

\newcommand{\prot}[2]{\ve{#1,#2}}

\newcommand{\protView}{\mathsf{VIEW}}

\newcommand{\zkP}{\mathsf{P}}

\newcommand{\zkS}{\mathsf{Sim}}

\newcommand{\zkFail}{\mathtt{Fail}}

\newcommand{\wiP}{\mathsf{WI.P}}

\newcommand{\zkmD}{\mathsf{D}}

\newcommand{\zkmP}{\zkP^*}

\newcommand{\zkV}{\mathsf{V}}
\newcommand{\zkmV}{\zkV^*}
\newcommand{\wiV}{\mathsf{WI.V}}

\newcommand{\view}{\mathsf{OUT}}

\def \SigP {\Sigma.\mathsf{P}}
\def \SigmP {\Sigma.\mathsf{P}^*}
\def \SigV {\Sigma.\mathsf{V}}
\def \SigS {\Sigma.\mathsf{S}}
\def \SigA {\alpha}
\def \SigB {\beta}
\def \SigC {\gamma}

\newcommand{\lang}{\mathcal{L}}

\newcommand{\rel}{\mathcal{R}}
\newcommand{\ins}{x}

\newcommand{\wit}{w}

\newcommand{\textabbrevstyle}[1]{\mbox{#1}}
\newcommand{\textabbrevstylebol}[1]{\mbox{\textbf{#1}}}
\newcommand{\newtextabbrev}[1]{\expandafter\newcommand\csname #1\endcsname{\textabbrevstyle{#1}\xspace}}
\newcommand{\newtextabbrevbol}[1]{\expandafter\newcommand\csname #1\endcsname{\textabbrevstylebol{#1}\xspace}}
\newcommand{\renewtextabbrevbol}[1]{\expandafter\renewcommand\csname
#1\endcsname{\textabbrevstylebol{#1}\xspace}}

\newtextabbrevbol{EXP}
\newtextabbrevbol{Dtime}
\newtextabbrev{PRG}
\newtextabbrev{PRF}
\newtextabbrev{OWF}
\newtextabbrev{PPT}
\newtextabbrev{QPT}
\newtextabbrev{EOWF}
\newtextabbrev{GEOWF}
\newtextabbrev{ZAP}
\newtextabbrevbol{Ntime}
\newtextabbrevbol{coAM}
\newtextabbrevbol{NP}
\newtextabbrevbol{MIP}
\newtextabbrevbol{NEXP}
\newtextabbrevbol{SZK}
\newtextabbrevbol{BPP}
\newtextabbrevbol{coNP}
\newtextabbrevbol{coMA}
\newtextabbrevbol{TFNP}
\newtextabbrev{NIWI}
\newtextabbrevbol{AM}
\newtextabbrev{SNARG}
\newtextabbrev{SNARK}
\newtextabbrev{coRP}
\newtextabbrev{NIUA}
\newtextabbrev{UA}
\newtextabbrev{ZK}
\newtextabbrev{WZK}
\newtextabbrev{WH}
\newtextabbrev{AI}
\newtextabbrev{ECRH}
\newtextabbrev{PIR}
\newtextabbrev{WI}
\newtextabbrev{WIPOK}
\newtextabbrev{CDS}
\newtextabbrev{POK}
\newtextabbrev{FE}
\newtextabbrev{IO}
\newtextabbrev{XIO}
\newtextabbrev{SXIO}
\newtextabbrev{SKFE}
\newtextabbrev{PKFE}
\newtextabbrev{PPRF}
\newtextabbrev{LWE}
\newtextabbrev{QLWE}
\newtextabbrev{PKE}

\newcommand{\QMA}{\textabbrevstylebol{QMA}}

\newtheorem{definition}{Definition}[section]

\newtheorem{lemma}{Lemma}[section]
\newtheorem{corollary}{Corollary}[section]
\newtheorem{theorem}{Theorem}[section]
\newtheorem{claim}{Claim}[section]

\newtheorem{proposition}{Proposition}[section]

\theoremstyle{remark}

\newtheorem{remark}{Remark}[section]

\newcommand{\figref}[1]{Figure~\protect\ref{#1}}

\newcommand{\proref}[1]{Protocol~\protect\ref{#1}}

\newenvironment{boxfig}[2]{\begin{figure}[#1]\fbox{\begin{minipage}{\linewidth}
                        \vspace{0.2em}
                        \makebox[0.025\linewidth]{}
                        \begin{minipage}{0.95\linewidth}
            {{
                        #2 }}
                        \end{minipage}
                        \vspace{0.2em}
                        \end{minipage}}}{\end{figure}}

\newcommand{\pprotocol}[4]{
\begin{boxfig}{h!}{
\begin{center}
\textbf{#1}
\end{center}
    #4
\vspace{0.2em} } \caption{\label{#3} #2}
\end{boxfig}
}

\newcommand{\protocol}[4]{
\pprotocol{#1}{#2}{#3}{#4} }

\newcommand{\negl}{\mathrm{negl}}

\newcommand{\ve}[1]{\langle #1 \rangle}

\newcommand{\set}[1]{\left\{#1\right\}}
\newcommand{\abs}[1]{\left|#1\right|}

\newcommand{\pST}{\; \middle\vert \;}

\newcommand{\zo}{\{0,1\}}

\newcommand{\Nat}{\mathbb{N}}
\newcommand{\secp}{\lambda}

\newcommand{\cupdot}{\mathbin{\mathaccent\cdot\cup}}

\title{Post-quantum Zero Knowledge in Constant Rounds\footnote{This work was supported in part by ISF grants 18/484 and 19/2137, by Len Blavatnik and the Blavatnik Family Foundation, and by the European Union Horizon 2020 Research and Innovation Program via ERC Project REACT (Grant 756482).}}

\author{Nir Bitansky\thanks{Tel Aviv University, \texttt{nirbitan@tau.ac.il}. Member of the Check Point Institute of Information Security. Supported also by the Alon Young Faculty Fellowship.} \and Omri Shmueli\thanks{Tel Aviv University, \texttt{omrishmueli@mail.tau.ac.il}. Supported also by the Zevulun Hammer Scholarship from the Council for Higher Education in Israel.}}
\date{\vspace{-5ex}}
\date{}

\begin{document}
\maketitle

\input{abstract}

\thispagestyle{empty}
\newpage
\tableofcontents
\thispagestyle{empty}
\newpage
\pagenumbering{arabic}

\input{introduction}

\input{technical_overview}

\input{preliminaries}

\input{constant_round_quantum_zk}

\input{quantum_extractable_commitments}

\input{constant_round_qma}

\subsubsection*{Acknowledgments}
We thank Zvika Brakerski for insightful discussions about Quantum Fully-Homomorphic Encryption.
We also thank Venkata Koppula for advice regarding the state of the art of Compute-and-Compare Obfuscation.

\bibliographystyle{alpha}
\bibliography{Bibliography}	

\end{document}

%% file: abstract.tex
\begin{abstract}
	We construct a constant-round zero-knowledge classical argument for \NP secure against quantum attacks. We assume the existence of Quantum Fully-Homomorphic Encryption and other standard primitives, known based on the Learning with Errors Assumption for quantum algorithms. As a corollary, we also obtain a constant-round zero-knowledge quantum argument for \QMA.

At the heart of our protocol is a new {\em no-cloning} non-black-box simulation technique.

\omri{Why did we take out the "first" in the NP and QMA protocols?}\nir{For submissions (Qcrypt/QIP) we can put "the first" for eprint/arxiv/CR no need --- it has no role but to quickly sell.}
\end{abstract}

%% file: introduction.tex
\section{Introduction}
Zero-knowledge protocols allow to prove statements without revealing anything but the mere fact that they are true. Since their introduction by Goldwasser, Micali, and Rackoff \cite{GoldwasserMR89} they have had a profound impact on modern cryptography and theoretical computer science at large. Following more than three decades of exploration, zero-knowledge protocols are now quite well understood in terms of their expressiveness and round complexity. In particular, under standard computational assumptions, arbitrary $\NP$ statements can be proved in only a constant number of rounds \cite{goldreich1986prove,goldreich1996construct}.

In this work, we consider classical zero-knowledge protocols with {\em post-quantum security}, namely, protocols that can be executed by classical parties, but where both soundness and zero knowledge are guaranteed even against efficient quantum adversaries. Here our understanding is far more restricted than in the classical setting. Indeed, not only are we faced with stronger adversaries, but also have to deal with the fact that quantum information behaves in a fundamentally different way than classical information, which summons new challenges in the design of zero-knowledge protocols.

In his seminal work \cite{watrous2009zero}, Watrous developed a new quantum simulation technique and used it to show that classical zero-knowledge protocols for $\NP$, such as the Goldreich-Micali-Wigderson $3$-coloring protocol \cite{goldreich1986prove}, are also zero knowledge against quantum verifiers, assuming commitments with post-quantum hiding. These protocols are, in fact, {\em proof systems} meaning that soundness holds against unbounded adversarial provers, let alone efficient quantum ones. As in the classical setting, to guarantee a negligible soundness error (the gold standard in cryptography) these protocols require a polynomial number of rounds.

Watrous' technique does not apply for classical constant-round protocols. In fact, constant-round zero-knowledge protocols with post-quantum security remains an open question, {\em even when the honest parties and communication are allowed to be quantum.} The gap between classical and quantum zero knowledge stems from fundamental aspects of quantum information such as the no-cloning theorem \cite{Wootters:1982zz} and quantum state disturbance \cite{FuchsPeres96}. These pose a substantial barrier for classical zero-knowledge simulation techniques, a barrier that has so far been circumvented only in specific settings (such as, \cite{watrous2009zero}). Overcoming these barriers in the context of constant-round zero-knowledge seems to require a new set of techniques.

\subsection{Results}\label{subsec:results}

Under standard computational assumptions, we resolve the above open question --- we construct a classical, post-quantumly secure, computational-zero-knowledge argument for \NP in a constant number of rounds (with a negligible soundness error). That is, the honest verifier and prover (given a witness) are efficient classical algorithms. In terms of security, both zero-knowledge and soundness hold against polynomial-size quantum circuits with non-uniform quantum advice.

Our construction is based on fully-homomorphic encryption supporting the evaluation of quantum circuits (QFHE) as well as additional standard classical cryptographic primitives. All are required to be secure against efficient quantum algorithms with non-uniform quantum advice. QFHE was recently constructed \cite{mahadev2018classical, brakerski2018quantum} based on the assumption that the Learning with Errors Problem \cite{Regev09} is hard for the above class of algorithms (from hereon, called \QLWE) and a circular security assumption (analogous to the assumptions required for multi-key FHE in the classical setting). All other required primitives can be based on the \QLWE assumption.

\begin{theorem}[informal]\label{intro_thm}
	Assuming \QLWE and QFHE, there exist a classical, post-quantumly secure, computational-zero-knowledge argument in a constant number of rounds for any $\lang \in \NP$.
\end{theorem}

Combining our zero-knowledge protocol with previous work by Broadbent et al. \cite{broadbent2016zero, broadbent2019zero}, yields constant-round zero-knowledge arguments for \QMA with quantum honest parties.

\begin{corollary}[informal]
	Assuming \QLWE and QFHE, there exist a quantum, post-quantumly secure, computational-zero-knowledge argument in a constant number of rounds for any $\lang \in \QMA$.
\end{corollary}

\paragraph{Main Technical Contribution: Non-Black-Box Quantum Extraction.} Our main technical contribution is a new technique for extracting information from quantum circuits in a constant number of rounds. The technique circumvents the quantum information barriers previously mentioned. A key feature that enables this is using the adversary's circuit representation in a non-black-box manner.

The technique, in particular, yields a constant round extractable commitment. In such a commitment protocol, the verifier can commit to a classical (polynomially long) string. This commitment is perfectly binding, and hiding against efficient quantum receivers. Furthermore, it guarantees the existence of a simulator, which given non-black-box access to the sender's code, can simulate its view while extracting the committed plaintext. Further details are given in the technical overview below.

%% file: technical_overview.tex
\subsection{Technical Overview}

We next discuss the main challenges in the design of post-quantum zero knowledge in constant rounds, and our main technical ideas toward overcoming these challenges.

\subsubsection{Classical Protocols and the Quantum Barrier} \label{sec:quantum_barrier}

To understand the challenges behind post-quantum zero knowledge, let us first recall how classical constant-round protocols work, and identify why they fail in the quantum setting. Classical constant-round protocols typically involve three main steps: (1) a prover commitment $\alpha$ to a set of bits, (2) a verifier challenge $\beta$, and (3) a prover response $\gamma$, in which it opens the commitments corresponding to the challenge $\beta$. For instance, in the 3-coloring protocol of \cite{goldreich1986prove}, the prover commits to the (randomly permuted) vertex colors, the verifier picks some challenge edge, and the prover opens the commitments corresponding to the vertices of that edge. To guarantee a negligible soundness error, this is repeated in parallel a polynomial number of times.

As describe so far, the protocol satisfies a rather weak zero-knowledge guarantee --- a simulator can efficiently simulate the verifier's view in the protocol {\em if it knows the verifier's challenge $\beta$ ahead of time.} To obtain an actual zero-knowledge protocol, we need to exhibit a simulator for any {\em malicious} verifier, including ones who may arbitrarily choose their challenge depending on the prover's message $\alpha$. For this purpose, an initial step (0) is added where the verifier commits ahead of time to its challenge, later opening it in step (2) \cite{goldreich1996construct}.

The added step allows the simulator to obtain the verifier's challenges ahead of time by means of {\em rewinding}. Specifically, having obtained the verifier commitment, the simulator takes a snapshot of the verifier's state and then runs it twice: first it generates a bogus prover commitment, and obtains the verifier challenge, then with the challenges at hand, it returns to the snapshot (effectively rewinding the verifier) and runs the verifier again to generate the simulated execution. The binding of the verifier's commitment guarantees that it will never use a different challenge, and thus simulation succeeds. 

\paragraph{Barriers to Post-Quantum Security.} By appropriately instantiating the verifier commitment, the above protocol can be shown  to be sound against unbounded provers, and in particular efficient quantum provers. One could expect that by instantiating the prover's commitments so to guarantee hiding against quantum adversaries, we would get post-quantum zero knowledge. However, we do not know how to prove that such a protocol is zero knowledge against quantum verifiers. Indeed, the simulation strategy described above fails due to two basic concepts of quantum information theory:

\begin{itemize}
	\item {\bf No Cloning:} General quantum states cannot be copied. In particular, the simulator cannot take a snapshot of the verifier's state.
	
	\item {\bf Quantum State Disturbance:} General quantum circuits, which in particular perform measurements, are not reversible. Once the simulator evaluates the verifier's quantum circuit to obtain its challenge, the verifier's original state (prior to this bogus execution) has already been disturbed and cannot be recovered.
\end{itemize}

Watrous \cite{watrous2009zero} showed that in certain settings the rewinding barrier can be circumvented. He presents a \emph{quantum rewinding lemma} that roughly, shows how {\em non-rewinding} simulators that succeed in simulating only with some noticeable probability can be amplified into full-fledged simulators. The quantum rewinding lemma allows proving that classical protocols, like the GMW protocol are post-quantum zero knowledge (assuming commitments with hiding against quantum adversaries). The technique is insufficient, however, to prove post-quantum zero knowledge of existing constant-round protocols {\em with a negligible soundness error}, such as the GK protocol described above. For such protocols, non-rewinding simulators with a noticeable success probability are not known.

\paragraph{Can Non-Black-Box Techniques Cross the Quantum Barriers?} Rewinding is, in fact, often an issue {\em also in the classical setting.} Starting with the work of Goldreich and Krawczyk \cite{GoldreichK96}, it was shown that constant-round zero-knowledge protocols with certain features, such as a public-coin verifier, cannot be obtained using simulators that only use the verifier's next message function as a black box. That is, simulators that are based solely on rewinding. Surprisingly, Barak \cite{Barak01} showed that these barriers can be circumvented using {\em non-black-box techniques}. He constructed a constant-round public-coin zero-knowledge protocol where the simulator takes advantage of the explicit circuit representation of the verifier. Following Barak's work, different non-black-box techniques have been introduced to solve various problems in cryptography (c.f., \cite{DengGS09,ChungLP13,Goyal13, BitanskyP15b,ChungPS16}).

A natural question is whether we can leverage classical protocols with non-black-box simulators, such as Barak's, in order to circumvent the discussed barriers in the quantum setting. Trying to answer this question reveals several challenges. One inherent challenge is that classical non-black-box techniques naturally involve cryptographic tools that support classical computations. Obtaining zero knowledge against quantum verifiers would require analogous tools for quantum computations. As an example, Barak relies on the existence of constant-round succinct proof systems for the correctness of classical computations; to obtain post-quantum zero knowledge, such a protocol would need to support also quantum computations, while (honest) verification should remain classical. Existing protocols for classical verification of quantum computations \cite{Mahadev18} are neither constant round nor succinct.

Another family of non-black-box techniques \cite{BitanskyP15b,bitansky2019weak}, different from that of Barak, is based on fully-homomorphic encryption. Here (as mentioned above) constructions for homomorphic evaluation of quantum computations exist \cite{mahadev2018classical,brakerski2018quantum}. The problem is that the mentioned non-black-box techniques {\em do perform state cloning}. Roughly speaking, starting from the same state, they evaluate the verifier's computation (at least) twice: once homomorphically, under the encryption, and once in the clear.\footnote{In fact, Barak's technique also seems to require state cloning. Roughly speaking, the same verifier state is used once for simulating the main verifier execution and once when computing the proof for the verifier's computation.} An additional hurdle is proving soundness against quantum provers. Known non-black-box techniques are sound against efficient classical provers, and often use tools that are not known in the quantum setting, such as constant-round knowledge extraction (which is further discussed below).

Our main technical contribution is devising a non-black-box technique that copes with the above challenges. We next explain the main ideas behind the technique.

\subsubsection{Our Technique: A No-Cloning Extraction Procedure}

Toward describing the technique, we restrict attention to a more specific problem. Specifically, constructing a constant-round post-quantum zero-knowledge protocol can be reduced to the problem of constructing constant-round {\em quantumly-extractable commitments}. We recall what such commitments are and why they are sufficient, and then move to discuss the commitments we construct.

A quantumly-extractable commitment is a classical protocol between a sender $\comS$ and a receiver $\comR$. The protocol satisfies the standard (statistical) binding and post-quantum hiding, along with a plaintext extraction guarantee. Extraction requires that there exists an efficient quantum simulator $\Ext$ that given any malicious sender $\commS$, represented by a polynomial-size quantum circuit, can simulate the view of $\commS$ in the commitment protocol while extracting the committed plaintext message. Specifically, $\Ext(\commS)$ outputs a classical transcript $\widetilde T$, a quantum state $\ket{\widetilde \psi}$, and an extracted plaintext $\widetilde m$ that are computationally indistinguishable from a real transcript, state, and plaintext $(T,\ket{\psi},m)$, where $T$ and $\ket{\psi}$ are the transcript and sender state generated at the end of a real interaction between the receiver $\comR$ and sender $\commS$, and $m$ is the plaintext fixed by the commitment transcript $T$.

Such commitments allow enhancing the classical four-step protocol described before to satisfy post-quantum zero-knowledge. We simply instantiate the verifier's commitment to the challenge $\beta$ in step (0) with a quantumly-extractable commitment. To simulate a malicious quantum verifier $\zkmV$, the zero-knowledge simulator can then invoke the commitment simulator $\Ext(\zkmV)$, with $\zkmV$ acting as the sender, to obtain a simulated commitment as well as the corresponding challenge $\beta$. Now the simulator knows the challenge ahead of time, before producing the prover message $\alpha$ in step (1), and using the (simulated) verifier state $\ket{\widetilde \psi}$, can complete the simulation, {\em without any state cloning.} (Proving soundness is actually tricky on its own due to malleability concerns. We remain focused on zero knowledge for now).

The challenge is of course to obtain constant-round commitments with {\em no-cloning extraction}. Indeed, classically-extractable commitments have been long known in constant rounds under minimal assumptions, based on rewinding (and thus state cloning) \cite{PrabhakaranRS02}. We next describe our non-black-box technique and how it enables quantum extraction without state cloning.

\paragraph{The Non-Black-Box Quantum Extraction Technique: A Simple Case.} To describe the technique, we first focus on a restricted class of adversarial senders that are {\em non-aborting and explainable}. The notion of non-aborting explainable senders considers senders $\commS$ whose messages can always be {\em explained} as a behavior of the honest (classical) sender with respect to {\em some} plaintext and randomness (finding this explanation may be inefficient); in particular, they never abort. The notion further restricts that of {\em (aborting) explainable adversaries} from \cite{bitansky2019weak}, which also allows aborts. To even further simplify our exposition, we first address classical (rather than quantum) senders, but crucially, while avoiding any form of state cloning. Later on, we shall address general quantum adversaries.

Our protocol is inspired by \cite{BitanskyP15b, bitansky2019weak} and relies on two basic tools. The first is fully-homomorphic encryption (FHE) --- an encryption scheme that allows to homomorphically apply any polynomial-size circuit $C$ to an encryption of $x$ to obtain a new encryption of $C(x)$, proportional in size to the result $|C(x)|$ (the size requirement is known as {\em compactness}). The second is {\em compute-and-compare program obfuscation} (CCO).  A compute-and-compare program $\CC{f,s,z}$ is given by a function $f$  (represented as a circuit), a target string $s$ in its range, and a message $z$; it outputs $z$ on every input $x$ such that $f(x)=s$, and rejects all other inputs. A corresponding obfuscator compiles any such program into a program $\obfC$ with the same functionality. In terms of security, provided that the target $s$ has high entropy conditioned on $f$ and $z$, the obfuscated program is computationally indistinguishable from a simulated dummy program, independent of $(f,s,z)$. Such post-quantumly-secure obfuscators are known under QLWE \cite{goyal2017lockable,wichs2017obfuscating,goyal2019perfect}.

\medskip\noindent
To commit to a message $m$, the protocol consists of three steps:
\begin{enumerate}
	\item
	The sender $\comS$ samples:
	\begin{itemize}
		\item two random strings $s$ and $t$, \omri{Maybe we should change it to $t$ and $s$?}\nir{yeah, make it consistent w/ the body (note that $u,v$ also appear in the paragraph above)}
		\item a secret key $\fheK$ for an FHE scheme,
		\item an FHE encryption $\fheCT_t = \fheE_{\fheK}(t)$ of $t$,
		\item an obfuscation $\obfC$ of $\CC{f,s,z}$, where $z = (m,\fheK)$ and $f = \fheD_{\fheK}$ is the FHE decryption circuit.
	\end{itemize}     It then sends $(\fheCT_t,\obfC)$ to the receiver $R$.
	
	\item
	The receiver $\comR$ sends a guess $t'$.
	
	\item
	$\comS$ rewards a successful guess: if $t = t'$, it sends back $s$ (and otherwise $\bot$).
\end{enumerate}
The described commitment protocol comes close to our objective. First, it is binding --- the obfuscation $\obfC$ uniquely determines $z=(m,\fheK)$. Second, it is hiding --- a receiver (even if malicious) gains no information about the message $m$. To see this, we argue that no receiver sends $t'=t$ at the second message, but with negligible probability. Indeed, given only the first sender message $(\fheCT_t,\obfC)$, the receiver obtains no information about $s$. Hence, we can invoke the CCO security and replace the obfuscation $\obfC$ with a simulated one, which is independent of the secret FHE key $\fheK$. This, in turn, allows us to invoke the security of encryption to argue that the first message $(\fheCT_t,\obfC)$ hides $t$. It follows that the third sender message is $\bot$ (rather than the target $s$) with overwhelming probability, which again by CCO security implies that the entire view of the receiver can be simulated independently of $m$.

Lastly, a non-black-box simulator, given the circuit representation of an explainable sender $\commS$, can simulate the sender's view, while extracting $m$. It first runs the sender to obtain the first message $(\fheCT_t,\obfC)$. At this point, it can use the sender's circuit $\commS$ to continue the emulation of $\commS$ {\em homomorphically under the encryption} $\fheCT_t$. The key point is that, under the encryption, we do have $t$. We can (homomorphically) feed $t$ to the sender, and obtain an encryption $\fheCT_s$ of $s$. Now, the simulator feeds $\fheCT_s$ to the obfuscation $\obfC$, and gets back $z = (m,\fheK)$. (Note that here the compactness of FHE is crucial --- the sender $\commS$ could be of arbitrary polynomial size, whereas $\obfC$ and thus also $\fheCT_s$ are of fixed size.)

Having extracted $m$, it remains to simulate the inner (for now, classical) state $\psi$ of the sender $S^*$ and the full interaction transcript $T$. These are actually available, but in encrypted form, as a result of the previous homomorphic computation. Here we use the fact that the extracted $z$ also includes the decryption key $\fheK$, allowing us to obtain the state $\psi$ and transcript $T$ {\em in the clear}.

An essential difference between the above extraction procedure and previous non-black-box extraction techniques (e.g., \cite{BitanskyP15b,bitansky2019weak}) is that {\em it does not perform any state cloning}. As explained earlier, previous procedures would perform the same computation twice, once under the encryption, and once in the clear. Here we perform the computation once, partially in the clear, and partially homomorphically. Crucially, we have a mechanism to peel off the encryption at the end of second part so that we do not have to redo the computation in the clear.

\paragraph{Indistinguishability through Secure Function Evaluation.} The described protocol does not quite achieve our objective. The simulated interaction is, in fact, easy to distinguish from a real one. Indeed, in a simulated interaction the simulator's guess in the second message is $t'=t$, whereas the receiver cannot produce this value. To cope with this problem, we augment the protocol yet again, and perform the second step under a {\em secure function evaluation} (SFE) protocol. This can be thought of as homomorphic encryption with an additional {\em circuit privacy} guarantee, which says that the result of homomorphic evaluation of a circuit, reveals nothing about the evaluated circuit to the decryptor, except of course from the result of evaluation.

\newpage\noindent
The augmented protocol is similar to the previous one, except for the last two steps, now done using SFE:
\begin{enumerate}
	\item
	{\color{gray} The sender $\comS$ samples:
		\begin{itemize}
			\item two random strings $s$ and $t$,
			\item a secret key $\fheK$ for an FHE scheme,
			\item an FHE encryption $\fheCT_t = \fheE_{\fheK}(t)$ of $t$,
			\item an obfuscation $\obfC$ of $\CC{f,s,z}$, where $z = (m,\fheK)$ and $f = \fheD_{\fheK}$ is the FHE decryption circuit.
		\end{itemize}     It then sends $(\fheCT_t,\obfC)$ to the receiver $\comR$.}
	
	\item
	The receiver $\comR$ sends $\fheCT'_{t'}$, a guess $t'$ encrypted using SFE. (The honest receiver sets $t'$ arbitrarily.)
	
	\item
	$\comS$ homomorphically evaluates the function that given input $t$, returns $s$ (and otherwise $\bot$). $\comS$ then returns the resulting ciphertext to $\comR$.
\end{enumerate}
\omri{Thinking about it now, I think that it might be beneficial to slightly clarify what exactly the simulator is. What about the following or something of that sort:
	"Given $\fheCT_t$, the simulator in such protocol first homormophically evaluates the SFE encryption circuit, then the circuit of the sender, and then the SFE decryption circuit.}\nir{Right. added text below. It's  more intuitive and simpler to think about emulating the entire protocol under the encryption --- no need to break it to different circuits --- see also relevant comments in the body (ZK section)}
The homomorphic computation done by the simulator in the new protocol is augmented accordingly --- instead of sending $t$ and obtaining $s$ directly, it now sends an SFE encryption of $t$ and obtains back an SFE encryption of $s$, which it can then decrypt to obtain $s$. Thus, as before, the homomorphic computation results in an FHE encryption of $s$. Indistinguishability of the simulated sender view from the real sender view now follows since the SFE encryption $\fheCT'_{t'}$ hides $t'$. The SFE circuit privacy guarantees that the homomorphic SFE evaluation does not leak any information about the target $s$, as long as the receiver does not send an SFE encryption of $t$.

\paragraph{A Malleability Problem and its Resolution.} While we could argue before that a malicious receiver cannot output $t$ in the clear, arguing that it does not output an SFE encryption of $t$ is more tricky. In particular, the receiver might be able to somehow maul the FHE encryption $\fheCT_t$ to get an SFE encryption $\fheCT'_{t}$ of the value $t$, without actually ``knowing'' the value $t$. Classically, such malleability problems are solved using {\em extraction}. If we could efficiently extract the value encrypted in the SFE encryption $\fheCT'$, then we could rely on the previous argument. However, as explained before, efficient extraction is classically achieved using rewinding and thus state cloning. While so far we have focused on avoiding state cloning for the sake of simulating the sender, we should also avoid state cloning when proving hiding of the commitment as we are dealing with quantum receivers. It seems like we are back to square one.

To circumvent the problem, we rely on the fact that the hiding requirement of the commitment is relatively modest --- commitments to different plaintexts should be indistinguishable. This is in contrast the efficient simulation requirement for the sender (needed for efficient zero knowledge simulation). Here one commonly used solution is {\em complexity leveraging} --- we can design the SFE, FHE, and CCO so that extraction from SFE encryptions can be done in brute force, without any state cloning, and without compromising the security of the FHE and CCO. This comes at the cost of assuming subexponential (rather than just polynomial) hardness of the primitives in use.

A different solution, which is also the one we use in the body of the paper, relies on hardness against efficient quantum adversaries with {\em non-uniform quantum advice} (instead of subexponential hardness). Specifically, the receiver sends a commitment to the SFE encryption key in the beginning of the protocol. The reduction establishing the hiding of the protocol gets as non-uniform advice the initial receiver (quantum) state that maximizes the probability of breaking hiding, along with the corresponding SFE key. This allows for easy extraction from SFE encryptions, without any state cloning.

The full solution contains additional steps meant to establish that the receiver's messages are appropriately structured (e.g., the receiver's commitment defines a valid SFE key, and the SFE encryption later indeed uses that key). This is done using standard techniques based on witness-indistinguishable proofs, which exist in a constant number of rounds \cite{goldreich1986prove} assuming commitments with post-quantum hiding (and in particular, QLWE).

\paragraph{Dealing with Quantum Adversaries.} Above, we have assumed for simplicity that the sender is classical and have shown a simulation strategy that requires no state cloning. We now explain how the protocol is augmented to deal with quantum senders (for now still restricting attention to non-aborting explainable senders). The first natural requirement in order to deal with quantum senders is that the cryptographic tools in use (e.g., SFE encryption) will be postqantum secure. This can be guaranteed assuming QLWE.

As already mentioned earlier in the introduction, post-quantum security alone is not enough --- we need to make sure that our non-black-box extraction technique can also work with quantum, rather than classical, circuits representing the sender $\commS$. For this purpose, we use {\em quantum} fully-homomorphic encryption (QFHE). In a QFHE scheme, the encryption and decryption keys are (classical) strings and the encryption and decryption algorithms are classical provided that the plaintext is classical (and otherwise quantum). Most importantly, QFHE allows to homomorphically evaluate quantum circuits. Such QFHE schemes were recently constructed in \cite{mahadev2018classical, brakerski2018quantum} based on \QLWE and a circular security assumption (analogous to the assumptions required for multi-key FHE in the classical setting).

The augmented protocol simply replaces the FHE scheme with a QFHE scheme (other primitives, such as the SFE and compute-and-compare are completely classical in terms of functionality and only need to be post-quantum secure). In the augmented protocol, the honest sender and receiver still act classically. In contrast, the non-black-box simulator described before is now quantum --- it homomorphically evaluates the quantum sender circuit $\commS$. A technical point is that QFHE should support the evaluation of a quantum circuit with an additional quantum auxiliary input --- in our case the quantum sender $\commS$ and its inner state after it sends the first message. This is achieved by existing QFHE schemes (for instance, by using their public key encryption mode, and encrypting the initial state prior to the computation).

\paragraph{Dealing with Aborts.} So far, we have dealt with explainable senders that are non-aborting. This is indeed a strong restriction and in fact, quantumly-extractable commitments against this class of senders can be achieved using black-box techniques (see more in the related work section). However, considering an adversary who, with noticeable probability, may abort at some stage of the protocol, existing black-box techniques completely fail (even if the adversary is explainable up to the abort). In contrast, as we shall see, our non-black-box technique will enable simulation also for aborting senders.

In our protocol, an aborting sender $\commS$ may refuse to perform the SFE evaluation in the last step of the protocol. In this case, the simulator will get stuck --- the simulated transcript and sender state $\ket{\psi}$ will remain forever locked under the encryption (since the simulator cannot use the obfuscation $\obfC$ to get the decryption key $\fheK$). Accordingly, the described simulator successfully simulates senders that never abort, but fails to simulate senders that abort (noticeably often). We next observe that there is, in fact, a non-rewinding simulation strategy also for the other extreme, namely for senders $\commS$ that (almost) always abort. Here the simulator would simply send {\em in the clear} (rather than under FHE) an SFE encryption $\fheCT'_{t'}$ of an arbitrary string $t'$, just like the honest receiver $\comR$. In this case, the simulated sender view is identical to its view in a real interaction (and since the sender $\commS$ aborts, there is no need to extract the plaintext message).

We show that the two simulators described, $\zkS_{\mathrm{na}}$ for never-aborting senders and $\zkS_{\mathrm{aa}}$ for always-aborting senders, can be combined into a simulator for general senders (which sometimes abort). This is enabled by the fact that simulated receiver messages $\fheCT'_{t'}$ generated by the two simulators are indistinguishable due to the hiding of SFE encryptions. Accordingly, the sender's choice of whether to abort or not is (computationally) independent of whether we are simulating using the first simulator $\zkS_{\mathrm{na}}$ or the second $\zkS_{\mathrm{aa}}$. This gives rise to a combined simulator $\zkS_{\mathrm{comb}}$, which flips a random coin $b\gets \set{\mathrm{na},\mathrm{aa}}$ to predict whether an abort will occur, and then runs $\zkS_{b}$. The combined simulator $\zkS_{\mathrm{comb}}$ succeeds if it guessed correctly, which occurs with probability (negligibly close to) half.

\omri{Changed the following paragraph. Previous version in comment in code.}
\nir{looks good, slightly edited}
\paragraph{Applying Watrous' Quantum Rewinding Lemma.} The above is reminiscent of the simulation strategy in classical 3-message zero-knowledge protocols (with a large soundness error), such as the GMW 3-coloring protocol \cite{GoldreichMW87}. In these protocols, for each possible verifier challenge $\beta$ there exists a non-rewinding simulator $\zkS_\beta$, and the combined simulator $\zkS_{\mathrm{comb}}$ tries to guess the challenge $\beta$ and apply the corresponding simulator.
Similarly to the combined simulator in our protocol, the verifier's choice of challenge $\beta$ is (computationally) independent of $\zkS_{\mathrm{comb}}$'s guess, and thus the simulator $\zkS_{\mathrm{comb}}$ succeeds in simulating with some fixed noticeable probability (specifically $2^{-|\beta|}$).

The advantage of such simulators (non-rewinding and successful with fixed noticeable probability) is that they can be amplified to full-fledged simulators, both classically and quantumly. In the classical setting, a full-fledged simulator $\zkS$ can be obtained by rerunning $\zkS_{\mathrm{comb}}$ until it succeeds. We can, in fact, apply the same rerunning strategy also for quantum verifiers. However, this does not guarantee zero knowledge against verifiers with quantum auxiliary input (since each execution of $\zkS_{\mathrm{comb}}$ may disturb the verifier's auxiliary state). To obtain zero knowledge against verifiers with quantum auxiliary input, we apply Watrous' quantum rewinding lemma \cite{watrous2009zero}, which shows how to faithfully amplify the combined simulator $\zkS_{\mathrm{comb}}$ in the presence of quantum auxiliary input.

\omri{Nir, what do you think about a figure (here) showing the resemblance between GMW and our protocol? btw, do we want to include any comments about the possibility that this simulator is relevant for other themes in ZK? CANCELLED}
\nir{it's a very simple resemblance so not sure what a picture can add, but I don't object. I wouldn't write anything about relevance to other themes, since we have nothing concrete to say here.}

%

\paragraph{From Explainable Adversaries to Malicious Ones.} The only remaining gap is the assumption that senders are explainable; that is, the messages they send (up to the point that they possibly abort), can always be explained as messages that would be sent by the honest (classical) sender for some plaintext and randomness. The simulator $\zkS_{\mathrm{na}}$ (for never-aborting verifiers) crucially relies on this; in particular, the CCO $\obfC$ and the FHE ciphertext $\fheCT_t$ must be formed consistently with each other for the simulator to work. Importantly, it suffices that {\em there exists an explanation} for the messages, and we do not have to efficiently extract it as part of the simulation;\footnote{This is in contrast to other restrictions of the adversary considered in the literature, like semi-honest and semi-malicious adversaries \cite{GoldreichMW87, HaitnerIKLP11,BoyleG0KS13}.} indeed, efficient quantum extraction is exactly the problem we are trying to solve.

The commitment protocol against explainable senders naturally gives rise to a zero-knowledge protocol against explainable verifiers. As is often the case in the design of zero knowledge protocols (see discussion in \cite{bitansky2019weak}), dealing with explainable verifiers is actually the hard part of designing zero-knowledge protocols. Indeed, we use a generic transformation of \cite{bitansky2019weak}, slightly adapted to our setting, which converts zero-knowledge protocols against explainable verifiers to ones against arbitrary malicious verifiers. The transformation is based on constant-round (post-quantumly-secure) witness-indistinguishable proofs, which as mentioned before can be obtained based on QLWE.

\subsection{More Related Work on Post-Quantum Zero Knowledge}
The study of post-quantum zero-knowledge (QZK) protocols was initiated by van de Graaf \cite{van1997towards}, who first observed that traditional zero-knowledge simulation techniques, based on rewinding, fail against quantum verifiers. Subsequent work has further explored different flavors of zero knowledge and their limitations \cite{watrous2002limits}, and also demonstrated that relaxed notions such as zero-knowledge with a trusted common reference string can be achieved \cite{Kobayashi03,damgaard2004zero}. Later on, Peikert and Shiehian constructed non-interactive post-quantum zero knowledge from QLWE in the common random string model \cite{peikert2019noninteractive}. Watrous \cite{watrous2009zero} was the first to show that the barriers of quantum information theory can be crossed, demonstrating a post-quantum zero-knowledge protocol for \NP in a polynomial number of rounds (in the plain model).
\omri{DFS show computational QZK proofs and perfect QZK arguments for NP, but with a CRS. Kobayashi shows qNIZK for NP, but assumes entanglement between parties (not only CRS), and also that they both are quantum.}

\paragraph{Zero Knowledge for \QMA.} Another line of work aims at constructing quantum (rather than classical) protocols for $\QMA$ (rather than \NP). Following a sequence of works \cite{ben2006secure,liu2006consistency,dupuis2010secure, dupuis2012actively,morimae2015quantum}, Broadbent, Ji, Song and Watrous \cite{broadbent2016zero} show a zero-knowledge quantum proof system for all of \QMA~(in a polynomial number of rounds).

\paragraph{Quantum Proofs and Arguments of Knowledge.} Extracting knowledge from quantum adversaries was investigated in a sequence of works \cite{unruh2012quantum,hallgren2011classical,lunemann2011fully,ambainis2014quantum}. A line of works considered different variants of quantum proofs and arguments of knowledge (of the witness), proving both feasibility results and limitations. In particular, Unruh  \cite{unruh2012quantum} shows that assuming post-quantum injective one-way functions, some existing systems are a quantum proof of knowledge. He identifies a certain {\em strict soundness} requirement that suffices for such an implication. Ambainis, Rosmanis and Unruh \cite{ambainis2014quantum} give evidence that this requirement may be necessary.

Based on QLWE, Hallgren, Smith, and Song \cite{hallgren2011classical} and Lunemann and Nielsen \cite{lunemann2011fully} show argument of knowledge where it is also possible to simulate the prover's state (akin to our simulation requirement of the sender's state). Unruh further explores arguments of knowledge in the context of computationally binding quantum commitments \cite{unruh2016computationally, unruh2016collapse}. All of the above require a polynomial number of rounds to achieve a negligible knowledge error.

\paragraph{Zero-Knowledge Multi-Prover Interactive Proofs.} Two recent works by Chiesa et al. \cite{ChiesaFGS18} and by Grilo, Slofstra, and Yuen \cite{GriloSY19} show that \NEXP and $\MIP^*$, respectively, have {\em perfect} zero-knowledge multi-prover interactive proofs (against entangled quantum provers).

\paragraph{Concurrent Work.}
Broadbent and Grilo \cite{broadbent2019zero} construct quantum sigma protocols for QMA, that is, 3-message protocols that are zero-knowledge but have large soundness error. Relying on their protocol and our zero-knowledge protocol and extractable commitment, we obtain a conceptually simple constant-round zero-knowledge protocol for QMA with a negligible soundness error (in a previous version of our work, we constructed such a protocol based on earlier work of \cite{broadbent2016zero}). Coladangelo, Vidick, and Zhang \cite{coladangelo2019non} construct non-interactive zero-knowledge arguments with preprocessing for QMA in the common reference string model. The challenges tackled and corresponding techniques in our work are substantially different than those in both of the above mentioned works.
%
\omri{Nir, note the above references to BG19 and CVZ19.}\nir{edited the above a bit}

Ananth and La Placa \cite{Priv_Comm} developed a non-black-box quantum extraction protocol that share some of our ideas and is based on similar computational assumptions. They used it to obtain quantum zero-knowledge, but only against explainable non-aborting verifiers.


\paragraph{A Word on Strict Commitments and Non-Aborting Verifiers.} In \cite{unruh2012quantum}, Unruh introduces a notion of {\em strict commitments}, which are commitments that fix not only the plaintext, but also the randomness (e.g. Blum-Micali \cite{BlumM84}), and are known to exist based on injective one-way functions. As mentioned in our technical overview, using such commitments it is possible to obtain zero-knowledge in constant rounds {\em against non-aborting explainable verifiers} through the GK four-step template we discussed in the overview. Roughly speaking, this is because when considering verifiers that always open their (strict) commitments, we are assured that measuring their answer does not disturb the verifier state, as this answer is information-theoretically fixed. This effectively allows to perform rewinding.

%% file: preliminaries.tex
\section{Preliminaries}\label{sec:prel}

\noindent We rely on standard notions of classical Turing machines and Boolean circuits:

\begin{itemize}
\item
A \PPT algorithm is a probabilistic polynomial-time Turing machine.

\item
We sometimes think about \PPT algorithms as polynomial-size uniform families of circuits, these are equivalent models.
A polynomial-size circuit family $\mathcal{C}$ is a sequence of circuits $\mathcal{C} = \set{C_\secp}_{\secp\in\Nat}$, such that each circuit $C_\secp$ is of polynomial size $\secp^{O(1)}$.
We say that the family is uniform if there exists a deterministic polynomial-time algorithm $M$ that on input $1^\secp$ outputs $C_\secp$.

\nir{3: Why do we have this item?}\omri{We sometimes think on the same algorithm in two different ways i.e. the SFE encryption and decryption algorithms are \PPT in definition (and should be, I don't want to think of cryptographic algorithms as circuits) but in the simulation we think about their circuit implementations as we evaluate both homomorphically. Same for the QFHE decryption algorithm, we put the decryption circuit inside the CC obfuscation program. Let's talk about this if you still feel the item is redundant.}\nir{You don't need this, it only makes things cumbersome. If you insist just have an item saying any $T$-time turing machine can be converted into a circuit of size $\tilde O(T)$, but this is well known...}

\item
For a \PPT algorithm $M$, we denote by $M(x;r)$ the output of $M$ on input $x$ and random coins $r$. For such an algorithm and any input $x$, we write $m\in M(x)$ to denote the fact that $m$ is in the support of $M(x;\cdot)$.
\end{itemize}

\noindent We follow standard notions from quantum computation.
\begin{itemize}
	\item
	A \QPT algorithm is a quantum polynomial-time Turing machine.
	
	\item
	We sometimes think about \QPT algorithms as polynomial-size uniform families of quantum circuits, these are equivalent models.
	A polynomial-size quantum circuit family $\mathcal{C}$ is a sequence of quantum circuits $\mathcal{C} = \set{C_\secp}_{\secp\in\Nat}$, such that each circuit $C_\secp$ is of polynomial size $\secp^{O(1)}$.
	We say that the family is uniform if there exists a deterministic polynomial-time algorithm $M$ that on input $1^\secp$ outputs $C_\secp$.
	
	\nir{3:again, why do we need this?}\omri{From the same reason. We think of the quantum algorithm of the simulator as \QPT and then generate its circuit and give it to the Watrous rewinding algorithm.}\nir{Again you don't need this.}
	
\item
	An interactive algorithm $M$, in a two-party setting, has input divided into two registers and output divided into two registers.
	For the input, one register $I_m$ is for an input message from the other party, and a second register $I_a$ is an auxiliary input that acts as an inner state of the party.
	For the output, one register $O_m$ is for a message to be sent to the other party, and another register $O_a$ is again for auxiliary output that acts again as an inner state. For a quantum interactive algorithm $M$, both input and output registers are quantum.
\end{itemize}

%
%

\paragraph{The Adversarial Model.}
Throughout, efficient adversaries are modeled as quantum circuits with non-uniform quantum advice (i.e. quantum auxiliary input).
Formally, {\em a polynomial-size adversary} $\A = \set{\A_\secp, \rho_\secp}_{\secp\in\Nat}$, consists of a polynomial-size non-uniform sequence of quantum circuits $\{ \A_\secp \}_{\secp \in \Nat}$, and a \nir{non-uniform?} sequence of polynomial-size mixed quantum states $\{ \rho_\secp \}_{\secp \in \Nat}$.

For an interactive quantum adversary in a classical protocol, it can be assumed without the\nir{remove "the"} loss of generality that its output message register (the register containing the message to be sent to the other side, not the register containing output quantum auxiliary information\nir{parenthesis more confusing than helpful}) is always measured in the computational basis at the end of computation. This assumption is indeed without the loss of generality, because whenever a quantum state is sent through a classical channel then qubits decohere and are effectively measured in the computational basis.

\paragraph{Indistinguishability in the Quantum Setting.}

\begin{itemize}
	\item
	Let $f:\Nat \rightarrow [0, 1]$ be a function.
	\begin{itemize}
		\item
		$f$ is negligible if for every constant $c \in \Nat$ there exists $N \in \Nat$ such that for all $n > N$, $f(n) < n^{-c}$.
		
		\item
		$f$ is noticeable if there exists $c \in \Nat, N \in \Nat$ such that for every $n \geq N$, $f(n) \geq n^{-c}$.

		\item
		$f$ is overwhelming if it is in\nir{in $=>$ of} the form $1 - \mu(n)$, for a negligible function $\mu$.
	\end{itemize}

	\item
	We may consider random variables over bit strings or over quantum states. This will be clear from the context. 

	\item
	For two random variables $X$ and $Y$ supported on quantum states, quantum distinguisher circuit $\zkmD$ with, quantum auxiliary input $\rho$, and $\mu \in [0, 1]$, we write $X \approx_{\zkmD, \rho, \mu} Y$ if
	\begin{align*}
	\abs{
		\Pr[\zkmD(X; \rho)=1] - \Pr[\zkmD(Y; \rho)=1]
	} \leq \mu.
	\end{align*}
	
	\item
	Two ensembles of random variables $\mathcal{X}=\{X_{i}\}_{\secp\in \Nat, i \in I_\secp}$, $\mathcal{Y}=\{Y_{i}\}_{\secp\in \Nat, i \in I_\secp}$ over the same set of indices $I = \cupdot_{\secp \in \Nat}I_\secp$ are said to be {\em computationally indistinguishable}, denoted by $\mathcal{X}\approx_{c} \mathcal{Y}$, if for every polynomial-size quantum distinguisher $\zkmD=\set{\zkmD_\secp, \rho_\secp}_{\secp\in\Nat}$ there exists a negligible function $\mu(\cdot)$ such that for all $\secp \in \Nat, i \in I_\secp$,
	\begin{align*}
	X_i \approx_{\zkmD_{\secp}, \rho_\secp, \mu(\secp)} Y_i\enspace.
	\end{align*}

	\item The trace distance between two distributions $X, Y$ supported over quantum states, denoted $\TD(X, Y)$, is a generalization of statistical distance to the quantum setting and represents the maximal distinguishing advantage between two distributions supported over quantum states, by unbounded quantum algorithms.
	We thus say that ensembles $\mathcal{X}=\{X_{i}\}_{\secp\in \Nat, i \in I_\secp}$, $\mathcal{Y}=\{Y_{i}\}_{\secp\in \Nat, i \in I_\secp}$, supported over quantum states, are statistically indistinguishable (and write $\mathcal{X}\approx_{s} \mathcal{Y}$), if there exists a negligible function $\mu(\cdot)$ such that for all $\secp \in \Nat, i \in I_\secp$,
	\begin{align*}
	\TD\left( X_i, Y_i \right) \leq \mu(\secp) \enspace.
	\end{align*}
\end{itemize}

\medskip
In what follows, we introduce the cryptographic tools used in this work.
By default, all algorithms are classical and efficient unless stated otherwise, and security holds against polynomial-size non-uniform quantum adversaries with quantum advice.

\subsection{Interactive Protocols, Witness Indistinguishability, and Zero Knowledge}
We define proof and argument systems that are secure against quantum adversaries.
We start with classical protocols and proceed to define quantum protocols.
In what follows, we denote by $(\zkP, \zkV)$ a protocol between two parties $\zkP$ and $\zkV$.
For common input $\ins$, we denote by $\view_{\zkV}\prot{\zkP}{\zkV}(\ins)$ the output of $\zkV$ in the protocol. For honest verifiers, this output will be a single bit indicating acceptance or rejection of the proof.
Malicious quantum verifiers may have arbitrary quantum output (which is formally captured by the verifier outputting its inner quantum state).

\begin{definition}[Classical Proof and Argument Systems for NP] \label{def:proofs_args}
	Let $(\zkP, \zkV)$ be a protocol with an honest \PPT prover $\zkP$ and an honest \PPT verifier $\zkV$ for a language $\lang \in\NP$, satisfying:
	\begin{enumerate}
		
		\item {\bf Perfect Completeness:}
		For any $\secp \in\Nat,\ins \in \lang \cap \zo^\secp, w \in \mathcal{R}_{\lang}(x)$,
		$$
		\Pr[ \view_{\zkV}\prot{\zkP(w)}{\zkV}(\ins) = 1 ]  = 1 \enspace.
		$$
		
		\item {\bf Soundness:}
		The protocol satisfies one of the following.
		\begin{itemize}
			\item {\bf Computational Soundness:} For any quantum polynomial-size prover $\zkmP = \set{\zkmP_\secp, \rho_\secp}_{\secp \in \Nat}$, there exists a negligible function $\mu(\cdot)$ such that for any security parameter $\secp\in \Nat$ and any $\ins \in \zo^\secp\setminus\lang$,
			\begin{align*}
			\Pr\left[ \view_{\zkV}\prot{\zkmP_\secp(\rho_\secp)}{\zkV}(\ins) = 1 \right] \leq \mu(\secp)\enspace.
			\end{align*}
			A protocol with computational soundness is called an argument.
			
			\item {\bf Statistical Soundness:} There exists a negligible function $\mu(\cdot)$, such that for any (unbounded) prover $\zkmP$, any security parameter $\secp\in \Nat$, and any $\ins \in \zo^\secp\setminus\lang$,
			\begin{align*}
			\Pr\left[ \view_{\zkV}\prot{\zkmP}{\zkV}(\ins)=1 \right] \leq \mu(\secp)\enspace.
			\end{align*}
			A protocol with statistical soundness is called a proof.
		\end{itemize}
	\end{enumerate}
\end{definition}

\begin{definition}[Quantum Proof and Argument Systems for QMA] \label{def:proofs_args_quantum}
	Let $(\zkP, \zkV)$ be a quantum protocol with an honest \QPT prover $\zkP$ and an honest \QPT verifier $\zkV$ for a language $\lang \in\QMA$, satisfying:
	\begin{enumerate}
		
		\item {\bf Statistical Completeness:}
		There is a polynomial $k(\cdot)$ and a negligible function $\mu(\cdot)$ s.t. for any $\secp \in\Nat$,$\ins \in \lang \cap \zo^\secp$, $w \in \mathcal{R}_{\lang}(x)$\footnote{For a language $\lang$ in QMA, for an instance $x \in \lang$ in the language, the set $\mathcal{R}_{\lang}(x)$ is the (possily infinite) set of quantum witnesses that make the BQP verification machine accept with some overwhelming probability $1 - \negl(\secp)$.},
		$$
		\Pr[ \view_{\zkV}\prot{\zkP(w^{\otimes k(\secp)})}{\zkV}(\ins) = 1 ] \geq 1 - \mu(\secp) \enspace.
		$$
		
		\item {\bf Soundness:}
		As in Definition \ref{def:proofs_args}.
	\end{enumerate}
\end{definition}

\subsubsection{Witness Indistinguishability}
We rely on classical constant-round (public-coin) proof systems for NP that are witness-indistinguishable; that is, proofs that use different witnesses (for the same statement) are computationally indistinguishable (for quantum attackers).

\begin{definition} [WI Proof System for NP] \label{def:arg}
A classical protocol proof system $(\zkP, \zkV)$ for a language $\lang \in\NP$ (as in Definition \ref{def:proofs_args}) is witness-indistinguishable if it satisfies:

\paragraph{Witness Indistinguishability:}
For every quantum polynomial-size verifier $\zkmV = \{ \zkmV_\secp,\rho_\secp \}_\secp$,$$
\{\view_{\zkmV_\secp}\prot{\zkP(w_0)}{\zkmV_\secp(\rho_\secp)}(\ins)\}_{\secp,\ins,\wit_0,\wit_1}
\approx_{c}
\{\view_{\zkmV_\secp}\prot{\zkP(w_1)}{\zkmV_\secp(\rho_\secp)}(\ins)\}_{\secp,\ins,\wit_0,\wit_1}
\enspace,
$$
where $\secp\in \Nat, \ins\in \lang\cap\zo^\secp$, and $\wit_0,\wit_1\in\rel_{\lang}(\ins)$ are witnesses for $\ins$.
\end{definition}

\paragraph{Instantiations.}
3-message, public-coin classical proof systems with WI follow from classical zero-knowledge proof systems such as the parallel repetition of the 3-coloring protocol \cite{goldreich1991proofs}, which is in turn based on non-interactive perfectly-binding commitments.
For the proof system to be WI against quantum attacks, we need the non-interactive commitments to be computationally hiding against quantum adversaries, which can be instantiated for example from QLWE.

\subsubsection{Sigma Protocols}
We use the abstraction of {\em Sigma Protocols}, which are public-coin three-message proof systems with a special zero knowledge guarantee.
We define both classical and quantum Sigma Protocols.

\begin{definition} [Classical Sigma Protocol for NP]
A classical sigma protocol for $\lang \in \NP$ is a classical proof system $(\sigmaP, \sigmaV)$ (as in Definition \ref{def:proofs_args}) with 3 messages and the following syntax.
\begin{itemize}
	\item
	$(\alpha, \tau) \gets \sigmaP_1(x,w):$ Given an instance $x \in \lang$ and a witness $w \in \rel_{\lang}(\ins)$, the first prover execution outputs a public message $\alpha$ for $\sigmaV$ and a private inner state $\tau$.
	
	\item
	$\beta \gets \sigmaV(x):$ The verifier simply outputs a string of $\poly(|x|)$ random bits.
	
	\item
	$\gamma \gets \sigmaP_3(\beta, \tau):$ Given the verifier's string $\beta$ and the private state $\tau$, the prover outputs a response $\gamma$.
\end{itemize}
\nir{2: say explicitly what are $P_1,P_3$. Also simplify notation: $(\alpha, \tau) \gets P_1(x,w)$ (better not to use $r$ for state) $\gamma \gets P_3(\beta;\tau)$.}\omri{ok?}
The protocol satisfies the following.

\paragraph{Special Zero-Knowledge:} There exists a \PPT simulator $\SigS$ such that,
$$
\set{(\alpha,\gamma) \; | \; (\alpha, \tau) \gets \SigP_1(\ins,\wit), \gamma \gets \SigP_3(\beta, \tau) }_{\secp, x, w, \beta}
\approx_{c}
\set{(\SigA,\SigC) \; | \; (\SigA,\SigC)\gets \SigS(\ins,\SigB)}_{\secp, x, w, \beta}\enspace,
$$
where $\secp\in \Nat$, $\ins\in \lang\cap\zo^\secp$, $\wit\in\rel_{\lang}(\ins)$ and $\beta \in \{ 0, 1 \}^{\poly(\secp)}$.
\end{definition}

\noindent The next claim follows directly from the special zero-knowledge requirement, and will be used throughout.
\begin{claim} [First-Message Indistinguishability, \cite{bitansky2018multi}, Claim 8.1] \label{first_message_indisting}
\nir{This claim is not needed. see zk section}
In every $\Sigma$ protocol:
$$
\set{\alpha \; | \; (\alpha, \tau) \gets \SigP_1(\ins,\wit) }_{\secp, x, w, \beta}
\approx_{c}
\set{\alpha \; | \; (\SigA,\SigC)\gets \SigS(\ins,0^{|\beta|})}_{\secp, x, w, \beta}\enspace,
$$
where $\secp\in \Nat$, $\ins\in \lang\cap\zo^\secp$, $\wit\in\rel_{\lang}(\ins)$ and $\beta \in \{ 0, 1 \}^{\poly(\secp)}$.

\nir{since you're using $\SigP_1$ also before, define it (and $\SigP_3$) in the beginning of the sigma def, then here only need to address $\SigS_1$}\omri{decided to write it differently.}
\end{claim}

\paragraph{Instantiations.}
Like witness-indistinguishable proofs, Sigma protocols are known to follow from the parallel repetition of the 3-coloring protocol \cite{goldreich1991proofs}.
For the protocol to have special zero knowledge against quantum attacks, we need the non-interactive commitment $\alpha$ to be computationally hiding against quantum adversaries, which can be instantiated for example from QLWE.

\begin{definition} [Quantum Sigma Protocol for QMA]
	A quantum sigma protocol for $\lang \in \QMA$ is a quantum proof system $(\qsigmaP, \qsigmaV)$ (as in Definition \ref{def:proofs_args_quantum}) with 3 messages and the following syntax.
	\begin{itemize}
		\item
		$(\alpha, \tau) \gets \qsigmaP_1(x,w^{\otimes k(\secp)}):$ Given an instance $x \in \lang\cap\{ 0, 1 \}^\secp$ and $k(\secp)$ witnesses $w \in \rel_{\lang}(\ins)$ (for a polynomial $k(\cdot)$), the first prover execution outputs a public message $\alpha$ for $\qsigmaV$ and a private inner state $\tau$.
		
		\item
		$\beta \gets \qsigmaV(x):$ The verifier simply outputs a string of $\poly(|x|)$ random bits.
		
		\item
		$\gamma \gets \qsigmaP_3(\beta, \tau):$ Given the verifier's string $\beta$ and the private state $\tau$, the prover outputs a response $\gamma$.
	\end{itemize}
	The protocol satisfies the following.
	
	\paragraph{Special Zero-Knowledge:} There exists a \QPT simulator $\qsigmaS$ such that,
	$$
	\set{(\alpha,\gamma) \; | \; (\alpha, \tau) \gets \qsigmaP_1(\ins,w^{\otimes k(\secp)}), \gamma \gets \qsigmaP_3(\beta, \tau) }_{\secp, x, w, \beta}
	\approx_{c}
	\set{(\SigA,\SigC) \; | \; (\SigA,\SigC)\gets \qsigmaS(\ins,\SigB)}_{\secp, x, w, \beta}\enspace,
	$$
	where $\secp\in \Nat$, $\ins\in \lang\cap\zo^\secp$, $\wit\in\rel_{\lang}(\ins)$ and $\beta \in \{ 0, 1 \}^{\poly(\secp)}$.
\end{definition}


\paragraph{Instantiations.}
Quantum sigma protocols follow from the parallel repetition of the 3-message quantum zero-knowledge protocols of \cite{broadbent2019zero} for QMA\footnote{The authors in \cite{broadbent2019zero} use the name "sigma protocols" differently then in this work. Specifically, \cite{broadbent2019zero} call their 3-message protocols, that are zero-knowledge but have large soundness error, "sigma protocols". In this work we call the parallel repetition of such protocols (which have amplified soundness but weakened zero knowledge) "sigma protocols".}.

\subsubsection{Quantum Zero-Knowledge Protocols}
We next define post-quantum zero-knowledge classical protocols and zero-knowledge quantum protocols.

\begin{definition}[Post-Quantum Zero-Knowledge Classical Protocol] \label{def:qzk}
	Let $(\zkP, \zkV)$ be a classical protocol (argument or proof) for a language $\lang \in\NP$ as in Definition \ref{def:proofs_args}.
	The protocol is quantum zero-knowledge if it satisfies:
	
	\paragraph{Quantum Zero Knowledge:} There exists a quantum polynomial-time simulator $\zkS$, such that for any quantum polynomial-size verifier $\zkmV = \set{\zkmV_\secp, \rho_\secp}_{\secp \in \Nat}$,
	$$
	\{ \view_{\zkmV_\secp}\prot{\zkP(w)}{\zkmV_\secp(\rho_{\secp})}(x)\}_{\secp, x, w}
	\approx_{c}
	\{\zkS(x,\zkmV_\secp, \rho_\secp)\}_{\secp, x, w}\enspace,
	$$
	where $\secp \in \Nat$, $x \in \lang \cap \{ 0, 1 \}^\secp$, $w \in \rel_{\lang}(x)$.
	\begin{itemize}
		\item If $\zkmV$ is a classical circuit, then the simulator is computable by a classical polynomial-time algorithm.
			
		\nir{why do we have this item?}\omri{Without this item, the definition of post-quantum ZK is incomparable to classical ZK.}
	\end{itemize}
\end{definition}

\omri{remarks about the classical simulator item, and about the simulator being able to simulate the transcript w.l.o.g., because we can always efficiently compute a verifier circuit that records the transcript inside its inner state, which the simulator simulates.}

\begin{definition}[Zero-Knowledge Quantum Protocol] \label{def:qma_qzk}
	Let $(\zkP, \zkV)$ be a quantum protocol (argument or proof) for a language $\lang \in\QMA$ as in Definition \ref{def:proofs_args_quantum}, where the prover uses $k(\secp)$ copies of a witness.
	The protocol is quantum zero-knowledge if it satisfies:
	
	\paragraph{Quantum Zero Knowledge:} There exists a quantum polynomial-time simulator $\zkS$, such that for any quantum polynomial-size verifier $\zkmV = \set{\zkmV_\secp, \rho_\secp}_{\secp \in \Nat}$,
	$$
	\{ \view_{\zkmV_\secp}\prot{\zkP(w^{\otimes k(\secp)})}{\zkmV_\secp(\rho_{\secp})}(x)\}_{\secp, x, w}
	\approx_{c}
	\{\zkS(x,\zkmV_\secp, \rho_\secp)\}_{\secp, x, w}\enspace,
	$$
	where $\secp \in \Nat$, $x \in \lang \cap \{ 0, 1 \}^\secp$, $w \in \rel_{\lang}(x)$.
\end{definition}

\subsection{Additional Tools}

\subsubsection{Compute-and-Compare Obfuscation}
We define compute-and-compare (CC) circuits and obfuscators for CC circuits.

\begin{definition}[Compute-and-Compare Circuit]
	Let $f:\zo^n\rightarrow\zo^\secp$ be a circuit, and let $u\in\zo^\secp, z\in \zo^*$ be strings. Then $\CC{f,u,z}(x)$ is a circuit that returns $z$ if $f(x)=y$, and $\bot$ otherwise. $\CC{f,u,z}$ has a canonical description from which $f$, $u$, and $z$ can be read.
\end{definition}

We now define compute-and-compare (CC) obfuscators (with perfect correctness). In what follows $\Obf$ is a \PPT algorithm that takes as input a CC circuit $\CC{f,u,z}$ and outputs a new circuit $\obfC$.

\begin{definition}[CC obfuscator]
A \PPT algorithm $\Obf$ is a compute-and-compare obfuscator if it satisfies:
\begin{enumerate}
\item
{\bf Perfect Correctness:} For any circuit $f:\zo^n\rightarrow\zo^\secp$, $u\in \zo^\secp$ and $z \in \zo^*$,
$$
\Pr\left[\forall x \in \zo^{n}:\obfC(x) = \CC{f,u,z}(x) \pST \obfC\gets \Obf(\CC{f,u,z})\right]=1 \enspace.
$$

\item
{\bf Simulation:} There exists a $\PPT$ simulator $\ccSim$ such that for every two polynomials $\ell_1(\cdot)$, $\ell_2(\cdot)$,
$$\{
\obfC \; | \; u \gets \{ 0, 1 \}^\secp , \obfC \gets \Obf(\CC{f,u,z})
\}_{\secp, f, z}
\approx_c
\{
\ccSim(1^{\ell_1(\secp)}, 1^{\ell_2(\secp)}, 1^{\secp})
\}_{\secp, f, z} \enspace ,
$$
where $\secp \in \Nat$, $f : \{ 0, 1 \}^{n} \rightarrow \{ 0, 1 \}^{\secp}$ is a $\ell_1(\secp)$-size circuit, $z \in \{ 0, 1 \}^{\ell_2(\secp)}$.

\nir{again the infinite sequence formulation is cumbersome and hard to read. Just have an ensemble indexed by $f,u,z$}\omri{now?}

\end{enumerate}

\end{definition}

\paragraph{Instantiations.} Compute-and-compare obfuscators with almost-perfect correctness are constructed in \cite{goyal2017lockable,wichs2017obfuscating} based on QLWE.
CC obfuscators with perfect correctness are constructed \cite{goyal2019perfect} by Goyal, Koppula, Vusirikala and Waters, also based on QLWE.

\subsubsection{Non-Interactive Commitments}
We define non-interactive commitment schemes. 

\begin{definition} [Non-Interactive Commitment] \label{def:com}
	A non-interactive commitment scheme is given by a \PPT algorithm $\Com(\cdot)$ with the following syntax:
	
	\begin{itemize}
		\item $\cmt \gets \Com(1^\secp, x) : $
		A randomized algorithm that takes as input a security parameter $1^\secp$ and input $x \in \{ 0, 1 \}^*$, and outputs a commitment $\cmt$.
	\end{itemize}

	The commitment algorithm satisfies:
	\begin{enumerate}
		\item {\bf Perfect Binding:}
		For any $\secp_0, \secp_1 \in \Nat$, $x_0, x_1 , r_0, r_1\in \{ 0, 1 \}^*$, $\Com(1^{\secp_0}, x_0 ; r_0) = \Com(1^{\secp_1}, x_1 ; r_1)$ implies $x_0 = x_1$.
		
		\item {\bf Computational Hiding:}
		For any polynomial $\ell(\cdot)$,
		$$
		\{ \Com(1^\secp, x_{0}) \}_{\secp, x_0, x_1}
		\approx_c
		\{ \Com(1^\secp, x_{1}) \}_{\secp, x_0, x_1} \enspace ,
		$$
		where $\secp \in \Nat$, $x_0, x_1 \in \{ 0, 1 \}^{\ell(\secp)}$.
	\end{enumerate}
\end{definition}

\paragraph{Instantiations.} The above non-interactive commitments are known based on various standard assumptions, including QLWE \cite{GoyalHKW17,LombardiS19}.

\subsubsection{Quantum Fully Homomorphic Encryption}
We rely on quantum fully homomorphic encryption, specifically, a scheme where a classical input can be encrypted classically and a quantum input quantumly.
The formal definition follows.
\begin{definition} [Quantum Fully-Homomorphic Encryption]
	A quantum fully homomorphic encryption scheme is given by six algorithms $(\qheG,$ $\qheE,$ $\qheQE,$ $\qheD,$ $\qheQD,$ $\qheEv)$ with the following syntax:
	\begin{itemize}
		\item
		$(\qhepk, \qhesk) \gets \qheG(1^\secp):$ A \PPT algorithm that given a security parameter $1^\secp$, samples a classical public key $\qhepk$ and a classical secret key $\qhesk$.
		
		\item
\nir{have macros for ciphertexts}\omri{done}
		$\ciph \gets \qheE_{\qhepk}(x):$ A \PPT algorithm that takes as input a classical string $x \in \{ 0, 1 \}^*$ and outputs a classical ciphertext $\ciph$.
		
\nir{why define bit encryption? we're encrypting strings, it would make the notation more cumbersome later.}\omri{done}

		\item
		$\ket{\phi} \gets \qheQE_{\qhepk}(\ket{\psi}):$ A \QPT algorithm that takes as input a quantum state $\ket{\psi}$ and outputs a quantum ciphertext $\ket{\phi}$.
		 		
		\item
		$x \gets \qheD_{\qhesk}(\ciph):$ A \PPT algorithm that takes as input a classical ciphertext $\ciph$ and outputs a string $x$.
		
		\item
		$\ket{\psi} \gets \qheQD_{\qhesk}(\ket{\phi}):$ A \QPT algorithm that takes as input a quantum ciphertext $\ket{\phi}$ and outputs a quantum state $\ket{\psi}$.
		
		\item
		$\ket{\hat{\phi}} \gets \qheEv_{\qhepk}(C, \ciph, \ket{\phi} ):$ A \QPT algorithm that takes as input a general quantum circuit $C$, a classical ciphertext $\ciph$ and a quantm ciphertext $\ket{\phi}$ and outputs an evaluated quantum ciphertext $\ket{\hat{\phi}}$
		
\nir{I suggest to give eval as input also explicit classical ciphertext $\mathbf{c}$ }\nir{remove all the lengths and the part about measurement --- should be addressed in the properties and not in the syntax. Just say it takes a quantum circuit C, classical CT, quantum CT, and outputs classical CT and quantum CT.}

	\end{itemize}
	The scheme satisfies the following.
	\begin{itemize}
		\item {\bf Quantum Semantic Security:} For every polynomial $\ell(\cdot)$,
		\begin{align*}
			&\left\{
			\left( \ciph, \ket{\phi} \right) \pST
			\begin{array}{l}
			(\qhepk, \qhesk) \gets \qheG(1^\secp), \\
			\ciph \gets \qheE_{\qhepk}(x_0), \\
			\ket{\phi} \gets \qheQE_{\qhepk}(\ket{\psi_0})
			\end{array}
			\right\}_{\secp, x_0, \ket{\psi_0}, x_1, \ket{\psi_1}}
			\approx_{c} \\
			&\left\{
			\left( \ciph, \ket{\phi} \right) \pST
			\begin{array}{l}
			(\qhepk, \qhesk) \gets \qheG(1^\secp), \\
			\ciph \gets \qheE_{\qhepk}(x_1), \\
			\ket{\phi} \gets \qheQE_{\qhepk}(\ket{\psi_1})
			\end{array}
			\right\}_{\secp, x_0, \ket{\psi_0}, x_1, \ket{\psi_1}} \enspace ,
		\end{align*}
		where $\secp \in \Nat$, $x_0, x_1 \in \{ 0, 1 \}^{\ell(\secp)}$ and $\ket{\psi_0}$, $\ket{\psi_1}$ are $\ell(\secp)$-qubit states.
				
		\nir{2: why don't we define this?}\omri{done}
		
		\item {\bf Compactness:}
		There exists a polynomial $\poly(\cdot)$ s.t. for every quantum circuit $C$ with $\ell$ output qubits and an enryption of an input for $C$, the output size of the evaluation algorithm is $\ell\cdot \poly(\secp)$, where $\secp$ is the security parameter of the scheme.
		\nir{the complexity of evaluation does scale with the size of $C$, the resulting ciphertext is compact.}\omri{This doesn't contradict what I wrote. I only wrote that the output size of evaluation is compact, not the computation itself, of course.}
		
		\nir{not sure that classicality is an actual word, or at least not one in use. Perhaps call it "Measurement-Preserving Homomorphism" (no need to add the word quantum)}\nir{this definition shouldn't be too asymptotic. say that for any poly $s(\secp)$ there is a negligible $\mu(\secp)$ such that for any size-$s$ circuit $C$ and state $\ket{\Psi}$ ... $TD(\rho_0,\rho_1)<\mu(\secp)$ ($pk,sk,r,\psi$ shouldn't be infinite sequences)}
		
		\item {\bf Measurement-Preserving Homomorphism:}
		For every polynomial $s(\cdot)$ there exists a neligible function $\negl(\cdot)$ such that for every $\secp \in \Nat$, size-$s(\secp)$ quantum circuit $C$, input $(x, \ket{\psi})$ for $C$ which is comprised of a classical string $x$ and quantum state $\ket{\psi}$, subset $M$ of the output qubits of $C$, public and secret key pair $(\qhepk, \qhesk)\in \qheG(1^\secp)$ and randomness strings $(r_x, r_{\ket{\psi}})$:
		$$
		\TD\left( D_0, D_1 \right) \leq \negl(\secp) \enspace ,
		$$
		where $D_0, D_1$ are the distributions which are defined as follows:
		\begin{itemize}
			\item $D_0:$
			Compute $\ket{\psi'} \gets C\left( x, \ket{\psi} \right)$, measure the subset of qubits of $\ket{\psi'}$ which are in $M$ and output the obtained state.
			
			\item $D_1:$
			\begin{itemize}
				\item Encrypt $\ciph = \qheE_{\qhepk}(x; r_x)$, $\ket{\phi} = \qheQE_{\qhepk}(\ket{\psi}; r_{\ket{\psi}})$.
				
				\item Evaluate $\ket{\hat{\phi}} \gets \qheEv_{\qhepk}(C, \ciph, \ket{\phi})$.
				
				\item Measure the $|M|$ packets of qubits that correspond to the output qubits in $M$ (by compactness, each packet is exactly of size $\poly(\secp)$).
				
				\item Decrypt the measured $|M|$ packets with $\qheD_{\qhesk}(\cdot)$, and decrypt the rest of the qubits with $\qheQD_{\qhesk}(\cdot)$. Output the obtained state.
			\end{itemize}
		\end{itemize}
%
	\end{itemize}
		
\end{definition}

%
%

\paragraph{Instantiations.} Mahadev \cite{mahadev2018classical} shows how to build quantum FHE based on super-polynomial QLWE modulus and a circular security assumption with respect to a secret key and an additional trapdoor information.
\nir{don't say slightly non-standard just say what she assumes}\omri{I don't know how to say this. I would just write "super-polynomial QLWE modulus and a non-standard circular security assumption".}
Brakerski \cite{brakerski2018quantum} subsequently shows how to construct quantum FHE based on polynomial QLWE modulus and a circular security assumption (analogous to the assumptions required for multi-key FHE in the classical setting).
\nir{don't call it standard, it's debatable. Use similar wording to intro.}\omri{ok}
\nir{It's not more general it's more specific... just say that the classical-preservation is not required in the standard definition, but is satisfied by B18...}\omri{ok}
The above definition is more specific then the standard definition of QFHE.
Specifically, \emph{measurement-preservation} and (statistical) correctness for \emph{every} triplet $(\qhepk, \qhesk, r)$ of public and secret keys and randomness $r$ for the encryption algorithm, is not an explicit part of the standard definition. The construction of Brakerski satisfies this more general definition. This follows readily from the main Theorem (4.1) in \cite{brakerski2018quantum}.


\subsubsection{Function-Hiding Secure Function Evaluation} \label{def:SFE}
We define two-message function evaluation protocols with statistical circuit privacy and quantum input privacy.
\begin{definition}[2-Message Function Hiding SFE] \label{def:1hop}
	A two-message secure function evaluation protocol $(\sfeGen,$ $\sfeEnc,$ $\sfeEval,$ $\sfeD)$ has the following syntax:
	\begin{itemize}
		\item
		$\sfedk \gets \sfeGen(1^\secp):$  a probabilistic algorithm that takes a security parameter $1^\secp$ and outputs a secret key $\sfedk$.\nir{call it sk (dk suggests decryption key, and it's not a decryption key). If you have multiple sk later, use sk,sk' etc}
		\item
		$\ciph \gets\sfeEnc_{\sfedk}(\SFEin):$ a probabilistic algorithm that takes a string $\SFEin \in \zo^*$, and outputs a ciphertext $\ciph$.
		\item
		$\evciph\gets\sfeEval(\Cir,\ciph):$ a probabilistic algorithm that takes a (classical) circuit $\Cir$ and a ciphertext $\ciph$ and outputs an evaluated ciphertext $\evciph$.
		\item
		$\SFEout = \sfeD_{\sfedk}(\evciph):$ a deterministic algorithm that takes a ciphertext $\evciph$ and outputs a string $\SFEout$.
	\end{itemize}
	
	\noindent The scheme satisfies the following.
\nir{no need to talk about a sequence of circuits here...}
	\begin{itemize}
		\item {\bf Perfect Correctness:} For any polynomial $s(\cdot)$, for any $\secp \in \Nat$, size-$s(\secp)$ circuit $C$ and input $x$ for $C$,
		$$
		\Pr\left[
		\sfeD_{\sfedk}(\evciph) = \Cir(\SFEin) \pST
		\begin{array}{l}
		\sfedk \gets \sfeGen(1^\secp),\\
		\ciph \gets \sfeEnc_{\sfedk}(\SFEin),\\
		\evciph \gets \sfeEval(\Cir,\ciph)
		\end{array}
		\right] = 1\enspace.
		$$
	
		
		\item {\bf Quantum Input Privacy:} For every polynomial $\ell(\cdot)$,
		\begin{align*}
		\left\{
		\ciph \pST
		\begin{array}{l}
		\sfedk \gets \sfeGen(1^\secp), \\
		\ciph \gets \sfeEnc_{\sfedk}(x_0)
		\end{array}
		\right\}_{\secp, x_0, x_1}
		\approx_{c}
		\left\{
		\ciph \pST
		\begin{array}{l}
		\sfedk \gets \sfeGen(1^\secp), \\
		\ciph \gets \sfeEnc_{\sfedk}(x_1)
		\end{array}
		\right\}_{\secp, x_0, x_1} \enspace ,
		\end{align*}
		where $\secp \in \Nat$ and $x_0, x_1 \in \{ 0, 1 \}^{\ell(\secp)}$.
		
\nir{perhaps use instead ensemble and computational ind notation as in previous defs. In any case drop the $\ell'$.}\omri{done}
		
		\item {\bf Statistical Circuit Privacy:} There exist unbounded algorithms, probabilistic $\HESim$ and deterministic $\Ext$ such that:
		\begin{itemize}
			\item For every $x \in \{ 0, 1 \}^*$, $\ciph \in \sfeEnc(x)$, the extractor outputs $\Ext(\ciph) = x$.
			
			\item For any polynomial $s(\cdot)$,
			\begin{align*}
			\{
			\sfeEval(\Cir, \ciph^*)
			\}_{\secp, \Cir, \ciph^*}
			\approx_s
			\{
			\HESim( \; 1^\secp, \Cir(\Ext(1^\secp, \ciph^*)) \; ) \}_{\secp, \Cir, \ciph^*}\enspace,
			\end{align*}
			where $\secp\in\Nat$, $\Cir$ is a $s(\secp)$-size circuit, and $\ciph^* \in \{ 0, 1 \}^{*}$.
		\end{itemize}
			
	\end{itemize}
\end{definition}

\noindent The next claim follows directly from the circuit privacy property, and will be used throughout the analysis.
\begin{claim}[Evaluations of Agreeing Circuits are Statistically Close] \label{claim:SFE_eval}
	\nir{if this is the only property, might as well simplify and state it as circuit privacy. Can probably throw the extractor and ask that for any two functionally equivalent $C_0,C_1$ (of the same size) statistical closeness holds. }\omri{We are using the extractor as well (in the soundness proof).}
	
	For any polynomial $s(\cdot)$,
	$$
	\{ \sfeEval(C_{0}, \ciph^*) \}_{\secp, C_0, C_1, \ciph}
	\approx_s
	\{ \sfeEval(C_{1}, \ciph^*) \}_{\secp, C_0, C_1, \ciph} \enspace ,
	$$
	where $\secp \in \Nat$, $C_0$, $C_1$ are two $s(\secp)$-size functionally-equivalent circuits, and $\ciph^* \in \{ 0, 1 \}^*$.
\end{claim}

\paragraph{Instantiations.} Such secure function evaluation schemes are known based on QLWE \cite{ostrovsky2014maliciously, brakerski2018two}.

\subsubsection{Quantum Rewinding Lemma}
\nir{change this to a more general lemma that accounts for non-purified circuits.}\nir{In the pretext, just say what the lemma means. Then state it. Only then discuss that in watrous it's stated differently and explain the difference.p}

We use Lemma 9 from \cite{watrous2009zero}, which constructs a quantum algorithm for amplifying the success probability of quantum sampler circuits under some conditions.


\begin{lemma} [Lemma 9, \cite{watrous2009zero}] \label{lem:watrous}
	There is a quantum algorithm $\Wat$ that gets as input:
	\begin{itemize}
		\item
		A general quantum circuit $\Q$ with $n$ input qubits that outputs a classical bit $b$ and an additional $m$ output qubits.
		
		\item
		An $n$-qubit state $\ket{\psi}$.
		
		\item
		A number $t \in \Nat$.
	\end{itemize}
	$\Wat$ executes in time $t\cdot \poly(|\Q|)$ and outputs a distribution over $m$-qubit states $D_{\psi} := \Wat( \Q, \ket{\psi}, t )$ with the following guarantees.
	
	For an $n$-qubit state $\ket{\psi}$, denote by $\Q_{\psi}$ the conditional distribution of the output distribution $\Q(\ket{\psi})$, conditioned on $b = 0$, and denote by $p(\psi)$ the probability that $b = 0$.
	If there exist $p_0, q \in (0, 1)$, $\varepsilon \in \left( 0, \frac{1}{2} \right)$ such that:	
	\begin{itemize}
		\item
		Amplification executes for enough time: $t \geq \frac{\log(1/\varepsilon)}{4\cdot p_0 (1 - p_0)}$,
		
		\item
		There is some minimal probability that $b = 0$:
		For every $n$-qubit state $\ket{\psi}$, $p_0 \leq p(\psi)$,
		
		\item
		$p(\psi)$ is input-independant, up to $\varepsilon$ distance: For every $n$-qubit state $\ket{\psi}$, $\abs{p(\psi) - q} < \varepsilon$, and
		
		\item
		$q$ is closer to $\frac{1}{2}$:
		$\; p_0(1 - p_0) \leq q(1 - q)$,
	\end{itemize}
	then for every $n$-qubit state $\ket{\psi}$,
	$$
	\TD\Big( \Q_{\psi}, D_{\psi} \Big) \leq 4\sqrt{\varepsilon}\frac{\log(1/\varepsilon)}{p_0 (1 - p_0)} \enspace .
	$$
\end{lemma}

The exact wording in the Lemma from \cite{watrous2009zero} differs from the above in two manners.
First, the original lemma states that for each circuit $\Q$ there exists an amplified circuit $\Q'$, but actually\nir{in actuality => actually}\omri{done} the proof of the Lemma proves that there is an algorithm $\Wat$ that on input $\Q$, executes an amplified version of $\Q$ (and thus the circuit implementation of $\Wat(\Q)$ can be thought of as $\Q'$).
Second, the original lemma deals with unitary quantum circuits i.e. $\Q$ contains no measurement gates.
By standard quantum circuit purification, it follows that the above formulation is equivalent to the analogous statement that includes only unitary circuits.

%% file: constant_round_quantum_zk.tex
\section{Constant-Round Zero-Knowledge Arguments for NP \nir{can remove "with quantum security". There's the context of the paper for that}\omri{Done.}} \label{sec:const-round_qzk}
	In this section we construct a classical argument system for an arbitrary NP language $\lang$, with a constant number of rounds, quantum soundness and quantum zero-knowledge (according to Definition \ref{def:qzk}).

	\nir{can move the ZK def to preliminaries (see relevant comments there)}	
	\nir{Just say it's $V*s$ (quantum) output (this is the most general thing --- you can in particular think of $V^*$ that outputs its view including the transcript). Also, no need to define $OUT$ separately for $\zkV$ and $\zkmV$, it's anyhow just the output, but when you define it for the first time (in the prelim) say that if it is quantum then this output is quantum etc.}
	\omri{done}
	
	\omri{At the end, make sure that the classicality preservation is noted and explained, this only requires iterating $\zkslS$ repeatedly $\secp$ times until there's a success or $\secp$ tries failed.
	Also, it should be noted that for the classical case the QHE is swapped to a regular FHE.}
	
	\paragraph{Ingredients and notation:}
	\begin{itemize}
		\item
		A non-interactive commitment scheme $\Com$.
		\item
		A CC obfuscator $\Obf$.
		\item
		A quantum fully homomorphic encryption scheme $(\qheG, \qheE, \qheQE, \qheD,\newline \qheQD, \qheEv )$.
		\item
		A 2-message function-hiding secure function evaluation scheme $(\sfeGen,$ $\sfeEnc,$ $\sfeEval,$ $\sfeD)$.
		\item
		A 3-message WI proof $(\wiP, \wiV)$ for $\lang \in \NP$.
		\item
		A 3-message sigma protocol $(\SigP, \SigV)$ for $\lang \in \NP$.
	\end{itemize}
	
	\nir{Remove this from here. Instead, add to the protocol an item: "If either party sends a message of an incorrect form or provides a non-convincing WI proof, the other party terminates the interaction".}\omri{done.}
%
	
	\noindent We describe the protocol in \figref{fig:const_round_qzk}.
	
	\protocol
	{\proref{fig:const_round_qzk}}
	{A classical constant-round zero-knowledge argument for $\lang \in \NP$ with quantum security.}
	{fig:const_round_qzk}
	{
		\begin{description}
			\item[Common Input:] An instance $\ins \in \lang\cap \zo^\secp$, for security parameter $\secp \in \Nat$.
			
			\item[$\zkP$'s private input:] A classical witness $w \in \mathcal{R}_{\lang}(\ins)$ for $\ins$.
			\nir{remove (it goes w/o saying it's poly-size, since it's an NP relation)}\omri{done}\nir{remove "showing that..." (in general, concise is easier to read. It's good to make things more explicit when  they're not clear/standard, but otherwise, you want the protocol to be concise}\omri{done}
		\end{description}
		
		\begin{enumerate}
			\item {\bf Prover Commitment:} \nir{just "Prover Commitment"}\omri{done} \label{qzk:prover_commit}
			$\zkP$ sends non-interactive commitments to the witness $\wit$ and to a string of zeros in the length of an SFE secret key $\sfedk$: $\cmt_1 \gets \Com(1^\secp, w)$, $\cmt_2 \gets \Com(1^\secp, 0^{|\sfedk|})$.
			
			\item {\bf Extractable Commitment to Verifier Challenge:} \label{qzk:verifier_commit}
			\nir{Extractable Commitment to Verifier Challenge}\omri{done}
			\nir{make this the first item in the list of steps.}\omri{done}
			\begin{enumerate}
				\item
				$\zkV$ computes a challenge $\beta \gets \SigV$.
				
				\item \label{qzk:verifier_CC}
				$\zkV$ computes $s \gets \{ 0, 1 \}^{\secp}$, $t \gets \{ 0, 1 \}^{\secp}$, $(\qhepk, \qhesk) = \qheG(1^{\secp}; r)$ where $r$ is the sampled randomness for the QFHE key generation algorithm\omri{change to $\qhesk$ in the future.}. $\zkV$ sends
				$$
				\qhepk, \enspace
				\ciph_{\zkV}\gets \qheE_{\qhepk}(t), \enspace
				\obfC \gets \Obf\Big(\oCC\big[ \qheD_{\qhesk}(\cdot), s, (r, \beta) \big]\Big)\enspace .
				$$
				
				\nir{I think you should just give the secret key and not r (see corresponding remark in the description of the simulator)}
				
				\item \label{qzk:prover_sfe}
				$\zkP$ computes $\sfedk \gets \sfeGen(1^\secp)$ and sends $\ciph_{\zkP} \gets \sfeEnc_{\sfedk}(0^\secp)$.
				
				\item \label{qzk:verifier_response}
				$\zkV$ sends $\evciph \gets \sfeEval\Big(\oCC\big[ \mathsf{Id(\cdot)}, t, s \big],\, \ciph_{\zkP} \Big)$, where $\mathsf{Id(\cdot)}$ is the identity function.
			\end{enumerate}
			
			\item {\bf Sigma Protocol Execution:} \nir{Sigma Protocol Execution (in any case messages exchange should be message exchange)}\omri{done}
			\begin{enumerate}
				\item \label{qzk:prover_alpha}
				$\zkP$ computes $(\alpha, \tau) \gets \SigP_1(\ins, \wit)$ and sends $\alpha$.
				\item \label{qzk:verifier_beta}
				$\zkV$ sends the challenge $\beta$. \nir{$\zkV$ sends the challenge $\beta$.}\omri{done}
			\end{enumerate}
			
			\item {\bf WI Proof by the Verifier:}\label{qzk:verifier_wi}
			$\zkV$ gives a WI proof of the following statement: \nir{this subscript notation is very confusing, lose it. Just say that $\zkV$ gives a WI proof that... etc}
			\begin{itemize}
				\item
				The transcript of the verifier so far is explainable.
				\item
				{\bf Or, }$\cmt_1$ is a commitment to a non-witness $u \notin \mathcal{R}_{\lang}(\ins)$.
				\nir{would be more clear to write or $\cmt_1$ is a commitment to a non-witness $u \notin \mathcal{R}_{\lang}(\ins)$. If you insist, add "namely $\cmt_1=...$ for some $r$."}\omri{done}
			\end{itemize}
		The witness that $\zkV$ uses for the proof is its randomness, that proves that the transcript is explainable.
		\nir{also, say which witness $V$ uses}\omri{done}
			
			\item {\bf WI Proof by the Prover:}\label{qzk:prover_wi}
			\nir{lose confusing subscript notation}
			$\zkP$ gives a WI proof of the following statement:
			\begin{itemize}
				\item $x \in \lang$.
								
				\item {\bf Or,} $\cmt_1$, $\cmt_2$ are both valid commitments and furthermore, $\ciph_{\zkP}$ is a valid SFE encryption and is encrypted with a key $\sfedk$ which is the content of the commitment $\cmt_2$.
			\end{itemize}
		The witness that $\zkP$ uses for the proof is $w$, that proves $x \in \lang$.
			
			\item {\bf Sigma Protocol Completion:} \label{qzk:prover_gamma} \nir{Sigma Protocol Completion?}\omri{yes}
			$\zkP$ sends $\gamma = \SigP_3(\beta, \tau)$.
			
			\item
			{\bf Acceptance:} $\zkV$ accepts if $\SigV(\alpha, \beta, \gamma) = 1$.
			
			\item
			{\bf Reactions to Aborts:} During the protocol, if either party sends a message of an incorrect form or provides a non-convincing WI proof, the other party terminates the interaction.
		\end{enumerate}
	}
	
%
%
%
%
%
%

	\subsection{Quantum Soundness}
	\begin{proposition} [The Protocol is Sound]
		Let $\zkV$ be the verifier from \proref{fig:const_round_qzk}.
		For any quantum polynomial-size prover $\zkmP = \set{\zkmP_\secp, \rho_\secp}_{\secp \in \Nat}$, there exists a negligible function $\mu(\cdot)$ such that for any security parameter $\secp\in \Nat$ and any $\ins \in \zo^\secp\setminus\lang$,
		\begin{align*}
		\Pr\left[ \view_{\zkV}\prot{\zkmP_\secp(\rho_\secp)}{\zkV}(\ins) = 1 \right] \leq \mu(\secp)\enspace.
		\end{align*}
	\end{proposition}

	\nir{Use Proposition. Lemma is typically used for mathematical statement that can probably be used elsewhere. }\omri{ok, changing all lemmas that match this description as "Proposition".}

	\begin{proof}
		
\nir{I think that the structure of the proof could be organized better. A reduction to the soundness of the sigma protocol is a bit unnatural in the sense that the sigma protocol is information theoretically sound, so it's weird that you're bothering to construct and efficient prover $\zkmP$. Instead, I would describe a sequence of experiments, going from the real interaction to an ideal interaction, show that the chance of winning is preserved through the hybrids, and that this chance is negligible in the last hybrid, due to the (statistical) soundness of the sigma protocol. This is not a must, but in my eyes would make things easier to read. In particular, it allows understanding exactly where each one of the tools is used.}
\omri{done?}

		Let $\zkmP = \{ \zkmP_\secp, \rho_\secp \}_{\secp \in \Nat}$ a polynomial-size quantum prover and let $x = \{ x_\secp \}_{\secp \in \Nat}$ be a sequence such that $\forall \secp \in \Nat : x_\secp \in \{ 0, 1 \}^\secp \setminus \lang$.
		We prove soundness by a hybrid argument. We consider a series of hybrid processes with output over $\{ 0, 1 \}$, starting from $\view_{\zkV}\prot{\zkmP(\rho)}{\zkV}(\ins)$ the output distribution of $\zkV$ in the interaction with $\zkmP$.
		The proof will show that the probability to output $1$ is negligible, which proves the soundness of the protocol.
		
		We assume without the loss of generality that the first prover message is deterministic, and that the commitments $\cmt_1$, $\cmt_2$ it sends are both valid commitments and furthermore, there is some SFE secret key $\sfedk \in \sfeGen(1^\secp)$ such that $\cmt_2 \in \Com(1^\secp, \sfedk)$.
		First note that if the above property is false, then the whole WI statement of the prover is false (because the first statement in $\zkmP$'s OR statement, that claims $x \in \lang$, is always false in the case of a cheating prover).
		
		This assumption is without the loss of generality because we can consider a new prover that chooses the first message (and quantum inner state at the end of this message) as the message that maximizes the probability that $\zkV$ outputs $1$.
		If this message is such that $\cmt_1$, $\cmt_2$ are not consistent with the prover's WI statement, then by the soundness of the proof that $\zkmP$ gives, with overwhelming probability $\zkV$ outputs $0$ and soundness already holds.
		
		As a final note, observe that because $\cmt_1$, $\cmt_2$ are consistent with the prover's WI statement, $\cmt_1$ is necessarily a commitment to a non-witness $u \notin \rel_{\lang}(x)$, and denote by $r_u$ a string s.t. $\cmt_1 = \Com(1^\secp, u; r_u)$.
	
		Define the following hybrid distributions.
		\begin{itemize}
			\item $\Hyb_0:$
			The output distribution of $\view_{\zkV}\prot{\zkmP(\rho)}{\zkV}(\ins)$.
			
			\item $\Hyb_1:$
			This hybrid process is identical to $\Hyb_0$, with the exception that in step \ref{qzk:verifier_wi} (verifier WI), $\zkV$ uses the information $(u, r_u)$ as witness for its WI statement, instead of the witness that shows its transcript is explainable.
			
			\item $\Hyb_2:$
			This hybrid process is identical to $\Hyb_1$, with the exception that $\zkV$ obtains $\sfedk$, and when it gets the prover message $\ciph_{\zkP}$ in step \ref{qzk:prover_sfe}, it performs the following check: If $t = \sfeD_{\sfedk}( \ciph_{\zkP} )$ then the process halts and outputs $\bot$, otherwise the interaction carries on regularly.
			
			\item $\Hyb_3:$
			This hybrid process is identical to $\Hyb_2$, except that in step \ref{qzk:verifier_response} when the verifier responds with an SFE evaluation, instead of performing an SFE evaluation of $\CC{\mathsf{Id(\cdot)}, t, s}$, the verifier performs an SFE evaluation of $C_\bot$, a circuit that always outputs $\bot$.
			
			\item $\Hyb_4:$
			This hybrid process is identical to $\Hyb_3$, except that the verifier does not perform the check at step \ref{qzk:prover_sfe} like described in $\Hyb_2$.
			
			\item $\Hyb_5:$
			This hybrid process is identical to $\Hyb_4$, except that in step \ref{qzk:verifier_CC} where $\zkV$ sends its first message, the reward value of the CC program $\obfC$ it uses is $(r, 0^{|\beta|})$ instead of $(r, \beta)$.
		\end{itemize}
		
		We now explain why each consecutive pair of the distributions above are statistically indistinguishable (recall that for a pair of distributions over a single bit, they are statistically indistinguishable iff they are computationally indistinguishable).
		We will then use the last process $\Hyb_5$ to show that soundness follows from the soundness of the sigma protocol $(\SigP, \SigV)$.
		\begin{itemize}
			\item $\Hyb_0 \approx_{s} \Hyb_1:$
			Follows from the witness indistinguishability property of the WI proof that the verifier gives.
			
			\item $\Hyb_1 \approx_{s} \Hyb_2:$
			Follows from Claim \ref{claim:find_t}, which says that the probability that $\ciph_{\zkP}$ is an encryption of (the correct) $t$ with the secret key $\sfedk$ (that is inside $\cmt_2$) is negligible, and thus the erasure of such cases can't be noticed by a distinguisher.
			
			\item $\Hyb_2 \approx_{s} \Hyb_3:$
			As a basic explanation, this indistinguishability follows from the combination of the circuit privacy property of the SFE and the soundness of the WI proof that $\zkmP$ gives.
			
			As a fuller explanation, assume toward contradiction there's a distinguisher $\Disting$ that tells the difference between the two distributions, and by an averaging argument, consider the transcript (and inner quantum state of $\zkmP$) generated at the end of step \ref{qzk:prover_sfe} (where $\zkmP$ sends $\ciph_\zkP$), which maximizes $\Disting$'s distinguisability adventage - other than the prover's ciphertext $\ciph_{\zkP}$, this transcript fixes $t, s$, which in turn fix the circuit $\oCC\big[ \mathsf{Id(\cdot)}, t, s \big]$.
			We now consider three cases, and explain why we get a contradiction in each of them.
			\begin{enumerate}
				\item $\ciph_{\zkP} \in \sfeEnc_{\sfedk}(t)$:
				In this case, no matter what will be generated next in the transcript, the output will be $\bot$ (by the check described in $\Hyb_2$), thus it is impossible to distinguish the outputs of the two processes and we get a contradiction.
				
				\item $\exists y \in \{ 0, 1 \}^\secp \setminus \{t\}$ s.t. $\ciph_{\zkP} \in \sfeEnc_{\sfedk}(y)$:
				In this case, $\bot = \oCC\big[ \mathsf{Id(\cdot)}, t, s \big](y)$ and thus we get a contradiction by using the circuit privacy property of the SFE.
				
				\item Else:
				In that case, either $\ciph_{\zkP}$ is a ciphertext encrypted with some other SFE key $\sfedk'$, or it is not a valid ciphertext at all and in any case, it is not a valid ciphertext encrypted with the secret key $\sfedk$.
				In that case, the WI statement of the prover is necessarily false, and thus a $1$ output happens with at most negligible probability in both cases (by the soundness of the WI proof of $\zkmP$), thus the statistical distance between them is at most negligible, in contradiction.
			\end{enumerate}
			
			\item $\Hyb_3 \approx_{s} \Hyb_4:$
			Follows from the same reasoning as in the indistinguishability $\Hyb_1 \approx_{s} \Hyb_2$.
			
			\item $\Hyb_4 \approx_{s} \Hyb_5:$
			Follows from the simulation property (obfuscation security) of the CC obfuscation scheme.
		\end{itemize}
		
		Now, assume toward contradiction that $\zkmP$ succeeds in making the verifier accept with some noticeable probability $\varepsilon(\secp)$, that is, the probability for the output $1$ in $\Hyb_0$ is noticeable.
		$\Hyb_0 \approx_s \Hyb_5$, and thus the probability for the output $1$ in $\Hyb_5$ is also noticeable.
		Finally, we get a contradiction to the soundness of the sigma protocol $(\sigmaP, \sigmaV)$, by using the prover sigma protocol messages from steps \ref{qzk:prover_alpha}, \ref{qzk:prover_gamma} as messages to convince a sigma protocol verifier $\sigmaV$. Since the probability that the verifier $\zkV$ is convinced in $\Hyb_5$ is noticeable, and such verifier is convinced if and only if the sigma protocol verifier is convinced, we get our contradiction.

	\end{proof}

	%
	\begin{claim}[Producing an SFE Encryption of $t$ with $\sfedk$ is Hard] \label{claim:find_t}
		Let $\zkmP = \set{\zkmP_\secp, \rho_\secp}_{\secp \in \Nat}$ be a quantum polynomial-size prover in \proref{fig:const_round_qzk}, sending a deterministic first message $\cmt_1$, $\cmt_2$ where there exists $\sfedk \in \sfeGen(1^\secp)$ s.t. $\cmt_2 \in \Com(1^\secp, \sfedk)$.
		Then there exists a negligible function $\mu(\cdot)$ such that the probability that $t = \sfeD_{\sfedk}(\ciph_{\zkP})$ is bounded by $\mu(\secp)$.
	\end{claim}
	
	\begin{proof}
		The proof will be based on the security of the QFHE, and on the security of the CC obfuscation.
		We start with observing that the security of the QFHE implies that for every efficient quantum adversary $\A = \{ \A_\secp, \rho_\secp \}_{\secp \in \Nat}$, there's a negligible function $\mu(\cdot)$ s.t. the probability that $\A$ finds $t$ given $\qhepk, \ciph \gets \qheE_{\qhepk}(t)$ for a uniformly random $t \gets \{ 0, 1 \}^{\secp}$, is bounded by $\mu(\secp)$ - we will assume toward contradiction that our claim is false, that is, we assume that $\zkmP$ sends $\ciph_{\zkP}$ s.t. $t = \sfeD_{\sfedk}(\ciph_{\zkP})$ with noticeable probability (for infinitely many security parameters), and get a contradiction with the last claim about the hardness of finding a random encrypted $t$.
		
		Using $\zkmP$ and the fact that $t = \sfeD_{\sfedk}(\ciph_{\zkP})$ with noticeable probability, we now describe a (non-uniform) algorithm $\A$ that finds $t$ given $\qhepk, \ciph \gets \qheE_{\qhepk}(t)$ for $t \gets \{ 0, 1 \}_{\secp}$ and thus breaks the security of the QFHE.
		As part of the non-uniform advice of $\A$, it will have the secret SFE key $\sfedk$, which is fixed.
		Given $\qhepk, \ciph \gets \qheE_{\qhepk}(t)$, the algorithm $\A$ will use the simulator $\ccSim^{CC}$ (from the simulation property of the CC obfuscation) and send to $\zkmP$ the following, as the protocol message sent at step \ref{qzk:verifier_CC},
		$$
		\qhepk, \; \ciph, \; \ccSim^{\text{CC}}(1^{|\qheD|}, 1^{\ell + |\beta|}, 1^{\secp}) \enspace ,
		$$
		where $\ell$ is the randomness complexity of the QFHE key generation algorithm $\qheG$.
		$\zkmP$ will respond with $\ciph_{\zkP}$, and $\A$ uses $\sfedk$ to output $\sfeD_{\sfedk}(\ciph_{\zkP})$.
		
		We now use the simulation property guarantee of the CC obfuscation: Note that the probability that $\zkmP$ outputs $\ciph_{\zkP}$ s.t. $\sfeD_{\sfedk}(\ciph_{\zkP}) = t$ in the simulated setting, where $\A$ sends $\ccSim^{\text{CC}}(1^{|\qheD_{0^{|\qhesk|}}|}, 1^{|\qhesk| + |\beta|}, 1^{\secp})$ instead of $\obfC$, is negligibly close to the probability that it outputs $\ciph_{\zkP}$ s.t. $\sfeD_{\sfedk}(\ciph_{\zkP}) = t$ in the regular setting where it gets $\obfC$ - this is due to the security of the CC obfuscator.
		Because we know that in the regular interaction, $\zkmP$ sends $\ciph_{\zkP}$ s.t. $t = \sfeD_{\sfedk}(\ciph_{\zkP})$ with a noticeable probability, this implies that so does $\A$, in contradiction.
	\end{proof}

\omri{Old version of soundness proof is inside a false statement in code. I suggest to get rid of it.}

	\subsection{Quantum Zero-Knowledge} \label{subsec:qzk}
\nir{General comment: I think that the choice to not to abstract out explainability and have a separate compiler makes the whole thing way more difficult to read and analyze. I strongly suggest that for the journal version we separate the two layers. In retrospect I would also abstract out the extractable commitment. I would: (1) construct extractable commitments against explainable senders and semi-malicious receivers (no WI proofs needed). Then (2) use that to construct ZK against explainable verifiers and arbitrary provers (uses WI proofs from the prover and non-uniformity), and then (3) compile to ZK against arbitrary verifier using WI from verifier.}

	We construct a quantum polynomial-time universal simulator $\zkS$ that for a quantum verifier $\zkmV$, an arbitrary quantum auxiliary input $\rho$ and an instance in the language $x \in \lang$, simulates the output distribution of the verifier after the real interaction, $\view_{\zkmV}\prot{\zkP}{\zkmV(\rho)}(\ins)$.
	Throughout this section, a malicious verifier $\zkmV$ is modeled as a family of non-uniform quantum circuits with auxiliary quantum input, consistently with the rest of the paper.
	
	\nir{Throughout this section $\zkmV$ should be a family of circuits consistently with the rest of the paper}\omri{ok?}
	
	\paragraph{High-Level Description of Simulation.}
	Our simulation is composed as follows.
	We first describe two simulators, $\zkS_{\mathrm{a}}$ and $\zkS_{\mathrm{na}}$ that try to simulate different types of transcripts, specifically, $\zkS_{\mathrm{a}}$ will try to simulate an aborting interaction, and $\zkS_{\mathrm{na}}$ will try to simulate a non-aborting interaction.
	By "aborting interaction" and "non-aborting interaction" we formally mean the following:
	\nir{I suggest to use "aborting" and "non-aborting" terminology and instead of $\zkS_0,\zkS_1,\zkS'$ use $\zkS_{\mathrm{a}}$ $\zkS_{\mathrm{na}}$ and $\zkS_{\mathrm{comb}}$. Whatever you do, lose the ' symbol here.}\omri{done}
	\begin{itemize}
		\item
		{\bf An aborting interaction} is one where the verifier $\zkmV$ either aborts before the end of step \ref{qzk:prover_wi} (prover WI), or fails to prove its WI statement in step \ref{qzk:verifier_wi}.
		
		\item
		{\bf A non-aborting interaction} is one that is not aborting.
		More precisely, a non-aborting interaction is one where the verifier did not abort before the end of step \ref{qzk:prover_wi} (prover WI), and also succeeded in proving its WI statement in step \ref{qzk:verifier_wi}.
	\end{itemize}
	
	Our next step will be to describe a unified simulator $\zkS_{\mathrm{comb}}$ that randomly chooses $b \gets \{ \mathrm{a}, \mathrm{na} \}$ and then uses $\zkS_b$ to simulate the interaction.
	We will prove that on input $(x, \zkmV, \rho)$, $\zkS_{\mathrm{comb}}$ outputs a quantum state that is computationally indistinguishable from $\view_{\zkmV}\prot{\zkP}{\zkmV(\rho)}(\ins)$, with the following exception: $\zkslS$ outputs a quantum state $\widetilde{\view}$ that indistinguishable from the real verifier output $\view_{\zkmV}\prot{\zkP}{\zkmV(\rho)}(\ins)$ conditioned on  $\widetilde{\view}\neq \zkFail$. Furthermore $\widetilde{\view} \neq \zkFail$ with probability negligibly close to $1/2$.
	In other words, $\zkslS$ is going to succeed simulating only with probability (negligibly close to) $\frac{1}{2}$.
	\nir{Last sentence should be more clear. $\zkslS$'s outputs a quantum state $\widetilde{\view}$ that indistinguishable from the real verifier output $\view_{\zkmV}\prot{\zkP}{\zkmV(\rho)}(\ins)$ conditioned on  $\widetilde{\view}\neq \mathrm{fail}$. Furthermore $\widetilde{\view} \neq \mathrm{fail}$ with probability negligibly close to $1/2$.}
	
	We further show that $\zkslS$ satisfies the required conditions for applying Watrous' quantum rewinding lemma so that the success probability can be amplified from $\approx 1/2$ to $\approx 1$.
	\nir{Last sentence is out of context. We further show that $\zkslS$ satisfies the required conditions for applying Watrous' quantum rewinding lemma so that the success probability can be amplified from $\approx 1/2$ to $\approx 1$.}

\paragraph{The Actual Proof.} We start by describing the above mentioned simulators.
	
	\medskip
	\paragraph{$\zkS_{\mathrm{a}}(x, \zkmV, \rho):$} \nir{make step titles consistent with the protocol}
	\begin{enumerate}
		\item {\bf Simulation of Initial Commitments and Verifier Message:}
		\begin{enumerate}
			\item
			$\zkS_{\mathrm{a}}$ computes $\sfedk \gets \sfeGen(1^\secp)$ and sends to $\zkmV$ the commitments $\cmt_1 \gets \Com(1^\secp, 0^{|\wit|})$, $\cmt_2 \gets \Com(1^{\secp}, \sfedk)$.
			
			\item \label{simulation:verifier_CC}
			$\zkmV$ sends $\qhepk$, $\ciph_{\zkmV}$, $\obfC$.		
		\end{enumerate}
		
		\item {\bf Trying to get an Abort:}
		$\zkS_{\mathrm{a}}$ interacts with $\zkmV$ as the honest prover $\zkP$ until the end of step \ref{qzk:prover_wi} of the original protocol, with exactly 2 differences:
		\begin{itemize}
			\item
			The message $\alpha$ at step \ref{qzk:prover_alpha} is generated by the sigma protocol simulator $\alpha \gets \SigS(\ins, 0^{|\beta|})$, and not by the sigma protocol prover.
			
			\item
			At step \ref{qzk:prover_wi}, the witness used to prove the WI statement is for the second statement in the OR expression (that the commitments $\cmt_1, \cmt_2$ are valid and consistent), and not the first (that $\ins \in \lang$).
		\end{itemize}
	\nir{I think this is making it hard for the reader, who doesn't really remember the prover. I would write all the steps explicitly and mark in red those different from the real prover (including the previous prover commitment step)}
		
		\item {\bf Simulation Verdict:}
		If at some point $\zkmV$ aborts or fails in its WI proof, $\zkS_{\mathrm{a}}$ outputs the aborting verifier's output. Otherwise, $\zkS_{\mathrm{a}}$ outputs $\zkFail$.
		\nir{This description of the simulator should be as concise as possible (not the place for a high-level explanation). For example in the last item can be: "If at some point $\zkmV$ aborts or fails in its WI proof, $\zkS_0$ outputs the aborting verifier's output. Otherwise, $\zkS_0$ outputs $\mathrm{fail}$."}\omri{done}
	\end{enumerate}

	\medskip
	$\zkS_{\mathrm{na}}(x, \zkmV, \rho):$
	\begin{enumerate}
		
		\item {\bf Simulation of Initial Commitments and Verifier Message:}
		\begin{enumerate}
			\item \label{simulation:prover_commit}
			$\zkS_{\mathrm{na}}$ computes $\sfedk \gets \sfeGen(1^\secp)$ and sends to $\zkmV$ the commitments $\cmt_1 \gets \Com(1^\secp, 0^{|\wit|})$, $\cmt_2 \gets \Com(1^{\secp}, \sfedk)$.
			
			\item \label{simulation:verifier_CC}
			$\zkmV$ sends $\qhepk$, $\ciph_{\zkmV}$, $\obfC$.		
		\end{enumerate}
		
		\item {\bf Non-Black-Box Extraction Attempt:} \label{simulation:extraction}
		\nir{remove  this sentence, already captured by title. Should be clear but concise.}\omri{done.}
		\begin{enumerate}
			\item
			$\zkS_{\mathrm{na}}$ computes 
			$$
			r_1 \gets \{ 0, 1 \}^*,\hspace{5mm} \ciph_{t}^{\text{SFE}} = \qheEv_{\qhepk}(\sfeEnc_{\sfedk}(\cdot; r_1), \, \ciph_{\zkmV})\enspace.
			$$ 
			$\zkS_{\mathrm{na}}$ also encrypts $\rho^{(1)}$, the inner (quantum) state of the verifier after its first message:
			$$
			\ciph_{\rho^{(1)}} \gets \qheQE_{\qhepk}(\rho^{(1)})\enspace.
			$$
			
			\nir{This is not the place to analyze the validity. If you want can put a short comment in parenthesis and different color/font, like you do when you write comments in code. (I would remove this altogether, the notation is already expressive enough)}\omri{done.}
			
			\nir{Breaking the homomorphic computation this way is a bit cumbersome and breaks the intuition. I suggest to define the circuit $\zkmV_2$ that given value $\tilde t$ and verifier state after first message $\rho_1$ emulates the second verifier message given an sfe enc of $\tilde t$. Then perform one homomorphic computation of this circuit all at once, no need to think about intermediate homomorphic states (in fact, you only defined homomorphism over fresh ct's, which is fine and sufficient).}
			\nir{the notation $\ciph_{t}^{\text{SFE}}$ is confusing, it's an FHE CT of an SFE CT (if you don't break the FHE evaluation, then this confusion will go away)}
			
			\item \label{simulation:extraction_second_step}
			$\zkS_{\mathrm{na}}$ performs a quantum homomorphic evaluation of the verifier's response. It computes,
			$$
			\left( \ciph_{s}^{\text{SFE}}, \ciph_{\rho^{(2)}} \right) \gets \qheEv_{\qhepk}\left(\zkmV, \big( \ciph_{t}^{\text{SFE}}, \ciph_{\rho^{(1)}} \big)\right) \enspace .
			$$
			\nir{same...validity analysis shouldn't be mixed with the simulator's description.}\omri{done}
			
			\item \label{simulation:extraction_last_step}
			$\zkS_{\mathrm{na}}$ computes $\ciph_s \gets \qheEv_{\qhepk}\left(\sfeD_{\sfedk}(\cdot), \ciph_{s}^{\text{SFE}} \right)$, and then computes $(r, \beta') = \obfC(\ciph_s)$.
			
			\nir{here the homomorphic evaluation should cease, and the CC evaluation should be a separate logical step}
			
			\item \label{simulation:extraction_check}
			$\zkS_{\mathrm{na}}$ checks validity: $(\qhepk', \qhesk) = \qheG(1^\secp ; r)$, if $\qhepk' \neq \qhepk$ then it halts simulation and outputs $\zkFail$.
			Otherwise, $\zkS_{\mathrm{na}}$ obtains the inner state of $\zkmV$ by decryption: $\rho^{(2)} \gets \qheQD_{\qhesk}(\ciph_{\rho^{(2)}})$.
			Additionally, $\zkS_{\mathrm{na}}$ simulates the missing transcript (for the verifier to later prove that its messages were explainable): for the prover message at step \ref{qzk:prover_sfe} it inserts $\ciph_t = \sfeEnc_{\sfedk}(t; r_1)$, and for the verifier message at step \ref{qzk:verifier_response} it inserts $\evciph_s = \qheD_{\qhesk}(\ciph_{s}^{\text{SFE}})$.
			
			\nir{I don't think this check is needed. At this point you're assuming the verifier is explainable, if later it will fail in the WI you're anyhow going to announce $\mathrm{fail}$. You can just give the secret key instead of $r$.}
			\nir{just define the transcript as part of the output of the circuit that you're homomoprhically evaluating, and throw away the last line.}
		\end{enumerate}
		
		\item {\bf Sigma Protocol Messages Simulation:}
		\begin{enumerate}
			\item \label{simulation:sigma_simulation}
			$\zkS_{\mathrm{na}}$ executes the sigma protocol simulator $(\alpha, \gamma) \gets \SigS(x, \beta')$ and sends $\alpha$ to $\zkmV$.
			
			\item \label{simulation:verifier_sends_beta}
			$\zkmV$ returns $\beta$.
		\end{enumerate}
		
		\item {\bf WI Proof by the Malicious Verifier:} \label{simulation:verifier_wi}
		$\zkS_{\mathrm{na}}$ takes the role of the honest prover $\zkP$ in the WI proof $\zkmV$ gives.
		If $\zkmV$ fails to prove the statement, the simulation fails and the output is $\zkFail$.
		
		\item {\bf Simulation of the Prover's WI Proof and Information Reveal:} \label{simulation:prover_wi}
		$\zkS_{\mathrm{na}}$ gives $\zkmV$ a WI proof using the witness that shows $\cmt_1$, $\cmt_2$ are both valid commitments (and that $\cmt_2$ is a commitment to the SFE key $\sfedk$ used in step \ref{qzk:prover_sfe}).
		After the proof, $\zkS_{\mathrm{na}}$ sends $\gamma$ to $\zkmV$.
		
		\item {\bf Simulation Verdict:}
		If $\zkmV$ completed interaction without aborting and gave a convincing WI proof, $\zkS_{\mathrm{na}}$ outputs the verifier's output. Otherwise, $\zkS_{\mathrm{na}}$ outputs $\zkFail$.
		\nir{make concise as for previous sim.}\omri{done}
	\end{enumerate}

	\medskip
	\paragraph{$\zkslS(x, \zkmV, \rho):$} Sample $b \gets \{ 0, 1 \}$ and execute $\zkS_b(x, \zkmV, \rho)$.

\nir{doesn't make sense to have an enumerated list for a single item}\omri{done}

	\medskip
	\paragraph{$\zkS(x, \zkmV, \rho):$}
	\begin{enumerate}
		\item
		Generate the circuit $\zkslSmV$, which is the circuit implementation of $\zkslS$, with hardwired input $x$, $\zkmV$, that is, the only input to $\zkslSmV$ is the quantum state $\rho$.
		
		\item
		Let $\Wat$ be the algorithm from Lemma \ref{lem:watrous}.
		The output of the simulation is $\Wat(\zkslSmV, \rho, \secp)$.
	\end{enumerate}

	\medskip
	\paragraph{Proof of Simulation Validity.}\nir{edited the text below a bit}
	We now turn to prove that the simulated output $\zkS(x, \zkmV, \rho)$ is computationally indistinguishable from $\view_{\zkmV}\prot{\zkP}{\zkmV(\rho)}(\ins)$.
	This is done in several steps:
	\begin{enumerate}
		\item
		{\bf Simulating aborting interactions:} Let $\zkmV_{\mathrm{a}}$ \nir{replace with $\zkmV_{\mathrm{a}}$}\omri{done} be the augmented verifier that is identical to $\zkmV$, with the exception that if $\zkmV$ does not abort, $\zkmV_{\mathrm{a}}$ outputs $\zkFail$.\nir{I suggest to change $\bot$ to $\mathrm{fail}$.}\omri{done} Then the output of $\zkS_{\mathrm{a}}$ is indistinguishable from the output of $\zkmV_{\mathrm{a}}$ in a real interaction.

		\item
		{\bf Simulating non-aborting interactions:} Let $\zkmV_{\mathrm{na}}$ \nir{replace with $\zkmV_{\mathrm{na}}$}\omri{done} be the augmented verifier that is identical to $\zkmV$, with the exception that if $\zkmV$ aborts, $\zkmV_{\mathrm{na}}$ outputs $\zkFail$.\nir{I suggest to change $\bot$ to $\mathrm{fail}$.}\omri{done} Then the output of $\zkS_{\mathrm{na}}$ is indistinguishable from the output of $\zkmV_{\mathrm{na}}$ in a real interaction.
		\item
		The above two statements imply:
		\begin{itemize}
			\item
			$\zkslS \neq \zkFail$ with probability negligibly close to $\frac{1}{2}$, for every verifier and auxiliary input $\rho$.
			
			\item The output of $\zkmV$ in a real interaction is indistinguishable from the output of $\zkslS$ conditioned on $\zkslS \neq \zkFail$.
 \nir{$\bot$ to $\mathrm{fail}$}\omri{done}
		\end{itemize}
	These in turn imply that we can use Watrous' quantum rewinding lemma in order to amplify $\zkslS$ into a full-fledged simulator $\zkS$
	\end{enumerate}

\nir{lemma => proposition}\omri{done}
	\begin{proposition} [Similarity of Aborting Part] \label{claim:aborting_similarity}		
		Let $\zkmV = \set{\zkmV_\secp, \rho_\secp}_{\secp \in \Nat}$ a polynomial-size quantum verifier, and let $\view_{\zkmV_{\mathrm{a}}}$ be the verifier's output at the end of protocol such that if $\zkmV$ does not abort, the output is $\zkFail$. Then,
		$$
		\{ \view_{\zkmV_{\mathrm{a}}}\prot{\zkP(w)}{\zkmV_\secp(\rho_{\secp})}(\ins)\}_{\secp, \ins, w}
		\approx_{c}
		\{ \zkS_{\mathrm{a}}(\ins, \zkmV_\secp, \rho_\secp)\}_{\secp, \ins, w} \enspace ,
		$$
		where $\secp\in\Nat$, $\ins\in\lang\cap\zo^\secp$, $w \in \mathcal{R}_{\lang}(x)$.
		
		\nir{the notation would make more sense as $\view_{\zkmV}^A$ also (identify previous simulator name with this $A->0$ here or $0->A$ there)}\omri{for now I matched it with the high-level simulators description.}
		\nir{did not end with a verifier abort}\omri{done, "successful" is switched to "non-aborting" in entire section.}
		\nir{remove rest of sentence}\omri{done}
		\nir{there's inconsistency above with $\secp$ subscripts for $\rho$ and also for $\ins$}\omri{fixed, I think.}
	\end{proposition}

	\begin{proof}
		We prove the claim by a hybrid argument, specifically, we consider hybrid distributions, all of which will be computationally indistinguishable.
		\nir{Why not $H$ or $Hyb$, which is more standard. Also use plain subscripts instead of $H^{(i)}$}\omri{done}
		\begin{itemize}
			\item $\Hyb_0:$ The output distribution of $\zkS_{\mathrm{a}}$.
			
			\item $\Hyb_1:$ This hybrid process is identical to $\Hyb_0$, with the exception that when the simulator gives a WI proof in the simulation, it uses the witness $w$ in the proof, that proves the first statement in the OR statement ($x \in \lang$) rather then the second statement.
			\nir{why do you have this? the simulation never reaches this step anyhow, if the verifier didn't abort, then it outputs $\mathrm{fail}$}
			\nir{reader does not remember what is step 5, hence you should refer to it by first its title the "Prover WI (step 5)" and be more explicit }\omri{I hope that now it is better.}
			
			\item $\Hyb_2:$ This hybrid process is identical to $\Hyb_1$, with the exception that $\cmt_1$ is a commitment to $\wit$ rather than to $0^{|\wit|}$.
			
			\item $\Hyb_3:$ This hybrid process is identical to $\Hyb_2$, with the exception that the message $\alpha$ that the simulator sends to $\zkmV$ is generated by the actual sigma protocol $(\alpha, \tau) \gets \SigP_1(\ins, \wit)$, and not by the sigma protocol simulator $\SigS(\ins, 0^{|\beta|})$.
			Note that this process is exactly $\view_{\zkmV_{\mathrm{a}}}\prot{\zkP(w)}{\zkmV(\rho)}(\ins)$.
			\nir{is this consistent with the sigma protocol notation from the definition?}\omri{I think it is now.}
		\end{itemize}
		
		It is \nir{=> It is}\omri{done} left to reason about the indistinguishability between each two subsequent hybrids\nir{two subsequent hybrids}\omri{done}.
		\begin{itemize}
			\item $\Hyb_0 \approx_{c} \Hyb_1:$
			Follows from the witness-indistinguishability property of the WI proof that the simulator gives (as the prover) in step \ref{qzk:prover_wi} of the protocol.
			
			\item $\Hyb_1 \approx_{c} \Hyb_2:$
			Follows from the hiding property of the commitment $\cmt_1$.
			
			\item $\Hyb_2 \approx_{c} \Hyb_3:$
			Follows from Claim \ref{first_message_indisting}.
		\end{itemize}
		\nir{People cannot be expected to remember the hybrids. Before proving the indistinguishability of each two you need to remind them of the difference, even in just a few words.}
	\end{proof}
	
	\omri{This is the heaviest proof in the paper, it is exhausting and technical. I wasn't sure how to make it significantly simpler without paying in its "convincingness".}

	\begin{proposition} [Similarity of Non-Aborting Part] \label{claim:non_aborting_similarity}
		Let $\zkmV = \set{\zkmV_\secp, \rho_\secp}_{\secp \in \Nat}$ a polynomial-size quantum verifier, and let $\view_{\zkmV_{\mathrm{na}}}$ be the verifier's output at the end of protocol such that if $\zkmV$ aborts, the output is $\zkFail$. Then,
		$$
		\{ \view_{\zkmV_{\mathrm{na}}}\prot{\zkP(w)}{\zkmV_\secp(\rho_{\secp})}(\ins)\}_{\secp, \ins, w}
		\approx_{c}
		\{ \zkS_{\mathrm{na}}(\ins, \zkmV_\secp, \rho_\secp)\}_{\secp, \ins, w} \enspace ,
		$$
		where $\secp\in\Nat$, $\ins\in\lang\cap\zo^\secp$, $w \in \mathcal{R}_{\lang}(x)$.
	\nir{same comments as for the previous lemma}\omri{Hopefully they are addressed as well?}
	\end{proposition}

	\begin{proof}
		We prove the claim by a hybrid argument, specifically, we consider hybrid distributions, all of which will be computationally indistinguishable.
		\nir{same comments about notation as before}\omri{corrected.} \nir{as before, when referring to steps aim to use their semantic name and just step number.}

\omri{NIR'S SUGGESTION FOR SIMPLER PROOF IN FALSE STATEMENT INSIDE CODE. For next version.}

		\begin{itemize}
		\item $\Hyb_0:$ The output distribution of $\zkS_{\mathrm{na}}$.
		
		\item $\Hyb_1:$ This hybrid process is identical to $\Hyb_0$, with the exception that when the simulator gives a WI proof in the simulation, it uses the witness $w$ in the proof, that proves the first statement in the OR statement ($x \in \lang$) rather then the second statement.
		
		\item $\Hyb_2:$ This hybrid process is identical to $\Hyb_1$, with the exception that $\cmt_1$ is a commitment to the witness $\wit$ rather than to $0^{|\wit|}$, and $\cmt_2$ is a commitment to $0^{|\sfedk|}$ rather than to the generated SFE key $\sfedk \gets \sfeGen(1^\secp)$.
		
		\item $\Hyb_3:$ This hybrid process is identical to $\Hyb_2$, with the exception that if the verifier's message $\beta$ from part \ref{simulation:verifier_sends_beta} of the simulation does not match the extracted $\beta'$ from part \ref{simulation:extraction_last_step} of the simulation, the process halts on the spot and outputs $\zkFail$.
		
		\item $\Hyb_4:$ This hybrid process is identical to $\Hyb_3$, with the exception that in parts \ref{simulation:sigma_simulation}, \ref{simulation:prover_wi} where the simulator sends $\alpha$ and $\gamma$, instead of computing $\alpha, \gamma$ using $\SigS$, it computes $(\alpha, \tau) \gets \SigP_1(\ins, \wit)$ and $\gamma \gets \SigP_3(\beta, \tau)$.
		\nir{is this consistent with the sigma protocol notation from the definition?}\omri{I think it is now.}
		
		\item $\Hyb_5:$ This hybrid process is identical to $\Hyb_4$, with the exception that it does not perform the check described in $\Hyb_3$, that is, even if the extracted challenge and sent challenge are distinct $\beta'\neq \beta$, the process carries on regularly.
		
		\item $\Hyb_6:$ \nir{this seems more like a proof then a hybrid. Keep your hybrids simple enough to describe in three lines: either break it, or push the complexity to the proof of indistinguishability. Currently this is too hard to read.}At this point in our series of hybrid distributions, we do not use the extracted challenge $\beta'$, and we would like to move to a process that does not perform extraction.
		The current hybrid will still perform extraction, but will not use the extracted information.
		This hybrid process is identical to $\Hyb_5$, with the changes described next.
		If the first verifier message \emph{is not} explainable then the process chooses to fail and outputs $\zkFail$.
		If the first verifier message \emph{is} explainable, note that it fixes a public and secret key pair $(\qhepk, \qhesk) = \qheG(1^\secp; r)$, and a string $s \in \{ 0, 1 \}^{\secp}$ hidden inside the CC program $\obfC$.
		In that case, the process acts like $\Hyb_5$, except that at the end of step \ref{simulation:extraction_second_step} of the simulation, the process inefficiently obtains $\qhesk$\nir{from where? from $\qhepk$? you didn't explicitly define that it sets a unique secret key, so you formally need to say some $\qhesk$ that is consistent with it.}\omri{I think I did explicitly say it: "First, note that if the first verifier message...", but now changed a bit.} and uses it to decrypt $\left( \ciph_{s}^{\text{SFE}}, \ciph_{\rho^{(2)}} \right)$, instead of using the program $\obfC$ to get $\qhesk$.
		The process also inefficiently obtains $s$ and performs a check: if $s \neq \sfeD_{\sfedk}(\qheD_{\qhesk}( \ciph_{s}^{\text{SFE}} ))$\nir{from where does the process get $s$?}\omri{right, I didn't say that $s$ is also fixed. added.} then the process fails and outputs $\zkFail$, and otherwise continues the simulation regularly as in the rest of $\Hyb_5$.
		
		\item $\Hyb_7:$ This process will get rid of extraction altogether and will not perform the homomorphic evaluation of the verifier's response.
		This distributions is the output of a process that acts like $\Hyb_6$, with the exception that if the first verifier message is explainable (in particular, $\ciph_{\zkV}$ is a QFHE encryption of some $t \in \{ 0, 1 \}^{\secp}$), then as the prover message from step \ref{qzk:prover_sfe} of the protocol, the process sends $\sfeEnc_{\sfedk}(t)$.
		If at step \ref{qzk:verifier_response} the verifier responds with $\evciph$ s.t. $s = \sfeD_{\sfedk}(\evciph)$ then the simulation continues regularly as in $\Hyb_6$, and otherwise the process fails and outputs $\zkFail$.
		
		\item $\Hyb_8:$ Like the previous two processes, this process is also inefficient.
		This hybrid process is identical to $\Hyb_7$, with the exception that it does not perform the check on the verifier's response $\evciph$, and continues regularly either way, even when $s \neq \sfeD_{\sfedk}(\evciph)$.
		
		\item $\Hyb_9:$ We now go back to an efficient hybrid process.
		This hybrid process is identical to $\Hyb_8$, with the exception that instead of performing the inefficient check on the verifier's first message from step \ref{simulation:verifier_CC} (and then either halting and outputting $\zkFail$, or sending $\sfeEnc_{\sfedk}(t)$ to $\zkmV$), the process always sends $\sfeEnc_{\sfedk}(0^\secp)$ to $\zkmV$, and continues simuation regularly.
		Observe that this process is exactly $\view_{\zkmV_{\mathrm{na}}}\prot{\zkP(w)}{\zkmV(\rho)}(\ins)$.
	\end{itemize}
		
	We now prove that\nir{prove that}\omri{done} each pair of consecutive distributions are computationally indistinguishable, and our proof is finished.
	\begin{itemize}
		\item $\Hyb_0 \approx_{c} \Hyb_1:$
		This indistinguishability follows from the witness-indistinguishability property of the WI proof that the simulator gives in step \ref{simulation:prover_wi} of the simulation.
		
		\item $\Hyb_1 \approx_{c} \Hyb_2:$
		This indistinguishability follows from the hiding of the commitments $\cmt_1, \cmt_2$ that the simulator gives in step \ref{simulation:prover_commit} of the simulation.
		
		\item $\Hyb_2 \approx_{s} \Hyb_3:$
		\nir{here this wouldn't be true without FHE correctness etc. For some reason you're arguing twice about the correctness of the process...}\omri{I am not sure about this. We are not really using the FHE correctness here yet, as we are not getting rid of extraction yet.}
		The indistinguishability follows from the perfect correctness of both the CC obfuscation and the SFE schemes, along with the soundness of the WI proof.
		Assume toward contradiction that the two distributions are distinguishable and fix, by an averaging argument, the partial transcript $T'$ that is generated at the end of step \ref{simulation:verifier_sends_beta} of the simulation, which maximizes distinguishability.
		We consider two cases for $T'$:
		\begin{itemize}
			\item
			$T'$ is explainable.
			In that case it follows from the perfect correctness of the CC obfuscation and the perfect correctness of the SFE evaluation, that the extracted $\beta'$ and the sent $\beta$ are necessarily equal, and the processes are identical (and have statistical distance of 0), in contradiction.
			
			\item
			$T'$ is not explainable.
			In that case, recall that $\cmt_1$ is a commitment to a witness and thus the statement in the verifier's WI proof is necessarily false.
			By the soundness of the WI proof, $\zkmV$ will fail in proving the statement with overwhelming probability, which implies that with the same probability the output in the process $\Hyb_2$ is $\zkFail$.
			Because with \emph{at least} the same probability, the output in $\Hyb_3$ is also $\zkFail$, the contradiction follows.
		\end{itemize}
			
		\item $\Hyb_3 \approx_{c} \Hyb_4:$
		This indistinguishability follows from the special zero-knowledge property of the sigma protocol.
		
		\item $\Hyb_4 \approx_{s} \Hyb_5:$
		The statistical indistinguishability follows from the exact same reasoning that explains why distributions $\Hyb_2 \approx_{s} \Hyb_3$.
		
		\item $\Hyb_5 \approx_{s} \Hyb_6:$
		This indistinguishability will follow from the perfect correctness of the CC obfuscation, the statistical correctness of the QFHE and from the soundness of the WI proof that $\zkmV$ gives.
		Formally, assume toward contradiction that the two distributions are distinguishable and fix, by an averaging argument, the partial transcript $T'$ and inner quantum state $\rho^{(1)}$ of $\zkmV$ generated at the end of step \ref{simulation:verifier_CC} of the simulation, that maximize the distinguishability.
		Denote by $\tilde{\Hyb_5}$, $\tilde{\Hyb_6}$ the distributions that carry on from the point that $T', \rho^{(1)}$ are fixed, according to $\Hyb_5$, $\Hyb_6$, respectively.
		Consider two cases for $T'$.
		\begin{itemize}
			\item
			$T'$ is not explainable.
			In that case, $\tilde{\Hyb_6}$ outputs $\zkFail$ with probability 1, and by the soundness of the WI proof of the verifier, the proof is going to fail with overwhelming probability (in the process $\tilde{\Hyb_5}$) and with at least the same probability the output is going to be $\zkFail$, and the two distributions will have at most negligible statistical distance, in contradiction.
			
			\item
			$T'$ is explainable, which means that the verifier's first message fixes $t, s, \qhesk$, and also $r_t$ the QFHE encryption randomness s.t. $\ciph_{\zkmV} = \qheE_{\qhepk}(t ; r_t)$.
			Consider the quantum circuit $C$ that for input $(t, \rho^{(1)})$, encrypts $\ciph_t \gets \sfeEnc_{\sfedk}(t)$, executes $(\evciph, \rho^{(2)}) \gets \zkmV(\ciph_t, \rho^{(1)})$, decrypts $s' = \sfeD_{\sfedk}(\evciph)$ and outputs $s', \evciph, \rho^{(2)}$.
			Now, observe the following about the distributions $\tilde{\Hyb_5}$, $\tilde{\Hyb_6}$.
			\begin{itemize}
				\item
				$\tilde{\Hyb_5}$ can be described by the following process: Encrypt $\ciph_{\zkmV} = \qheE_{\qhepk}(t ; r_t)$, $\ciph_{\rho^{(1)}} \gets \gets \qheQE_{\qhepk}(\rho^{(1)})$, perform homomorphic evaluation of the circuit $C$, and then decrypt with $\qhesk$ to get $\left( s', \evciph, \rho^{(2)} \right)$.
				
				If $s' \neq s$ then output $\zkFail$, otherwise carry on the simulation as in $\Hyb_5$.
				The fact that $\tilde{\Hyb_5}$ can be described by this process follows from the fact that by the perfect correctness of the CC obfuscation, if the first verifier message is explainable then $\obfC$ indeed executes the decryption circuit $\qheD_{\qhesk}(\cdot)$ s.t. if the result was $s$, it outputs the QFHE key-generation randomness $r$ (which in turn yields $\qhesk$), and if the result wasn't $s$, $\obfC$ necessarily yields $\bot$.
				
				\item
				$\tilde{\Hyb_6}$ can be described by the following process: the exact same homomorphic evaluation process as described above, except that after getting the output $\left( s', \evciph, \rho^{(2)} \right)$, the check is that $s = \sfeD_{\sfedk}(\evciph)$, and if the check fails the output is $\zkFail$, and if the check succeeds then the process continues simulation regularly in the exact same way as in $\tilde{\Hyb_5}$.
			\end{itemize}
		
			The above descriptions of $\tilde{\Hyb_5}$, $\tilde{\Hyb_6}$ imply that the statistical distance between them is bounded by the probability that the check in one process fails and in the other it succeeds, which in turn bounded by the probability that $s' \neq \sfeD_{\sfedk}(\evciph)$.
			The point is, due to the fact that the SFE decryption algorithm is deterministic, it is always the case when evaluating the circuit $C$ (out in the open, not under homomorphic evaluation) we have $s' = \sfeD_{\sfedk}(\evciph)$.
			It follows by the statistical correctness of the QFHE that the probability that the evaluated $s', \evciph$ are s.t. $s' \neq \sfeD_{\sfedk}(\evciph)$, and thus the bound on the statistical distance between $\tilde{\Hyb_5}$, $\tilde{\Hyb_6}$, in contradiction.
		\end{itemize}
		
		\item $\Hyb_6 \approx_{s} \Hyb_7:$
		This indistinguishability follows directly from the statistical correctness of the QFHE.
		Assume toward contradiction that the two distributions are distinguishable and fix, by an averaging argument, the partial transcript $T'$ and inner quantum state $\rho^{(1)}$ of $\zkmV$ generated at the end of step \ref{simulation:verifier_CC} of the simulation, that maximize the distinguishability.
		Denote by $\tilde{\Hyb_6}$, $\tilde{\Hyb_7}$ the distributions that carry on from the point that $T', \rho^{(1)}$ are fixed, according to $\Hyb_6$, $\Hyb_7$, respectively.
		Consider two cases for $T'$.
		\begin{itemize}
			\item
			$T'$ is not explainable.
			In that case both processes act the same and output $\zkFail$, and are indistinguishable.
			
			\item
			$T'$ is explainable, which means that the verifier's first message fixes $t, s, \qhesk$, and also $r_t$ the QFHE encryption randomness s.t. $\ciph_{\zkmV} = \qheE_{\qhepk}(t ; r_t)$.
			In that case, recall the circuit $C$ from the above proof of the indistinguishability $\Hyb_5 \approx_{s} \Hyb_6$, and observe the following about the distributions $\tilde{\Hyb_6}$, $\tilde{\Hyb_7}$.
			\begin{itemize}
				\item
				The distribution $\tilde{\Hyb_6}$ can be described by the following process: Encrypt $\ciph_{\zkmV} = \qheE_{\qhepk}(t ; r_t)$, $\ciph_{\rho^{(1)}} \gets \gets \qheQE_{\qhepk}(\rho^{(1)})$, perform homomorphic evaluation of the circuit $C$, and then decrypt with $\qhesk$ to get $\left( s', \evciph, \rho^{(2)} \right)$.
				If $s = \sfeD_{\sfedk}(\evciph)$ then process continues simulation regularly as in $\Hyb_6$, and otherwise fails and outputs $\zkFail$.
				
				\item
				The distribution $\tilde{\Hyb_7}$ can be described by the following process: Instead of encrypting $t, \rho^{(1)}$ and computing $C$ under homomorphic evaluation (and then decrypting), we simply execute $\left( \evciph, \rho^{(2)} \right) \gets C(t, \rho^{(1)})$ in the clear.
				The process continues in the exact same way as described after the homomorphic evaluation in $\tilde{\Hyb_6}$; If $s = \sfeD_{\sfedk}(\evciph)$ then process continues simulation regularly, and otherwise fails and outputs $\zkFail$.
			\end{itemize}
			The above implies that the only difference between the two processes is the fact that in $\tilde{\Hyb_6}$ we execute $C$ under homomorphic evaluation, and in $\tilde{\Hyb_7}$ we execute $C$ in the clear.
			By the statistical correctness of the QFHE, it follows that the two processes are statistically indistinguishable, in contradiction.	
		\end{itemize}
		
		\item $\Hyb_7 \approx_{s} \Hyb_8:$
		This indistinguishability follows from the perfect correctness of the SFE encryption and the soundness of the WI proof that $\zkmV$ gives.
		Assume toward contradiction that the distributions are distinguishable and fix, by an averaging argument, the partial transcript $T'$ (and inner verifier state $\rho^{(2)}$) generated after the verifier's second message $\evciph$.
		If the first verifier message was explainable and also $s = \sfeD_{\sfedk}(\evciph)$ then then processes are identical, as they carry on simulation in the exact same way.
		If the first verifier message was not explainable then again, both processes fail and output $\zkFail$ and are identical, and if the first verifier message is explainable but $s \neq \sfeD_{\sfedk}(\evciph)$, it follows $\Hyb_7$ outputs $\zkFail$, and in $\Hyb_8$, by the perfect correctness of the SFE evaluation, the transcript cannot be explainable, and thus the WI proof by the verifier fails with overwhelming probability, and with the same probability the output of $\Hyb_8$ is $\zkFail$, and the processes are indistinguishable.
		
		\item $\Hyb_8 \approx_{c} \Hyb_9:$
		This indistinguishability follows from the input privacy property of the SFE encryption.
		More precisely, as usual, we assume toward contradiction that the distributions are distinguishable and we fix the transcript until the end of step \ref{simulation:verifier_CC} of the simulation.
		If the transcript is not explainable then $\Hyb_8$ outputs $\zkFail$, and $\Hyb_9$ outputs $\zkFail$ with overwhelming probability, because the WI proof of the verifier will fail with overwhelming probability.
		If the transcript is explainable, we can get either an SFE encryption of $t$ or of $0$, as $t$ is fixed by the averaging argument.
		By continuing the simulation regularly, as identically performed in both processes, we get the reduction from breaking the security of the SFE encryption to distinguishing between $\Hyb_8$ and $\Hyb_9$.
		\end{itemize}
	\end{proof}

	\nir{the rest of the proofs and corollaries in this section can be replaced with one short proof. }
	
	\begin{corollary} [Probabilities to Abort are Negligibly Close Over Different Cases] \label{cor:probabilities}
		For a quantum auxiliary input $\rho$, instance in the language $x \in \{ 0, 1 \}^\secp \cap \lang$ and witness $w \in \mathcal{R}_{\lang}(x)$, define the following probabilities.
		\begin{itemize}
			\item
			$a(x, \rho):$ The probability that in the simulation $\zkS_{\mathrm{a}}(x, \zkmV, \rho)$, the verifier $\zkmV$ aborted before the end of step \ref{qzk:prover_wi} where the simulator simulates the Prover's WI proof, or failed to prove its WI statement in step \ref{qzk:verifier_wi} (i.e. the simulation of $\zkS_{\mathrm{a}}(x, \zkmV, \rho)$ was aboting).
			
			\item
			$b(x, \rho):$ The probability that in the simulation $\zkS_{\mathrm{na}}(x, \zkmV, \rho)$, the verifier $\zkmV$ aborted before the end of step \ref{simulation:prover_wi} where the simulator simulates the Prover's WI proof, or failed to prove its WI statement in step \ref{simulation:verifier_wi} (i.e. the simulation of $\zkS_{\mathrm{na}}(x, \zkmV, \rho)$ was aboting).

			\item
			$c(x, \rho, w):$ The probability that the interaction $\prot{\zkP(w)}{\zkmV(\rho)}(\ins)$ was aborting.
		\end{itemize}
		There exists a negligible function $\negl(\cdot)$ s.t. for every sequences $\rho = \{ \rho_\secp \}_{\secp \in \Nat}$, $x = \{ x_\secp \}_{\secp \in \Nat}$, $w = \{ w_\secp \}_{\secp \in \Nat}$ where,
		\begin{itemize}
			\item
			$\forall \secp \in \Nat : \rho_\secp$ is a $\secp^c$-size quantum state (for some constant $c \in \Nat$),
			
			\item
			$\forall \secp \in \Nat : x_\secp \in \{ 0, 1 \}^\secp \cap \lang$,
			
			\item
			$\forall \secp \in \Nat : w_\secp \in \mathcal{R}_{\lang}(x_\secp)$,
		\end{itemize}
		we have
		$$
		\forall \secp \in \Nat : \abs{a(x_\secp, \rho_\secp) - b(x_\secp, \rho_\secp)}, \abs{b(x_\secp, \rho_\secp) - c(x_\secp, \rho_\secp, w_\secp)}, \abs{c(x_\secp, \rho_\secp, w_\secp) - a(x_\secp, \rho_\secp)} \leq \negl(\secp) \enspace .
		$$
	\end{corollary}
	
	\begin{proof}
		It immediately follows from Proposition \ref{claim:aborting_similarity} that the distance between $a(x, \rho)$ and $c(x, \rho, w)$ is negligible. By the same reasoning it follows from Proposition \ref{claim:non_aborting_similarity} that $b(x, \rho)$ and $c(x, \rho, w)$ are negligibly close.
		By triangle inequality it follows that also $a(x, \rho)$ and $b(x, \rho)$ are negligibly close.
	\end{proof}

	From the above it follows that the success probability of $\zkslS(x, \zkmV, \rho)$ is negligibly close to $\frac{1}{2}$, regardless of the quantum state $\rho$.
	\begin{corollary} [Success Probability of $\zkslS$ is Input-Oblivious] \label{cor:zkslS_success}
		For every quantum verifier $\zkmV = \{ \zkmV_\secp \}_{\secp \in \Nat}$ there exists a negligible function $\negl(\cdot)$ s.t. for every instance in the language $x = \{ x_\secp \}_{\secp \in \Nat}$ and quantum auxiliary input $\rho = \{ \rho_\secp \}_{\secp \in \Nat}$ for the verifier, we have
		$$
		\forall \secp \in \Nat : \abs{ \Pr\left[ \text{The simulation } \zkslS(x_\secp, \zkmV_\secp, \rho_\secp) \text{ succeeds} \right] - \frac{1}{2}} \leq \negl(\secp) \enspace .
		$$
	\end{corollary}
	
	\begin{proof}
		\begin{align*}
			\forall \secp \in \Nat : &\abs{ \Pr\left[ \text{The simulation } \zkslS(x_\secp, \zkmV_\secp, \rho_\secp) \text{ succeeds} \right] - \frac{1}{2}} \\
			&= \bigg| \frac{1}{2} \cdot \Pr\left[ \text{The simulation } \zkS_{\mathrm{a}}(x_\secp, \zkmV_\secp, \rho_\secp) \text{ succeeds} \right] \\
				&
				\qquad + \frac{1}{2} \cdot \Pr\left[ \text{The simulation } \zkS_{\mathrm{na}}(x_\secp, \zkmV_\secp, \rho_\secp) \text{ succeeds} \right] - \frac{1}{2}\bigg| \\
			&= \bigg| \frac{1}{2} \cdot a(x_\secp, \rho_\secp) + \frac{1}{2} \cdot \big( 1 - b(x_\secp, \rho_\secp) \big) - \frac{1}{2} \bigg| \\
			&= \frac{1}{2} \cdot \abs{ a(x_\secp, \rho_\secp) - b(x_\secp, \rho_\secp) } \leq \negl(\secp) \enspace ,
		\end{align*}
		where the last inequality is due to Corollary \ref{cor:probabilities}.
	\end{proof}
	
	\noindent We next prove that conditioned on succeeding, the output distribution of the simulator $\zkslS$ is indistinguishable from the real interaction.
	
	\begin{proposition} [The Output of a Successful $\zkslS$ is Indistinguishable from Real Interaction] \label{lem:conditional_success}
		Let $\zkmV = \{ \zkmV_\secp, \rho_\secp \}_{\secp \in \Nat}$ be a polynomial-size quantum verifier.
		For $x \in \lang$, let $\widetilde{\zkslS}(x, \zkmV, \rho)$ denote the conditional distribution of $\zkslS(x, \zkmV, \rho)$, conditioned on the simulation being successful.
		Then,
		$$
		\{ \view_{\zkmV}\prot{\zkP(w)}{\zkmV_\secp(\rho_{\secp})}(\ins)\}_{\substack{\secp\in\Nat,\\ \ins\in\lang\cap\zo^\secp,\\ w \in \mathcal{R}_{\lang}(x)}}
		\approx_{c}
		\{ \widetilde{\zkslS}(x, \zkmV_\secp, \rho_\secp) \}_{\substack{\secp\in\Nat,\\ \ins\in\lang\cap\zo^\secp,\\ w \in \mathcal{R}_{\lang}(x)}}\enspace.
		$$
	\end{proposition}
	\begin{proof}		
		Denote the following conditional distributions.
		\begin{itemize}
			\item
			$A_{\zkS} = \{ A_{\zkS, \secp} \}_{\secp\in \Nat}:$
			A conditional distribution of $\zkS_{\mathrm{a}}(x, \zkmV, \rho)$, conditioned on that the output is not $\zkFail$ (might be an empty distribution, if $a(x, \rho) = 0$).
			
			\item
			$S_{\zkS} = \{ S_{\zkS, \secp} \}_{\secp\in \Nat}:$
			A conditional distribution of $\zkS_{\mathrm{na}}(x, \zkmV, \rho)$, conditioned on that the output is not $\zkFail$ (might be an empty distribution, if $b(x, \rho) = 1$).
			
			\item $A_{\prot{\zkP}{\zkmV}} = \{ A_{\prot{\zkP}{\zkmV}, \secp} \}_{\secp\in \Nat}:$
			A conditional distribution of $\view_{\zkmV_{\mathrm{a}}}\prot{\zkP}{\zkmV}$ (from \ref{claim:aborting_similarity}), conditioned on that the output is not $\zkFail$ (might be an empty distribution, if $c(x, \rho, w) = 0$).
			
			\item $S_{\prot{\zkP}{\zkmV}} = \{ S_{\prot{\zkP}{\zkmV}, \secp} \}_{\secp\in \Nat}:$
			A conditional distribution of $\view_{\zkmV_{\mathrm{nm}}}\prot{\zkP}{\zkmV}$ (from \ref{claim:non_aborting_similarity}), conditioned on that the output is not $\zkFail$ (might be an empty distribution, if $c(x, \rho, w) = 1$).
		\end{itemize}
		
		Observe that the distribution $\widetilde{\zkslS}(x, \zkmV, \rho)$ is the distribution generated by outputting a sample from $A_{\zkS}$ with probability $\frac{a(x, \rho)}{1 + a(x, \rho) - b(x, \rho)}$, and a sample from $S_{\zkS}$ with probability $\frac{1 - b(x, \rho)}{1 + a(x, \rho) - b(x, \rho)}$.
		Additionally, observe that the distribution $\view_{\zkmV}\prot{\zkP(w)}{\zkmV(\rho)}(x)$ is the distribution generated by outputting a sample from $A_{\prot{\zkP}{\zkmV}}$ with probability $c(x, \rho, w)$ and from $S_{\prot{\zkP}{\zkmV}}$ with probability $1 - c(x, \rho, w)$.
		We will show that the two distributions are computationally indistinguishable by a hybrid argument.
		Consider the following distributions.
		\begin{itemize}
			\item $\Hyb_0:$
			The distribution $\widetilde{\zkslS}(x, \zkmV, \rho)$.
			
			\item $\Hyb_1:$
			Same as in $\Hyb_0$, with the exception that instead of sampling from $A_{\zkS}$ with probability $\frac{a(x, \rho)}{1 + a(x, \rho) - b(x, \rho)}$ (and from $S_{\zkS}$ with probability $\frac{1 - b(x, \rho)}{1 + a(x, \rho) - b(x, \rho)}$), it samples from $A_{\zkS}$ with probability $a(x, \rho)$ (and from $S_{\zkS}$ with probability $1 - a(x, \rho)$).
			
			\item $\Hyb_2:$
			Same as in $\Hyb_1$, but the probability $a(x, \rho)$ is changed to $c(x, \rho, w)$.
			
			\item $\Hyb_3:$
			Same as in $\Hyb_2$, with the exception that with probability $c(x, \rho, w)$, the process outputs a sample from $A_{\prot{\zkP}{\zkmV}}$ rather than from $A_{\zkS}$.
			
			\item $\Hyb_4:$
			Same as in $\Hyb_3$, with the exception that with probability $1 - c(x, \rho, w)$, the process outputs a sample from $S_{\prot{\zkP}{\zkmV}}$ rather than from $S_{\zkS}$.
			This process is exactly $\view_{\zkmV}\prot{\zkP(w)}{\zkmV(\rho)}(x)$.
		\end{itemize}		
		It is left to explain why each consecutive pair of distributions are computationally indistinguishable, and our proof is finished.
		For the following, define $a'(\secp) := a(x_\secp, \rho_\secp)$, $b'(\secp) := b(x_\secp, \rho_\secp)$, $c'(\secp) := c(x_\secp, \rho_\secp, w_\secp)$.
		\begin{itemize}
			\item $\Hyb_0 \approx_{s} \Hyb_1:$
			Due to the fact that $a'(\secp)$ and $b'(\secp)$ are negligibly close (Corollary \ref{cor:probabilities}), it follows that $a'(\secp)$ and $\frac{a(x, \rho)}{1 + a(x, \rho) - b(x, \rho)}$ are also negligibly close, and thus follows the statistical indistinguishability.
			
			\item $\Hyb_1 \approx_{s} \Hyb_2:$
			The probabilities $a'(\secp)$ and $c'(\secp)$ are negligibly close due to Corollary \ref{cor:probabilities}, and the statistical indistinguishability follows.
			
			\item $\Hyb_2 \approx_c \Hyb_3:$
			Assume toward contradiction that the indistinguishbility does not hold, this means there is a distinguisher $\Disting$, an infinite subset $Q \subseteq \Nat$ and a polynomial $p : \Nat \rightarrow \Nat$, s.t. for all $\secp \in Q$, $\Disting$ distinguishes with advantage at least $1/p(\secp)$ between $\Hyb_{2, \secp}$ and $\Hyb_{3, \secp}$.
			We consider two cases for the function $c'$, and show that in both of them the contradiction follows from \ref{claim:aborting_similarity}.
			\begin{itemize}
				\item {\bf Case 1:} for every polynomial function $q : \Nat \rightarrow \Nat$, there are only finitely-many $\secp \in Q$ s.t. $1 - c'(\secp) > 1/q(\secp)$.
				This means that there is a negligible function $\mu$ s.t. $\forall \secp \in \Nat : 1 - c'(\secp) \leq \mu(\secp)$.
				In that case, the contradiction follows directly from Proposition \ref{claim:aborting_similarity}, because for indices $\secp \in Q$, a sample from $\zkS_{\mathrm{a}}(x_\secp, \zkmV_\secp, \rho_\secp)$ is statistically indistinguishable from a sample from $\Hyb_{2, \secp}$, and a sample from $\view_{\zkmV_{\mathrm{a}}}\prot{\zkP(w_\secp)}{\zkmV_\secp(\rho_\secp)}(\ins_\secp)$ is statistically indistinguishable from a sample from $\Hyb_{3, \secp}$.
				
				\item {\bf Case 2:} there is a polynomial function $q' : \Nat \rightarrow \Nat$, s.t. there are infinitely-many $\secp \in Q$ s.t. $1 - c'(\secp) > 1/q'(\secp)$, we denote this infinite set of indices by $Q'$.
				For these indices we can violate the indistinguishbility from \ref{claim:aborting_similarity}.
				More specifically, for $\secp \in Q'$ we can sample in polynomial time (say, $q'(\secp)^2$) and using polynomial-size quantum advice, from a distribution that is statistically indistinguishable from $S_{\zkS, \secp}$, and reduce distinguishing between $\zkS_{\mathrm{a}}(x_\secp, \zkmV_\secp, \rho_\secp)$ and $\view_{\zkmV_{\mathrm{a}}}\prot{\zkP(w_\secp)}{\zkmV_\secp(\rho_\secp)}(\ins_\secp)$, to distinguishing between $\Hyb_{2, \secp}$ and $\Hyb_{3, \secp}$ in the following way.
				
				\noindent When getting a sample from either $\zkS_{\mathrm{a}}(x_\secp, \zkmV_\secp, \rho_\secp)$ or $\view_{\zkmV_{\mathrm{a}}}\prot{\zkP(w_\secp)}{\zkmV_\secp(\rho_{\secp})}(\ins_\secp)$, if the sample's value was $\zkFail$, approximately sample (as mentioned above, in time $q(\secp)^2$) from $S_{\zkS, \secp}$.
				This can be done, for example, by using a polynomial amount of copies (i.e. $q'(\secp)^2$) of the quantum advice $\rho$ of the verifier.
				This output of the reduction (whether it was $\zkFail$ that was swapped to a sample that is close to $S_{\zkS, \secp}$, or whether it was a non-$\zkFail$ and was not swapped) is sent to the distinguisher $\Disting$.
				Due to the fact that for the cases we got $\zkFail$, the generated sample is statistically indistinguishable from $S_{\zkS, \secp}$, it follows that when we get a sample from $\zkS_{\mathrm{a}}(x_\secp, \zkmV_\secp, \rho_\secp)$ then the output sample of our reduction is statistically close to $\Hyb_{2, \secp}$, and when we get a sample from $\view_{\zkmV_{\mathrm{a}}}\prot{\zkP(w_\secp)}{\zkmV_\secp(\rho_{\secp})}(\ins_\secp)$ then the output sample of our reduction is statistically close to $\Hyb_{3, \secp}$, and we get a contradiction.
			\end{itemize}
			
			\item $\Hyb_3 \approx_c \Hyb_4:$
			Assume toward contradiction that the indistinguishbility does not hold, this means there is a distinguisher $\Disting$, an infinite subset $Q \subseteq \Nat$ and a polynomial $p : \Nat \rightarrow \Nat$, s.t. for all $\secp \in Q$, $\Disting$ distinguishes with advantage at least $1/p(\secp)$ between $\Hyb_{3, \secp}$ and $\Hyb_{4, \secp}$.
			We consider two cases for the function $c'$, and show that in both of them the contradiction follows from \ref{claim:non_aborting_similarity}.
			\begin{itemize}
				\item {\bf Case 1:} for every polynomial function $q : \Nat \rightarrow \Nat$, there are only finitely-many $\secp \in Q$ s.t. $c'(\secp) > 1/q(\secp)$.
				This means that there is a negligible function $\mu$ s.t. $\forall \secp \in \Nat : c'(\secp) \leq \mu(\secp)$.
				In that case, the contradiction follows directly from Proposition \ref{claim:non_aborting_similarity}, because for indices $\secp \in Q$, a sample from $\zkS_{\mathrm{na}}(x_\secp, \zkmV_\secp, \rho_\secp)$ is statistically indistinguishable from a sample from $\Hyb_{3, \secp}$, and a sample from $\view_{\zkmV_{\mathrm{na}}}\prot{\zkP(w_\secp)}{\zkmV_\secp(\rho_\secp)}(\ins_\secp)$ is statistically indistinguishable from a sample from $\Hyb_{4, \secp}$.
				
				\item {\bf Case 2:} there is a polynomial function $q' : \Nat \rightarrow \Nat$, s.t. there are infinitely-many $\secp \in Q$ s.t. $c'(\secp) > 1/q'(\secp)$, we denote this infinite set of indices by $Q'$.
				For these indices we can violate the indistinguishbility from \ref{claim:non_aborting_similarity}.
				More specifically, for $\secp \in Q'$ we can sample, in polynomial time (say, $q'(\secp)^2$) and using polynomial-size quantum advice, from a distribution that is statistically indistinguishable from $A_{\zkS, \secp}$, and reduce distinguishing between $\zkS_{\mathrm{na}}(x_\secp, \zkmV_\secp, \rho_\secp)$ and $\view_{\zkmV_{\mathrm{na}}}\prot{\zkP(w_\secp)}{\zkmV_\secp(\rho_\secp)}(\ins_\secp)$, to distinguishing between $\Hyb_{3, \secp}$ and $\Hyb_{4, \secp}$ in the following way.
				
				\noindent When getting a sample from either $\zkS_{\mathrm{na}}(x_\secp, \zkmV_\secp, \rho_\secp)$ or $\view_{\zkmV_{\mathrm{na}}}\prot{\zkP(w_\secp)}{\zkmV_\secp(\rho_{\secp})}(\ins_\secp)$, if the sample's value was $\zkFail$, approximately sample (as mentioned above, in time $q(\secp)^2$) from $A_{\zkS, \secp}$.
				This can be done, for example, by using a polynomial amount of copies (i.e. $q'(\secp)^2$) of the quantum advice $\rho$ of the verifier.
				This output of the reduction (whether it was $\zkFail$ that was swapped to a sample that is close to $A_{\zkS, \secp}$, or whether it was a non-$\zkFail$ and was not swapped) is sent to the distinguisher $\Disting$.
				Due to the fact that for the cases we got $\zkFail$, the generated sample is statistically indistinguishable from $A_{\zkS, \secp}$, it follows that when we get a sample from $\zkS_{\mathrm{na}}(x_\secp, \zkmV_\secp, \rho_\secp)$ then the output sample of our reduction is statistically close to $\Hyb_{3, \secp}$, and when we get a sample from $\view_{\zkmV_{\mathrm{na}}}\prot{\zkP(w_\secp)}{\zkmV_\secp(\rho_{\secp})}(\ins_\secp)$ then the output sample of our reduction is statistically close to $\Hyb_{4, \secp}$, and we get a contradiction.
			\end{itemize}
		\end{itemize}
	\end{proof}
	
	We conclude with proving that the output of the simulation $\zkS(x, \zkmV, \rho)$ is indeed computationally indistinguishable from the output of the real interaction $\view_{\zkV}\prot{\zkP}{\zkmV(\rho)}(x)$.
	
	\begin{proposition} [Simulation Output is Indistinguishable from Interaction] \label{lem:simulation_indistinguishable}
		For any quantum polynomial-size verifier $\zkmV = \set{\zkmV_\secp, \rho_\secp}_{\secp \in \Nat}$,
		$$
		\{ \view_{\zkmV}\prot{\zkP(w)}{\zkmV_\secp(\rho_{\secp})}(\ins)\}_{\substack{\secp\in\Nat,\\ \ins\in\lang\cap\zo^\secp,\\ w \in \mathcal{R}_{\lang}(x)}}
		\approx_{c}
		\{\zkS(\ins, \zkmV_\secp, \rho_\secp)\}_{\substack{\secp\in\Nat,\\ \ins\in\lang\cap\zo^\secp,\\ w \in \mathcal{R}_{\lang}(x)}} \enspace.
		$$
	\end{proposition}

	\begin{proof}
		Let $\zkmV = \set{\zkmV_\secp}_{\secp \in \Nat}$ be a quantum polynomial-size verifier.
		According to Corollary \ref{cor:zkslS_success}, there is a negligible function $\negl(\cdot)$ s.t. for every instance in the language $x = \{ x_\secp \}_{\secp \in \Nat}$ and auxiliary input quantum state $\rho = \{ \rho_\secp \}_{\secp \in \Nat}$ for the verifier, we have
		$$
		\forall \secp \in \Nat : \abs{ \Pr\left[ \text{The simulation } \zkslS(x_\secp, \zkmV_\secp, \rho_\secp) \text{ succeeds} \right] - \frac{1}{2}} \leq \negl(\secp) \enspace .
		$$		
		Consider the quantum circuit $\zkslSmV$, which is the circuit implementation of $\zkslS$ with hardwired inputs $x$ and $\zkmV$, that gets as input only the quantum state $\rho$.
		As mentioned above, the success probability of $\zkslSmV$ is negligibly close $\frac{1}{2}$, for \emph{any} quantum state $\rho$.
		If we denote the success probability for input $\rho$ by $p(\rho)$ and denote $\varepsilon := \negl(\secp) + 2^{-\secp\cdot\frac{3}{4}}$, $p_0 := \frac{1}{4}$ and $q := \frac{1}{2}$, we can see that the 4 conditions for the Quantum Rewinding Lemma \ref{lem:watrous} are satisfied:
		\begin{itemize}
			\item
			$\secp \geq \frac{\log(1/\varepsilon)}{4\cdot p_0 (1 - p_0)}$.
			
			\item
			For every state $\rho$, $p_0 \leq p(\rho)$.
			
			\item
			For every state $\rho$, $\abs{p(\rho) - q} < \varepsilon$.
			
			\item
			$p_0(1 - p_0) \leq q(1 - q)$.
		\end{itemize}
		This implies that $\Wat(\zkslSmV, \rho, \secp)$ has trace distance bounded by $4\sqrt{\varepsilon}\frac{\log(1/\varepsilon)}{p_0 (1 - p_0)}$ from the success-conditioned output distribution of $\zkslS(x, \zkmV, \rho)$. Since $\varepsilon$ is a negligible function of $\secp$, so is $4\sqrt{\varepsilon}\frac{\log(1/\varepsilon)}{p_0 (1 - p_0)}$.
		
		Finally, recall that $\zkS(x, \zkmV, \rho) = \Wat(\zkslSmV, \rho, \secp)$, and that proposition \ref{lem:conditional_success} says that the success-conditioned distribution of $\zkslS(x, \zkmV, \rho)$ is computationally indistinguishable from \newline$\view_{\zkmV}\prot{\zkP(\wit)}{\zkmV(\rho)}(\ins)$, and our proof is concluded.
	\end{proof}

	\begin{remark} [Classical Universal Simulator for Classical Verifiers]
		As a side note, we observe that the protocol preserves the trait of classical ZK, that is, classical verifiers learn nothing from the protocol (formally, for classical verifiers there is a classical simulator).
		A classical simulator showing this will simply execute $\zkslS(x, \zkmV)$ repeatedly some polynomial number of times (either until it succeeds in one of the tries, or fails in all and then the output is $\zkFail$), specifically, $\secp$ tries will do.
		Since the probability for $\zkslS$ to succeed is $\approx \frac{1}{2}$, the probability to successfully sample from the success-conditioned distribution is overwhelming, and thus the output of the simulator is indistinguishable from the output of the verifier in the real interaction.
	\end{remark}
	\nir{this is a distraction! if you want can add a remark after the proof.}
	\omri{How about now?}

%% file: quantum_extractable_commitments.tex
\section{Quantumly-Extractable Classical Commitments}
\nir{stopped putting comments, until we change}

	\nir{The above text shouldn't be here. Instead, have an explicit remark the definition.}\omri{ok}
	
	In this section we show how to use any constant-round post-quantum zero-knowledge argument for NP (and standard cryptographic assumptions) in order to construct a constant-round, quantumly-extractable classical commitment scheme.
	\nir{remove "as an independent cryptographic primitive"}\omri{ok}
	We start with the definition, and proceed to the construction.
	
	\begin{definition}[Quantumly-Extractable Commitment] \label{def:q_ex_commit}
		A quantumly-extractable commitment scheme consists of three interactive \PPT algorithms $(\comS, \comR, \decom)$ with the following syntax.
		
		\begin{itemize}
			\item $\comS(1^\secp, m):$
			The sender algorithm gets as input the public security parameter $1^\secp$ and the secret message $m$ to commit to.
			
			\item $\comR(1^\secp):$
			The receiver algorithm gets only the public security parameter $1^\secp$.
			
			\item The algorithms $\comS, \comR$ interact \nir{remove "in a constant number of rounds", irrelevant to the definition} and generate transcript $T$.\nir{remove "to generate a transcript" (this is not "the purpose" of the interaction)}\omri{ok}
			
			\item $\decom(T, m, r):$
			For a transcript, message and randomness, the decommitment verification algorithm outputs a bit.
			\nir{Decommitment usually refers to the process of sending the sender's randomness. Perhaps call it VDcom or something. Also, the text is not clear --- just say it outputs a bit, otherwise people may think that  "the verification or rejection of the decommitment is a thing" (like i did). Can remove the security parameter (w.l.o.g it's explicit in $T$)}\omri{done}
		\end{itemize}
		
		The scheme satisfies the following conditions.
		
		\begin{itemize}
			\item {\bf Perfect Binding:} Let $m_0, m_1, r_0, r_1 \in \{ 0, 1 \}^*$, and let $T$ be a transcript.
			If $\decom(T, m_0, r_0) = \decom(T, m_1, r_1) = 1$, then $m_0 = m_1$.
			Accordingly, for a transcript $T$ denote by $m_T$ the (unique) string such that if there exist $r$ s.t. $\decom(T, m, r) = 1$, then $m_T := m$, and $m_T := \bot$ otherwise. 
			
			\item {\bf Computational Hiding:} For every polynomial-size quantum receiver $\commR = \{ \commR_\secp, \rho_\secp \}_{\secp \in \Nat}$ and polynomial $\ell(\cdot)$, 
			$$
			\{
			\view_{\commR_\secp}\prot{\comS(m_0)}{\commR_{\secp}(\rho_\secp)}(1^\secp)
			\}_{\secp, m_0, m_1}
			\approx_{c}
			\{
			\view_{\commR_\secp}\prot{\comS(m_1)}{\commR_{\secp}(\rho_\secp)}(1^\secp)
			\}_{\secp, m_0, m_1} \enspace ,
			$$
			where $\secp \in \Nat$, $m_0, m_1 \in \{ 0, 1 \}^{\ell(\secp)}$.

\nir{VIEW wasn't defined. Can just use OUT.}\omri{ok}\omri{I know we can also lose the receiver and just take a polynomial $s(\cdot)$ and consider size-$s(\secp)$ adversaries, but I think this is clearer.}
			
\nir{Remove density}\omri{done}
			
			\item {\bf Extractability:} There exists a quantum polynomial-time algorithm $\Ext$ s.t. for every polynomial-size quantum sender $\commS = \{ \commS_\secp, \rho_\secp \}_{\secp \in \Nat}$ outputs a quantum state $\sigma_\Ext$ and message $m_\Ext$, with the following guarantee.
			$$
			\Bigl\{
			(\sigma, m_T) \; | \; (T, \sigma, m_T) \gets \prot{\commS_\secp(\rho_\secp)}{\comR}(1^\secp)
			\Bigr\}_{\secp \in \Nat}
			$$
			$$
			\approx_{c}
			\Bigl\{ (\sigma_\Ext, m_\Ext) \; | \; (\sigma_\Ext, m_\Ext) \gets \Ext(1^\secp, \commS_\secp, \rho_\secp) \Bigr\}_{\secp \in \Nat} \enspace ,
			$$
			where $\sigma$ is the inner state of $\commS$ after executing the interaction with $\comR$.
			\nir{remove $T$}
			\nir{remove validity}
				
			\end{itemize}
			
			\omri{Add at the end that for classical senders there is a classical extractor, this is the same extractor, only that in the ZK simulation inside the extraction, it uses the classical version of the ZK simulator of the argument system.}
				
	\end{definition}

	\begin{remark} \label{remark:extractable_commitment}
		In the standard definition of extraction and more broadly, of simulation, the simulator does not output the interaction transcript (in classical-interaction protocols).
		It is noted however that it can be assumed without the loss of generality that the simulator also outputs the simulated transcript whenever needed.
		This is because, given a classical (or quantum) interactive circuit, it can be compiled in polynomial time (in the circuit size) to a circuit with identical functionality, that records the interaction transcript into its private inner state. Since the simulator simulates the inner state of the adversary at the end of interaction it in particular simulates the transcript.
	\end{remark}
	
	\medskip\noindent
	We describe the protocol between $\comS$ and $\comR$ in \figref{fig:q_ex_commitment}.
	
	\paragraph{Ingredients and notation:}
	\begin{itemize}
		\item
		A non-interactive commitment scheme $\Com$.
		\item
		A 2-message function-hiding secure function evaluation scheme $(\sfeGen,$ $\sfeEnc,$ $\sfeEval,$ $\sfeD)$.
		\item
		A constant-round post-quantum zero-knowledge argument system $(\zkP_{\text{NP}}, \zkV_{\text{NP}})$ for NP.
	\end{itemize}
	
%
	
	\protocol
	{\proref{fig:q_ex_commitment}}
	{A Quantumly-Extractable Classical Commitment Scheme.}
	{fig:q_ex_commitment}
	{
		\begin{description}
			\item[Common Input:] A security parameter $\secp \in \Nat$.
			\item[Private Input of $\comS$:] A message $m \in \{0, 1\}^*$ to commit to.
		\end{description}
		
		\begin{enumerate}
			\item {\bf Commitment by $\comS$:} \label{com:Sen_commit}
			$\comS$ sends a commitment to $m$, $\cmt_\comS \gets \Com(1^\secp, m)$.
			
			\item {\bf Commitment by $\comR$:} \label{com:Rec_commit}
			$\comR$ sends a commitment to $0$, $\cmt_\comR \gets \Com(1^\secp, 0)$.
			
			\item {\bf ZK Argument by $\comR$:} \label{com:Rec_argument}
			$\comR$ interacts with $\comS$ through $(\zkP_{\text{NP}}, \zkV_{\text{NP}})$ to give a ZK argument that $\cmt_\comR$ is indeed a commitment to $0$, that is, there exists randomness $r_0 \in \{ 0, 1 \}^{\poly(\secp, 1)}$ string\footnote{Let $\poly(\secp, \ell)$ denote the polynomial that represents the amount of randomness the commitment algorithm $\Com(\cdot)$ needs for security parameter $\secp$ and message length $\ell$.} s.t. $\cmt_\comR = \Com(1^\secp, 0; r_0)$.
			
			\item {\bf $\comS$ Challenges $\comR$:} \label{com:Sen_challenge}
			The parties interact so that $\comS$ can offer to send $m$ if $\comR$ managed to trick $\comS$ in the ZK argument.
			\begin{enumerate}
				\item \label{com:Rec_challenge}
				$\comR$ computes $\sfedk \gets \sfeGen(1^\secp)$ and sends $\ciph_{\comR} \gets \sfeEnc_{\sfedk}(0^{\poly(\secp, 1)})$.
				
				\item \label{com:Sen_response}
				$\comS$ sends $\evciph \gets \sfeEval\Big( C_{1\rightarrow m}, \ciph_{\comR} \Big)$, where $C_{1\rightarrow m}$ is the (canonical) circuit that for input $r_1 \in \{ 0, 1 \}^{\poly(\secp, 1)}$ s.t. $\cmt_\comR = \Com(1^\secp, 1; r_1)$, outputs $m$, and for any other input outputs $\bot$.
			\end{enumerate}
			
			\item {\bf ZK Argument by $\comS$:} \label{com:Sen_argument}
			$\comS$ interacts with $\comR$ through $(\zkP_{\text{NP}}, \zkV_{\text{NP}})$ to give a ZK argument for the statement that its transcript until the end of step \ref{com:Sen_response} is consistent, that is, there exists a message and randomness for the honest sender algorithm $\comS$ that generates the transcript.
			
		\end{enumerate}
	}
	
	\paragraph{Decommitment Verification.}
	On input $(T, m, r)$ the decommitment verification algorithm $\decom$ deduces the security parameter $\secp$ (the security parameter is public and can be assumed to be part of the transript).
	It then checks two things:
	\begin{itemize}
		\item The argument that $\comS$ gave at the last step of the transcript $T$ is convincing (this is possible as the argument is publicly verifiable). 
		
		\item The commitment $\cmt_\comS$ from step \ref{com:Sen_commit} in the transcript $T$ indeed decommits to $m, r$ (i.e. $\Com(1^\secp, m;r) = \cmt_\comS$).
	\end{itemize}
	The output is $1$ iff the check succeeds.
	
	\paragraph{Binding and Hiding.}
	The perfect binding property of the scheme follows readily from the perfect binding of the non-interactive commitment scheme $\Com$.
	We next show hiding.
	
	\begin{proposition} [The Commitment Scheme is Computationally Hiding] \label{lem:com_hiding}
		For every polynomial-size quantum receiver $\commR = \{ \commR_\secp, \rho_\secp \}_{\secp \in \Nat}$ and polynomial $\ell(\cdot)$, 
		$$
		\{
		\view_{\commR_\secp}\prot{\comS(m_0)}{\commR_{\secp}(\rho_\secp)}(1^\secp)
		\}_{\secp, m_0, m_1}
		\approx_{c}
		\{
		\view_{\commR_\secp}\prot{\comS(m_1)}{\commR_{\secp}(\rho_\secp)}(1^\secp)
		\}_{\secp, m_0, m_1} \enspace ,
		$$
		where $\secp \in \Nat$, $m_0, m_1 \in \{ 0, 1 \}^{\ell(\secp)}$.
	\end{proposition}

	\begin{proof}	
		We prove the claim by a hybrid argument.
		Define the following hybrid distributions on transcripts.
		\begin{itemize}
			\item $\Hyb_0:$
			This is the output distribution $\protView_{\commR}\prot{\comS(m_{0})}{\commR(\rho)}$.
			
			\item $\Hyb_1:$
			The output distribution of a process that acts like $\Hyb_0$, with the exception that in step \ref{com:Sen_argument}, instead of $\comS$ communicating with $\commR$ to give a ZK argument, we take the ZK simulator $\zkS$ of the argument system $(\zkP_{\text{NP}}, \zkV_{\text{NP}})$ and use it to simulate the argument by $\comS$, by executing $\zkS(T', \commR, \rho')$, where $T'$ (resp. $\rho'$) is the transcript (resp. inner quantum state of $\commR$) generated at the end of step \ref{com:Sen_challenge} of the interaction.
			
			\item
			$\Hyb_2:$
			The output distribution of a process that acts like $\Hyb_1$, with the exception that in step \ref{com:Sen_response}, instead of actually performing an SFE evaluation of $C_{1\rightarrow m_0}$, the process performs an SFE evaluation of the circuit $C_\bot$ that always outputs $\bot$.
			
			\item
			$\Hyb_3:$
			The output distribution of a process that acts like $\Hyb_2$, with the exception that in step \ref{com:Sen_commit}, instead of committing to $m_0$, the sender commits to $m_1$.
			
			\item
			$\Hyb_4:$
			The output distribution of a process that acts like $\Hyb_3$, with the exception that in step \ref{com:Sen_response}, the process performs an SFE evaluation of the circuit $C_{1\rightarrow m_1}$, and not of the circuit $C_\bot$.
			
			\item
			$\Hyb_5:$
			The output distribution of a process that acts like $\Hyb_4$, with the exception that in step \ref{com:Sen_argument}, instead of using the ZK simulator for the sender's argument, the process uses the ZK argument regularly, that is, the sender proves that the transcript so far is consistent.
			Observe that this is exactly the output distribution $\protView_{\commR}\prot{\comS(m_{1})}{\commR(\rho)}$.
		\end{itemize}
		
		We now explain why each consecutive pair of distributions are computationally indistinguishable, and our proof is finished.
		\begin{itemize}
			\item
			$\Hyb_0 \approx_c \Hyb_1:$
			Follows from the post-quantum zero-knowledge property of the protocol $(\zkP_{\text{NP}}, \zkV_{\text{NP}})$.
			
			\item
			$\Hyb_1 \approx_s \Hyb_2:$
			Assume toward contradiction that the two distributions are distinguishable, and fix, by an averaging argument, the partial transcript $T'$ (and inner state of $\commR$) that is generated at the end of step \ref{com:Rec_commit} of the protocol and maximizes the distinguishing advantage between the two distributions.
			
			We consider two cases for the commitment $\cmt_\comR$ in the transcript $T'$: The simpler case is if $\cmt_\comR$ is not a commitment to 0 (i.e. there is no $r_0 \in \{ 0, 1 \}^*$ s.t. $\cmt_\comR = \Com(1^\secp, 0; r_0)$), in that case, by the soundness of the argument that $\commR$ gives in step \ref{com:Rec_argument}, with overwhelming probability $\comS$ is going to reject the proof and end communication, and only with a negligible probability the process continues to a point where the two processes $\Hyb_1$, $\Hyb_2$ differ, in contradiction.
			
			In the second case $\cmt_\comR$ is a valid commitment to $0$.
			In that case, the contradiction follows from (an implication of) the circuit privacy property of the SFE encryption, specifically, it follows from Claim \ref{claim:SFE_eval}.
			From the perfect binding of the non-interactive commitment scheme $\Com$, there is no string $r_1$ s.t. $\cmt_\comR = \Com(1^\secp, 1; r_1)$, which in turn implies that the circuit $C_{1 \rightarrow m_0}$ is identical in functionality to the circuit $C_\bot$ that outputs $\bot$ on any input.
			By Claim \ref{claim:SFE_eval} it follows that the responses from $\comS$ in step \ref{com:Sen_response} are statistically indistinguishable, and thus also the distributions $\Hyb_1$ and $\Hyb_2$, again in contradiction.
			
			\item
			$\Hyb_2 \approx_c \Hyb_3:$
			Follows from the hiding of the commitment $\cmt_\comS$ that $\comS$ gives in step \ref{com:Sen_commit}, that is, the hiding property of the commitment scheme $\Com$.
			
			\item
			$\Hyb_3 \approx_s \Hyb_4:$
			This indistinguishability follows from the exact same reasoning as in the explanation for the indistinguishability $\Hyb_1 \approx_s \Hyb_2$, by swapping $m_0$ with $m_1$ in the explanation.
			
			\item
			$\Hyb_4 \approx_c \Hyb_5:$
			Like the indistinguishability $\Hyb_0 \approx_c \Hyb_1$, this indistinguishability follows again from the zero-knowledge property of the argument system $(\zkP_{\text{NP}}, \zkV_{\text{NP}})$.
		\end{itemize}
	\end{proof}

	\subsection{Extractability}
	We show a quantum polynomial-time extractor $\Ext$ s.t. for any polynomial-size quantum sender $\commS = \{ \commS_\secp, \rho_\secp \}_{\secp \in \Nat}$ extracts the sender's committed message (if it exists) and also simulates its quantum state at the end of the protocol.
	
	
	\medskip
	\paragraph{$\Ext(1^\secp, \commS, \rho):$}
	\begin{enumerate}
		\item {\bf Simulation of Commitments:} \label{com_simulation:Commitments}
		$\commS$ outputs $\cmt_{\comS}$.
		$\Ext$ then sends to $\commS$ a commitment to 1: $\cmt_{\Ext} = \Com(1^{\secp}, 1; r_1)$, where $r_1 \in \{ 0, 1 \}^{\poly(\secp, 1)}$ is the random string used as the randomness of the commitment algorithm.
		
		\item {\bf Simulation of ZK Argument by $\comR$:} \label{com_simulation:Rec_argument}
		$\Ext$ uses the zero-knowledge simulator $\zkS$ of the argument system $(\zkP_{\text{NP}}, \zkV_{\text{NP}})$.
		$\Ext$ executes $\zkS(\cmt_{\Ext}, \commS, \rho^{(1)})$ to simulate the argument that $\comR$ gives to $\commS$ at step \ref{com:Rec_argument} of the protocol ($\rho^{(1)}$ is the inner state of $\commS$ after step \ref{com_simulation:Commitments} of the extraction).
		At the end of the zero-knowledge simulation, we have a simulated argument transcript and a quantum state $\rho'$ for $\commS$ to carry on to the next step of extraction.
		
		\item {\bf Extraction of Message from $\commS$:} \label{com_simulation:Sen_extraction}
		\begin{itemize}
			\item
			$\Ext$ computes $\sfedk \gets \sfeGen(1^\secp)$ and sends $\ciph_{\Ext}\gets \sfeEnc_{\sfedk}(r_1)$.
			
			\item
			$\commS$ outputs a response $\evciph$.
		\end{itemize}
		$\Ext$ then decrypts the evaluated ciphertext to get a message $m'$.
		
		\item {\bf ZK Argument by $\commS$:} \label{com_simulation:Sen_argument}
		$\Ext$ takes the role of the honest receiver $\comR$ in the ZK argument $\commS$ gives.
		
		\item {\bf Extraction Procedure Output:} \label{com_simulation:output}
		The output $(\sigma_{\Ext}, m_{\Ext})$ of the extraction procedure is as follows.
		\begin{itemize}
			\item
			The simulated inner state $\sigma_{\Ext}$ for the sender is set to be the inner state of $\commS$ at the time of halting of the procedure.
			
			\item
			If the argument that $\commS$ gave in step \ref{com_simulation:Sen_argument} of the procedure is convincing then $m_{\Ext} = m'$, otherwise $m_{\Ext} = \bot$.
		\end{itemize}
	\end{enumerate}	
	
	\medskip
	It remains to explain why the extraction process yields an output that is computationally indistinguishable from a tuple $(T, \sigma, m_T)$ generated by the real interaction between $\commS(\rho)$ and $\comR$, and also that the extracted message $m_{\Ext}$ is indeed the message that $T_{\Ext}$ can be decommitted to.
	
	\begin{proposition} \label{lem:ext_successful}
		Let $\commS = \{ \commS_\secp, \rho_\secp \}_{\secp \in \Nat}$ be a polynomial-size quantum sender, then,
		$$
		\Bigl\{
		(\sigma, m_T) \; | \; (T, \sigma, m_T) \gets \prot{\commS_\secp(\rho_\secp)}{\comR}(1^\secp)
		\Bigr\}_{\secp \in \Nat}
		\approx_{c}
		\Bigl\{ \Ext(1^\secp, \commS_\secp, \rho_\secp) \Bigr\}_{\secp \in \Nat} \enspace .
		$$
	\end{proposition}
	
	\begin{proof}
		We prove the claim by a hybrid argument.
		Define the following hybrid processes:
		\begin{itemize}
			\item $\Hyb_0:$
			This distribution is the output distribution $(\sigma, m_T)$ of the real interaction $\prot{\commS(\rho)}{\comR}$.
			
			\item $\Hyb_1:$
			The output distribution of a process that acts like $\Hyb_0$, with the exception that in step \ref{com:Rec_argument}, instead of $\comR$ communicating with $\commS$ to give a ZK argument, we take the ZK simulator $\zkS$ of the argument system $(\zkP_{\text{NP}}, \zkV_{\text{NP}})$ and use it to simulate the argument by $\comR$ by executing $\zkS(T', \commS, \rho')$, where $T'$ (resp. $\rho'$) is the transcript (resp. inner quantum state of $\commR$) generated at the end of step \ref{com:Rec_commit} of the interaction.
			
			\item $\Hyb_2:$
			The output distribution of a process that acts like $\Hyb_1$, with the exception that when $\comR$ sends $\cmt_\comR$, it commits to $1$ instead of to $0$.
			
			\item $\Hyb_3:$
			The output distribution of a process that acts like $\Hyb_2$, with the exception that in step \ref{com:Rec_challenge}, $\comR$ sends an SFE encryption $\ciph_{\comR}$ of the randomness $r_1$ that it used in step \ref{com:Rec_commit} when it committed for $1$.
			Note that this output distribution is identical to the extraction's output $\Ext(1^\secp, \commS, \rho)$, with the only change being that $m_T$ is generated as in $\prot{\commS(\rho)}{\comR}$.
			
			\item $\Hyb_4:$
			The output distribution of a process that acts like $\Hyb_3$, with the exception that the output message $m_T$ is generated differently, specifically, $m_T$ is $m_{\Ext}$ that is generated as in step \ref{com_simulation:output} of the extraction procedure.
			Note that this process is exactly the output distribution $(\sigma_{\Ext}, m_{\Ext}) \gets \Ext(1^\secp, \commS, \rho)$.
		\end{itemize}
		
		We now explain why each pair of consecutive distributions are computationally indistinguishable, and our proof is finished.
		\begin{itemize}
			\item $\Hyb_0 \approx_c \Hyb_1:$
			Assume toward contradiction that the two distributions are distinguishable, and fix, by an averaging argument, the partial transcript $T'$ and inner quantum state $\sigma'$ of $\commS$ generated at the end of step \ref{com_simulation:Commitments} of the simulation, that maximizes the distinguishing advantage of the two distributions.
			Inside such transcript $T'$ we consider the sender commitment $\cmt_\comS$, and the (unique, by the perfect binding of the commitment scheme $\Com$) message $m_{T'}$ that is inside this commitment (if the commitment cannot be opened to any message, $m_{T'} := \bot$).
			
			From our assumption that $\Hyb_0$, $\Hyb_1$ are distinguishable, follows the existence of a distinguisher that breaks the zero-knowledge property of $(\zkP_{\text{NP}}, \zkV_{\text{NP}})$.
			Specifically, the distinguisher uses as non-uniform advice the partial transcript $T'$ and the message $m_{T'}$, gets either a real interaction transcript or a simulation of the argument that $\comR$ gives in step \ref{com:Rec_argument} of the protocol, then executes the rest of the commitment protocol, and uses the knowledge $m_{T'}$ at the end of protocol execution to output $m_T$.
			It follows that such distinguisher breaks the zero knowledge property of $(\zkP_{\text{NP}}, \zkV_{\text{NP}})$, in contradiction.
			
			In the following explanations for the indistinguishabilities we will use the same averaging argument and non-uniform advice that includes the message $m_{T'}$, and refer to it simply as the "averaging argument with non-uniform advice message".
			
			\item $\Hyb_1 \approx_c \Hyb_2:$
			Follows from the same averaging argument and non-uniform advice message reasoning from the proof of $\Hyb_0 \approx_c \Hyb_1$, along with the hiding of the commitment scheme $\Com$.
			
			\item $\Hyb_2 \approx_c \Hyb_3:$
			Follows from the same averaging argument and non-uniform advice message reasoning from the proof of $\Hyb_0 \approx_c \Hyb_1$, along with the input privacy (encryption security) property of the SFE encryption.
			
			\item $\Hyb_3 \approx_s \Hyb_4:$
			Recall that both processes $\Hyb_3$, $\Hyb_4$ generate the output state $\sigma_{\Ext}$ as in the extraction procedure, but differ only in the way they generate the output message.
			Assume toward contradiction that $\Hyb_3$, $\Hyb_4$ are distinguishable and fix, by an averaging argument, the partial transcript $T'$ (and inner state $\sigma'$ of $\commS$) generated at the end of step \ref{com_simulation:Sen_extraction} of the extraction.
			
			Consider two cases for the partial transcript $T'$.
			\begin{itemize}
				\item
				$T'$ is consistent.
				In that case, by the (perfect) correctness of the SFE evaluation it follows that the generated messages $m_T$ (from $\Hyb_3$) and $m_{\Ext}$ (from $\Hyb_4$) are identical. The rest of the protocol execution, which includes only the argument by $\commS$, is also identical between the two distributions. The distinguisher between $\Hyb_3$, $\Hyb_4$ (we assumed toward contradiction exists) cannot distinguish between these two distributions as they are identical, in contradiction.
				
				\item
				$T'$ is not consistent.
				In that case, by the soundness of the argument system $(\zkP_{\text{NP}}, \zkV_{\text{NP}})$, the argument by $\commS$ fails with overwhelming probability, and with the same probability the values of both $m_T$ (from $\Hyb_3$) and $m_{\Ext}$ (from $\Hyb_4$) are $\bot$.
				It follows that the statistical distance between the distributions is negligible, in contradiction.
			\end{itemize}
		\end{itemize}
	\end{proof}

%% file: constant_round_qma.tex
\section{Constant-Round Zero-Knowledge Quantum Arguments for QMA} \label{sec:qma}
	In this section we explain how the tools from previous sections imply a constant-round zero-knowledge quantum argument for QMA, that is, according to Definition \ref{def:qma_qzk} where honest parties are polynomial-time and quantum (prover is efficient given a quantum witness) and communication is quantum.
	
	The construction uses constant-round (post-quantum) zero-knowledge arguments for NP, quantumly-extractable commitments and a quantum sigma protocol for QMA\footnote{In a previous version of this work we used the QMA zero-knowledge (with large soundness error) protocol of \cite{broadbent2016zero} instead of sigma protocols. Using sigma protocols yields a simplified protocol.}.
	
	\noindent We now proceed to the construction and proof.
	
	\paragraph{Ingredients and notation:}
	\begin{itemize}
		\item
		A constant-round quantumly-extractable commitment scheme $(\comS, \comR)$.
		
		\item
		A constant-round post-quantum zero-knowledge argument system $(\zkP_{\text{NP}}, \zkV_{\text{NP}})$ for NP.
		
		\item
		A quantum sigma protocol for QMA $(\qsigmaP, \qsigmaV)$.
	\end{itemize}
	
%
	
	\noindent We describe the protocol in \figref{fig:const_round_qma}.
	
	\protocol
	{\proref{fig:const_round_qma}}
	{A quantum constant-round zero-knowledge argument for $\lang \in \QMA$.}
	{fig:const_round_qma}
	{
		\begin{description}
			\item[Common Input:] An instance $\ins \in \lang\cap \zo^\secp$, for security parameter $\secp \in \Nat$.
			\item[$\zkP$'s private input:] Polynomially many identical witnesses for $x$: $w^{\otimes k(\secp)}$ s.t. $w \in \mathcal{R}_{\lang}(\ins)$.
		\end{description}
		
		\begin{enumerate}
			
			\item {\bf Verifier Extractable Commitment to Challenge:} \label{qma:verifier_commit}
			$\zkV$ computes $\beta \gets \qsigmaV$ and commits to it using the extractable commitment $(\comS, \comR)$.
			$\zkV$ executes $\comS(1^\secp, \beta)$ and $\zkP$ executes $\comR(1^\secp)$, and commitment transcript $T_{\comS}$ is generated.
			
			\item {\bf Prover Commitment:} \label{qma:prover_alpha}
			$\zkP$ computes $(\alpha, \tau) \gets \qsigmaP_1(\ins,w^{\otimes k(\secp)})$ and sends $\alpha$ to $\zkV$.
			
			\item {\bf Verifier Challenge and ZK Argument:}
			\begin{enumerate}
				\item \label{qma:verifier_beta} $\zkV$ sends $\beta$.
				
				\item \label{qma:verifier_zk} $\zkV$ proves in ZK (using the argument system $(\zkP_{\text{NP}}, \zkV_{\text{NP}})$) that the sent $\beta$ is the value inside the extractable commitment, that is, $\exists r \in \{0, 1\}^*$ such that $1 = \decom(T_\comS, \beta, r)$.
				If the argument was not convincing $\zkP$ terminates communication.
			\end{enumerate}
			
			\item {\bf Sigma Protocol Completion:} \label{qma:prover_gamma}
			If the proof by $\zkV$ was convincing then $\zkP$ computes $\gamma \gets \qsigmaP_3(\beta, \tau)$ and sends $\gamma$.
			
			\item {\bf Acceptance:} The verifier accepts iff $1 = \qsigmaV(\alpha, \beta, \gamma)$.
			
		\end{enumerate}
	}

	\subsection{Computational Soundness}
	We prove that Protocol \proref{fig:const_round_qma} has quantum computational soundness.
	
	\begin{proposition}
		For any quantum polynomial-size prover $\zkmP = \set{\zkmP_\secp, \rho_\secp}_{\secp \in \Nat}$, there exists a negligible function $\mu(\cdot)$ such that for any security parameter $\secp\in \Nat$ and any $\ins \in \zo^\secp\setminus\lang$,
		\begin{align*}
			\Pr\left[ \view_{\zkV}\prot{\zkmP_\secp(\rho_\secp)}{\zkV}(\ins) = 1 \right] \leq \mu(\secp)\enspace.
		\end{align*}
	\end{proposition}
	
	\begin{proof}
		Let $\zkmP = \{ \zkmP_\secp, \rho_\secp \}_{\secp \in \Nat}$ a polynomial-size quantum prover and let $x = \{ x_\secp \}_{\secp \in \Nat}$ be a sequence such that $\forall \secp \in \Nat : x_\secp \in \{ 0, 1 \}^\secp \setminus \lang$.
		We prove soundness by a hybrid argument. We consider a series of hybrid processes with output over $\{ 0, 1 \}$, starting from $\view_{\zkV}\prot{\zkmP(\rho)}{\zkV}(\ins)$ the output distribution of $\zkV$ in the interaction with $\zkmP$.
		
		\begin{itemize}
			\item $\Hyb_0:$ The output distribution of $\view_{\zkV}\prot{\zkmP_\secp(\rho_\secp)}{\zkV}(\ins_\secp)$.
			
			\item $\Hyb_1:$ Identical to the process $\Hyb_0$, with the exception that in step \ref{qma:verifier_zk} when the verifier gives a ZK argument, the process instead uses the ZK simulator $\zkS$ of the argument system $(\zkP_{\text{NP}}, \zkV_{\text{NP}})$. To simulate the prover's view, the process executes $\zkS\left( (T_{\comS}, \beta), \zkmP, \rho' \right)$, where $\rho'$ is the inner quantum state of $\zkmP$ at the end of step \ref{qma:verifier_beta} where the verifier sends $\beta$.
			
			\item $\Hyb_2:$ Identical to the process $\Hyb_1$, with the exception that in step \ref{qma:verifier_commit} when the verifier commits to $\beta$, the process instead commits to $0^{|\beta|}$.
		\end{itemize}
	
		We next explain why each consecutive pair of distributions are indistinguishable.
		
		\begin{itemize}
			\item $\Hyb_0 \approx_c \Hyb_1:$ Follows from the quantum zero knowledge property of the protocol $(\zkP_{\text{NP}}, \zkV_{\text{NP}})$.
			
			\item $\Hyb_1 \approx_c \Hyb_2:$ Follows from the computational hiding of the commitment scheme $(\comS, \comR)$.
		\end{itemize}
	
		Now, assume toward contradiction that $\zkmP$ succeeds in making the verifier accept with some noticeable probability $\varepsilon(\secp)$, that is, the probability for the output $1$ in $\Hyb_0$ is noticeable.
		$\Hyb_0 \approx_c \Hyb_2$, and thus the probability for the output $1$ in $\Hyb_2$ is also noticeable.
		Finally, we get a contradiction to the soundness of the sigma protocol $(\qsigmaP, \qsigmaV)$, by using the prover sigma protocol messages from steps \ref{qma:prover_alpha}, \ref{qma:prover_gamma} as messages to convince a quantum sigma protocol verifier $\qsigmaV$. Since the probability that the verifier $\zkV$ is convinced in $\Hyb_2$ is noticeable, and such verifier is convinced if and only if the sigma protocol verifier is convinced, we get our contradiction.
	\end{proof}

	\subsection{Computational Zero Knowledge}
	We prove that Protocol \proref{fig:const_round_qma} is quantum computational zero knowledge.
	
	We describe a universal simulator $\zkS$ for the protocol.
	We denote by $\zkmV = \{ \zkmV_\secp, \rho_\secp \}_{\secp \in \Nat}$ a polynomial-size quantum verifier.
	The simulator takes as input an instance in the language $x \in \{ 0, 1 \}^{\secp} \cap \lang$, a verifier circuit $\zkmV_\secp$ and quantum auxiliary input $\rho_\secp$ for $\zkmV_\secp$.
	Subscripts are dropped when are clear from the context.
	
	\paragraph{$\zkS(x, \zkmV, \rho)$:}
		\begin{enumerate}
			\item {\bf Extraction of Message from Verifirer:} \label{qma_simulation:extraction}
			$\zkS$ executes the extractor $\Ext$ of the extractable commitment scheme $(\comS, \comR)$. $\zkS$ computes a simulation of the commitment interaction transcript, inner state at the end of interaction and extracted message $(T_{\Ext}, \sigma_{\Ext}, \beta_{\Ext}) \gets \Ext(1^\secp, \zkmV, \rho)$ and uses the simulated state $\sigma_{\Ext}$ as inner state for $\zkmV$ in order to continue the protocol simulation\footnote{By the standard definition, the extractor $\Ext$ simulates only the state and extracted message $(\sigma_{\Ext}, m_{\Ext})$, but recall we can assume without the loss of generality that it also simulates the commitment transcript $T_\Ext$ (see Remark \ref{remark:extractable_commitment}), and the triplet is indistinguishable from $(T, \sigma, m_T) \gets \prot{\commS(\rho)}{\comR}(1^\secp)$.}.
			
			\item {\bf Sigma Protocol First Part Simulation:} \label{qma_simulation:prover_alpha}
			$\zkS$ executes $(\alpha_{\zkS}, \gamma_{\zkS}) \gets \qsigmaS(x, \beta_{\Ext})$ and sends $\alpha_{\zkS}$.
			
			\item {\bf Malicious Verifier Challenge and ZK Argument:} \label{qma_simulation:verifier_beta}
			$\zkS$ takes the role of the honest prover $\zkP$ when the verifier sends $\beta$ and gives a ZK argument that $\exists r \in \{ 0, 1 \}^* : 1 = \decom(T_\Ext, \beta, r)$.
			If the argument was not convincing the simulator halts and concludes simulation.
			
			\item {\bf Sigma Protocol Second Part Simulation:}
			\label{qma_simulation:prover_gamma}
			$\zkS$ sends $\gamma_{\zkS}$ and concludes simulation.
			
		\end{enumerate}
	
	\medskip\noindent
	It remains to prove that the simulator's output is computationally indistinguishable from the verifier's output in the real interaction.
	
	\begin{proposition}
		For any polynomial-size quantum verifier $\zkmV = \set{\zkmV_\secp, \rho_\secp}_{\secp \in \Nat}$,
		$$
		\{ \view_{\zkmV_\secp}\prot{\zkP(w^{\otimes k(\secp)})}{\zkmV_\secp(\rho_{\secp})}(x)\}_{\secp, x, w}
		\approx_{c}
		\{\zkS(x,\zkmV_\secp, \rho_\secp)\}_{\secp, x, w}\enspace,
		$$
		where $\secp \in \Nat$, $x \in \lang \cap \{ 0, 1 \}^\secp$, $w \in \rel_{\lang}(x)$.
	\end{proposition} 
	
	\begin{proof}
		We prove the claim by a hybrid argument, specifically, we consider hybrid distributions, all of which will be computationally indistinguishable.
		
		\begin{itemize}
			\item $\Hyb_0:$ The output distribution of the simulator $\zkS(x,\zkmV, \rho)$.
			
			\item $\Hyb_1:$ Identical to the process $\Hyb_0$, except that we erase some extreme cases from the output distribution, by making a check.
			Specifically, in step \ref{qma_simulation:verifier_beta} when the verifier sends $\beta$ and a ZK argument, if $\beta_{\Ext} \neq \beta$ and also the argument by $\zkmV$ was convincing, the output of the process is $\bot$.
			
			\item $\Hyb_2:$ Identical to the process $\Hyb_1$, except that in steps \ref{qma_simulation:prover_alpha}, \ref{qma_simulation:prover_gamma} where the simulator sends $\alpha_{\zkS}$ and $\gamma_{\zkS}$, the process instead uses the real sigma protocol prover to generate the messages, $(\alpha, \tau) \gets \qsigmaP_1(\ins,w^{\otimes k(\secp)})$, $\gamma \gets \qsigmaP_3(\beta_{\Ext}, \tau)$.
			
			\item $\Hyb_3:$ Identical to the process $\Hyb_2$, except that when computing the the last sigma protocol message $\gamma \gets \qsigmaP_3(\beta_{\Ext}, \tau)$, the process uses the $\beta$ that $\zkmV$ sent instead of the extracted $\beta_\Ext$, that is, $\gamma \gets \qsigmaP_3(\beta, \tau)$.
			
			\item $\Hyb_4:$ Identical to the process $\Hyb_3$, except that the check described in $\Hyb_1$ is not performed, that is, even if the extracted challenge $\beta_\Ext$ and the challenge $\beta$ sent by $\zkmV$ are distinct and the ZK argument by $\zkmV$ succeeds, the process carries on to the last step \ref{qma_simulation:prover_gamma} and does not outputs $\bot$.
			
			\item $\Hyb_5:$ At this point in our series of hybrid distributions we do not use the extracted challenge $\beta_\Ext$, and we would like to move to a final process that does not use extraction at all.
			This process is identical to $\Hyb_4$, with the exception that in step \ref{qma_simulation:extraction} of the simulation, where the simulator executes $\Ext$ to simulate the transcript and inner state of $\zkmV$, the process simply executes the real interaction between $\zkmV$ and $\comR$, $(T, \sigma) \gets \prot{\zkmV(\rho)}{\comR}(1^\secp)$.
			Observe that $\Hyb_5$ is exactly the real interaction output $\view_{\zkmV}\prot{\zkP(w^{\otimes k})}{\zkmV(\rho)}(x)$.
		\end{itemize}
		Before proving that each consecutive pair of hybrids is indistnguishable, 
		
		We prove why each consecutive pair of distributions are computationally indistinguishable, and our proof is finished.
		
		\begin{itemize}
			\item $\Hyb_0 \approx_s \Hyb_1:$
			To show the indistinguishability we need to prove that the probabilistic event that exists in $\Hyb_0$ but is erased in $\Hyb_1$ happens with a negligible probability.
			This is exactly the statement proven in Claim \ref{claim:extracted_info_correct}.
			
			\item $\Hyb_1 \approx_c \Hyb_2:$
			This indistinguishability follows from the special zero knowledge property of the quantum sigma protocol.
			
			\item $\Hyb_2 \equiv \Hyb_3:$
			Due to the fact that in both hybrid processes, whenever $\beta_{\Ext} \neq \beta$ the process halts and outputs $\bot$, it is always the case that the first prover sigma protocol message $\gamma$ is computed with respect to the sent $\beta$.
			
			\item $\Hyb_3 \approx_s \Hyb_4:$
			The reasoning for this indistinguishability is identical to the reasoning for the indistinguishability $\Hyb_0 \approx_s \Hyb_1$, and follows from Claim \ref{claim:extracted_info_correct}.
			
			\item $\Hyb_4 \approx_c \Hyb_5:$
			This indistinguishability follows from the extractability property (in Definition \ref{def:q_ex_commit}) of the commitment scheme $(\comS, \comR)$.
		\end{itemize}
		
	\end{proof}
	
	\begin{claim} [Extracted Information is Correct Under an Argument] \label{claim:extracted_info_correct}
		Let $\zkmV = \set{\zkmV_\secp, \rho_\secp}_{\secp \in \Nat}$ a polynomial-size quantum verifier.
		Consider the process of interaction between $\zkmV(\rho)$ and $\zkP$ in the original protocol, with one change: when $\zkmV$ gives an extractable commitment, instead of executing the interaction $(T, \sigma) \gets \prot{\zkmV(\rho)}{\comR}(1^\secp)$, the process executes the extractor $(T_{\Ext}, \sigma_{\Ext}, \beta_{\Ext}) \gets \Ext(1^\secp, \zkmV, \rho)$.
		Then, there is some negligible function $\negl$ such that,
		$$
		\Pr\left[ \left( \beta \neq \beta_{\Ext} \right) \land \left( \text{$\zkmV$ gives a convincing argument} \right) \right] \leq \negl(\secp) \enspace .
		$$
	\end{claim}
	
	\begin{proof}
		Let $T_{\comS}$ be the transcript generated at the end of the extractable commitment protocol, in the original interaction between $\zkmV$ and $\zkP$.
		By the perfect binding of the commitment scheme $(\comS, \comR)$ it follows that if the statement from $\zkmV$'s ZK argument is correct, that is, there is some $r \in \{ 0, 1 \}^*$ s.t. $1 = \decom(T_{\comS}, \beta, r)$, then $\beta$ is necessarily the committed message, in symbols (denoted in the binding property in Definition \ref{def:q_ex_commit}) $\beta = m_{T_\comS}$.
		It follows from the soundness of the argument that $\zkmV$ gives, that only with a negligible probability $\negl'(\secp)$ it happens that both, $\beta \neq m_{T_\comS}$, and $\zkmV$ gives a convincing argument.
		
		Recall that by the extractability property of the commitment scheme (Extractability property in Definition \ref{def:q_ex_commit}), the following two distributions are indistinguishable,
		$$
		\Bigl\{
		(T, \sigma, m_T) \; | \; (T, \sigma, m_T) \gets \prot{\commS_\secp(\rho_\secp)}{\comR}(1^\secp)
		\Bigr\}_{\secp \in \Nat}
		$$
		$$
		\approx_{c}
		\Bigl\{ (T_\Ext, \sigma_\Ext, m_\Ext) \; | \; (\sigma_\Ext, m_\Ext) \gets \Ext(1^\secp, \commS_\secp, \rho_\secp) \Bigr\}_{\secp \in \Nat} \enspace .
		$$
		This means that when considering the process described in this claim's statement, where extraction takes place (instead of executing the commitment procedure), only with some negligible probability $\negl(\secp)$ it can happen that both, $\beta \neq \beta_\Ext$, and $\zkmV$ gives a convincing argument, this is because if this probability wasn't negligible we would be able to break the extractability property of the commitment scheme $(\comS, \comR)$. 
	\end{proof}